\documentclass[a4paper]{article}

\usepackage{amsmath}

\usepackage{amssymb}
\usepackage{amsthm}
\usepackage{graphicx}
\usepackage{ascmac}
\usepackage{fancybox}
\usepackage[stamp]{draftwatermark}
\SetWatermarkScale{10}
\SetWatermarkLightness{0.95}

\usepackage[all]{xy}
\usepackage[bibliography=common]{apxproof}
\usepackage{mathpartir}
\usepackage{proof}
\usepackage{xcolor}
\usepackage{comment}
\usepackage{url}
  \usepackage{geometry}

\newcommand{\Bf}[1]{{\bf #1}}

\newcommand{\Rm}[1]{{\rm #1}}
\newcommand{\ol}{\overline}
\newcommand{\ve}{\varepsilon}

\newcommand{\CSEPMet}{\Bf{CSEPMet}}
\newcommand{\QET}{\Bf{QET}}

\newcommand{\sep}{~\middle|~}

\newcommand{\fa}[1]{\forall{#1}~.~}
\newcommand{\ex}[1]{\exists{#1}~.~}

\newcommand{\lam}[1]{\lambda{#1}~.~}

\newcommand{\Typ}{\Bf{Typ}}

\newcommand{\lsem}{[\![}
\newcommand{\rsem}{]\!]}
\newcommand{\sem}[1]{\lsem #1\rsem}

\newcommand{\dom}{\Rm{dom}}

\newcommand{\Obj}[1]{\Bf{Obj}(#1)}

\newcommand{\Set}{\Bf{Set}}

\newcommand{\CAT}{\Bf{CAT}}

\newcommand{\arrow}{\rightarrow}

\newcommand{\Arrow}{\Rightarrow}
\newcommand{\bul}{\bullet}

\newcommand{\id}{\Rm{id}}

\newcommand{\unit}{\Bf I}
\newcommand{\ox}{\otimes}

\newcommand{\CC}{\mathbb C}
\newcommand{\DD}{\mathbb D}
\newcommand{\EE}{\mathbb E}

\newcommand{\NN}{\mathbb N}

\newcommand{\QQ}{\mathbb Q}
\newcommand{\RR}{\mathbb R}
 
\newcommand{\TT}{\mathbb T}

\newcommand{\VV}{\mathbb V}

\newcommand{\darrow}{\mathbin{\dot\arrow}}
\newcommand{\dto}{\mathbin{\dot\arrow}}

\newcommand{\dtimes}{\mathbin{\dot\times}}

\newcommand{\dArrow}{\mathbin{\dot\Arrow}}

\newcommand{\rue}{\ar@{=}[u]}
\newcommand{\rueh}[1]{\ar@{=}[u]^-{#1}}
\newcommand{\ruem}[1]{\ar@{=}[u]_-{#1}}
\newcommand{\ruen}[1]{\ar@{=}[u]|-{#1}}

\newcommand{\ruue}{\ar@{=}[uu]}
\newcommand{\ruueh}[1]{\ar@{=}[uu]^-{#1}}
\newcommand{\ruuem}[1]{\ar@{=}[uu]_-{#1}}
\newcommand{\ruuen}[1]{\ar@{=}[uu]|-{#1}}

\newcommand{\ruuue}{\ar@{=}[uuu]}
\newcommand{\ruuueh}[1]{\ar@{=}[uuu]^-{#1}}
\newcommand{\ruuuem}[1]{\ar@{=}[uuu]_-{#1}}
\newcommand{\ruuuen}[1]{\ar@{=}[uuu]|-{#1}}

\newcommand{\rdh}[1]{\ar[d]^-{#1}}
\newcommand{\rdm}[1]{\ar[d]_-{#1}}

\newcommand{\rde}{\ar@{=}[d]}
\newcommand{\rdeh}[1]{\ar@{=}[d]^-{#1}}
\newcommand{\rdem}[1]{\ar@{=}[d]_-{#1}}
\newcommand{\rden}[1]{\ar@{=}[d]|-{#1}}

\newcommand{\rdde}{\ar@{=}[dd]}
\newcommand{\rddeh}[1]{\ar@{=}[dd]^-{#1}}
\newcommand{\rddem}[1]{\ar@{=}[dd]_-{#1}}
\newcommand{\rdden}[1]{\ar@{=}[dd]|-{#1}}

\newcommand{\rddde}{\ar@{=}[ddd]}
\newcommand{\rdddeh}[1]{\ar@{=}[ddd]^-{#1}}
\newcommand{\rdddem}[1]{\ar@{=}[ddd]_-{#1}}
\newcommand{\rddden}[1]{\ar@{=}[ddd]|-{#1}}

\newcommand{\rrh}[1]{\ar[r]^-{#1}}
\newcommand{\rrm}[1]{\ar[r]_-{#1}}

\newcommand{\rre}{\ar@{=}[r]}
\newcommand{\rreh}[1]{\ar@{=}[r]^-{#1}}
\newcommand{\rrem}[1]{\ar@{=}[r]_-{#1}}
\newcommand{\rren}[1]{\ar@{=}[r]|-{#1}}

\newcommand{\rrue}{\ar@{=}[ru]}
\newcommand{\rrueh}[1]{\ar@{=}[ru]^-{#1}}
\newcommand{\rruem}[1]{\ar@{=}[ru]_-{#1}}
\newcommand{\rruen}[1]{\ar@{=}[ru]|-{#1}}

\newcommand{\rruue}{\ar@{=}[ruu]}
\newcommand{\rruueh}[1]{\ar@{=}[ruu]^-{#1}}
\newcommand{\rruuem}[1]{\ar@{=}[ruu]_-{#1}}
\newcommand{\rruuen}[1]{\ar@{=}[ruu]|-{#1}}

\newcommand{\rruuue}{\ar@{=}[ruuu]}
\newcommand{\rruuueh}[1]{\ar@{=}[ruuu]^-{#1}}
\newcommand{\rruuuem}[1]{\ar@{=}[ruuu]_-{#1}}
\newcommand{\rruuuen}[1]{\ar@{=}[ruuu]|-{#1}}

\newcommand{\rrdh}[1]{\ar[rd]^-{#1}}

\newcommand{\rrde}{\ar@{=}[rd]}
\newcommand{\rrdeh}[1]{\ar@{=}[rd]^-{#1}}
\newcommand{\rrdem}[1]{\ar@{=}[rd]_-{#1}}
\newcommand{\rrden}[1]{\ar@{=}[rd]|-{#1}}

\newcommand{\rrdde}{\ar@{=}[rdd]}
\newcommand{\rrddeh}[1]{\ar@{=}[rdd]^-{#1}}
\newcommand{\rrddem}[1]{\ar@{=}[rdd]_-{#1}}
\newcommand{\rrdden}[1]{\ar@{=}[rdd]|-{#1}}

\newcommand{\rrddde}{\ar@{=}[rddd]}
\newcommand{\rrdddeh}[1]{\ar@{=}[rddd]^-{#1}}
\newcommand{\rrdddem}[1]{\ar@{=}[rddd]_-{#1}}
\newcommand{\rrddden}[1]{\ar@{=}[rddd]|-{#1}}

\newcommand{\rrrm}[1]{\ar[rr]_-{#1}}

\newcommand{\rrre}{\ar@{=}[rr]}
\newcommand{\rrreh}[1]{\ar@{=}[rr]^-{#1}}
\newcommand{\rrrem}[1]{\ar@{=}[rr]_-{#1}}
\newcommand{\rrren}[1]{\ar@{=}[rr]|-{#1}}

\newcommand{\rrrue}{\ar@{=}[rru]}
\newcommand{\rrrueh}[1]{\ar@{=}[rru]^-{#1}}
\newcommand{\rrruem}[1]{\ar@{=}[rru]_-{#1}}
\newcommand{\rrruen}[1]{\ar@{=}[rru]|-{#1}}

\newcommand{\rrruue}{\ar@{=}[rruu]}
\newcommand{\rrruueh}[1]{\ar@{=}[rruu]^-{#1}}
\newcommand{\rrruuem}[1]{\ar@{=}[rruu]_-{#1}}
\newcommand{\rrruuen}[1]{\ar@{=}[rruu]|-{#1}}

\newcommand{\rrruuue}{\ar@{=}[rruuu]}
\newcommand{\rrruuueh}[1]{\ar@{=}[rruuu]^-{#1}}
\newcommand{\rrruuuem}[1]{\ar@{=}[rruuu]_-{#1}}
\newcommand{\rrruuuen}[1]{\ar@{=}[rruuu]|-{#1}}

\newcommand{\rrrde}{\ar@{=}[rrd]}
\newcommand{\rrrdeh}[1]{\ar@{=}[rrd]^-{#1}}
\newcommand{\rrrdem}[1]{\ar@{=}[rrd]_-{#1}}
\newcommand{\rrrden}[1]{\ar@{=}[rrd]|-{#1}}

\newcommand{\rrrdde}{\ar@{=}[rrdd]}
\newcommand{\rrrddeh}[1]{\ar@{=}[rrdd]^-{#1}}
\newcommand{\rrrddem}[1]{\ar@{=}[rrdd]_-{#1}}
\newcommand{\rrrdden}[1]{\ar@{=}[rrdd]|-{#1}}

\newcommand{\rrrddde}{\ar@{=}[rrddd]}
\newcommand{\rrrdddeh}[1]{\ar@{=}[rrddd]^-{#1}}
\newcommand{\rrrdddem}[1]{\ar@{=}[rrddd]_-{#1}}
\newcommand{\rrrddden}[1]{\ar@{=}[rrddd]|-{#1}}

\newcommand{\rrrre}{\ar@{=}[rrr]}
\newcommand{\rrrreh}[1]{\ar@{=}[rrr]^-{#1}}
\newcommand{\rrrrem}[1]{\ar@{=}[rrr]_-{#1}}
\newcommand{\rrrren}[1]{\ar@{=}[rrr]|-{#1}}

\newcommand{\rrrrue}{\ar@{=}[rrru]}
\newcommand{\rrrrueh}[1]{\ar@{=}[rrru]^-{#1}}
\newcommand{\rrrruem}[1]{\ar@{=}[rrru]_-{#1}}
\newcommand{\rrrruen}[1]{\ar@{=}[rrru]|-{#1}}

\newcommand{\rrrruue}{\ar@{=}[rrruu]}
\newcommand{\rrrruueh}[1]{\ar@{=}[rrruu]^-{#1}}
\newcommand{\rrrruuem}[1]{\ar@{=}[rrruu]_-{#1}}
\newcommand{\rrrruuen}[1]{\ar@{=}[rrruu]|-{#1}}

\newcommand{\rrrruuue}{\ar@{=}[rrruuu]}
\newcommand{\rrrruuueh}[1]{\ar@{=}[rrruuu]^-{#1}}
\newcommand{\rrrruuuem}[1]{\ar@{=}[rrruuu]_-{#1}}
\newcommand{\rrrruuuen}[1]{\ar@{=}[rrruuu]|-{#1}}

\newcommand{\rrrrde}{\ar@{=}[rrrd]}
\newcommand{\rrrrdeh}[1]{\ar@{=}[rrrd]^-{#1}}
\newcommand{\rrrrdem}[1]{\ar@{=}[rrrd]_-{#1}}
\newcommand{\rrrrden}[1]{\ar@{=}[rrrd]|-{#1}}

\newcommand{\rrrrdde}{\ar@{=}[rrrdd]}
\newcommand{\rrrrddeh}[1]{\ar@{=}[rrrdd]^-{#1}}
\newcommand{\rrrrddem}[1]{\ar@{=}[rrrdd]_-{#1}}
\newcommand{\rrrrdden}[1]{\ar@{=}[rrrdd]|-{#1}}

\newcommand{\rrrrddde}{\ar@{=}[rrrddd]}
\newcommand{\rrrrdddeh}[1]{\ar@{=}[rrrddd]^-{#1}}
\newcommand{\rrrrdddem}[1]{\ar@{=}[rrrddd]_-{#1}}
\newcommand{\rrrrddden}[1]{\ar@{=}[rrrddd]|-{#1}}

\newcommand{\rle}{\ar@{=}[l]}
\newcommand{\rleh}[1]{\ar@{=}[l]^-{#1}}
\newcommand{\rlem}[1]{\ar@{=}[l]_-{#1}}
\newcommand{\rlen}[1]{\ar@{=}[l]|-{#1}}

\newcommand{\rlue}{\ar@{=}[lu]}
\newcommand{\rlueh}[1]{\ar@{=}[lu]^-{#1}}
\newcommand{\rluem}[1]{\ar@{=}[lu]_-{#1}}
\newcommand{\rluen}[1]{\ar@{=}[lu]|-{#1}}

\newcommand{\rluue}{\ar@{=}[luu]}
\newcommand{\rluueh}[1]{\ar@{=}[luu]^-{#1}}
\newcommand{\rluuem}[1]{\ar@{=}[luu]_-{#1}}
\newcommand{\rluuen}[1]{\ar@{=}[luu]|-{#1}}

\newcommand{\rluuue}{\ar@{=}[luuu]}
\newcommand{\rluuueh}[1]{\ar@{=}[luuu]^-{#1}}
\newcommand{\rluuuem}[1]{\ar@{=}[luuu]_-{#1}}
\newcommand{\rluuuen}[1]{\ar@{=}[luuu]|-{#1}}

\newcommand{\rlde}{\ar@{=}[ld]}
\newcommand{\rldeh}[1]{\ar@{=}[ld]^-{#1}}
\newcommand{\rldem}[1]{\ar@{=}[ld]_-{#1}}
\newcommand{\rlden}[1]{\ar@{=}[ld]|-{#1}}

\newcommand{\rldde}{\ar@{=}[ldd]}
\newcommand{\rlddeh}[1]{\ar@{=}[ldd]^-{#1}}
\newcommand{\rlddem}[1]{\ar@{=}[ldd]_-{#1}}
\newcommand{\rldden}[1]{\ar@{=}[ldd]|-{#1}}

\newcommand{\rlddde}{\ar@{=}[lddd]}
\newcommand{\rldddeh}[1]{\ar@{=}[lddd]^-{#1}}
\newcommand{\rldddem}[1]{\ar@{=}[lddd]_-{#1}}
\newcommand{\rlddden}[1]{\ar@{=}[lddd]|-{#1}}

\newcommand{\rlle}{\ar@{=}[ll]}
\newcommand{\rlleh}[1]{\ar@{=}[ll]^-{#1}}
\newcommand{\rllem}[1]{\ar@{=}[ll]_-{#1}}
\newcommand{\rllen}[1]{\ar@{=}[ll]|-{#1}}

\newcommand{\rllue}{\ar@{=}[llu]}
\newcommand{\rllueh}[1]{\ar@{=}[llu]^-{#1}}
\newcommand{\rlluem}[1]{\ar@{=}[llu]_-{#1}}
\newcommand{\rlluen}[1]{\ar@{=}[llu]|-{#1}}

\newcommand{\rlluue}{\ar@{=}[lluu]}
\newcommand{\rlluueh}[1]{\ar@{=}[lluu]^-{#1}}
\newcommand{\rlluuem}[1]{\ar@{=}[lluu]_-{#1}}
\newcommand{\rlluuen}[1]{\ar@{=}[lluu]|-{#1}}

\newcommand{\rlluuue}{\ar@{=}[lluuu]}
\newcommand{\rlluuueh}[1]{\ar@{=}[lluuu]^-{#1}}
\newcommand{\rlluuuem}[1]{\ar@{=}[lluuu]_-{#1}}
\newcommand{\rlluuuen}[1]{\ar@{=}[lluuu]|-{#1}}

\newcommand{\rllde}{\ar@{=}[lld]}
\newcommand{\rlldeh}[1]{\ar@{=}[lld]^-{#1}}
\newcommand{\rlldem}[1]{\ar@{=}[lld]_-{#1}}
\newcommand{\rllden}[1]{\ar@{=}[lld]|-{#1}}

\newcommand{\rlldde}{\ar@{=}[lldd]}
\newcommand{\rllddeh}[1]{\ar@{=}[lldd]^-{#1}}
\newcommand{\rllddem}[1]{\ar@{=}[lldd]_-{#1}}
\newcommand{\rlldden}[1]{\ar@{=}[lldd]|-{#1}}

\newcommand{\rllddde}{\ar@{=}[llddd]}
\newcommand{\rlldddeh}[1]{\ar@{=}[llddd]^-{#1}}
\newcommand{\rlldddem}[1]{\ar@{=}[llddd]_-{#1}}
\newcommand{\rllddden}[1]{\ar@{=}[llddd]|-{#1}}

\newcommand{\rllle}{\ar@{=}[lll]}
\newcommand{\rllleh}[1]{\ar@{=}[lll]^-{#1}}
\newcommand{\rlllem}[1]{\ar@{=}[lll]_-{#1}}
\newcommand{\rlllen}[1]{\ar@{=}[lll]|-{#1}}

\newcommand{\rlllue}{\ar@{=}[lllu]}
\newcommand{\rlllueh}[1]{\ar@{=}[lllu]^-{#1}}
\newcommand{\rllluem}[1]{\ar@{=}[lllu]_-{#1}}
\newcommand{\rllluen}[1]{\ar@{=}[lllu]|-{#1}}

\newcommand{\rllluue}{\ar@{=}[llluu]}
\newcommand{\rllluueh}[1]{\ar@{=}[llluu]^-{#1}}
\newcommand{\rllluuem}[1]{\ar@{=}[llluu]_-{#1}}
\newcommand{\rllluuen}[1]{\ar@{=}[llluu]|-{#1}}

\newcommand{\rllluuue}{\ar@{=}[llluuu]}
\newcommand{\rllluuueh}[1]{\ar@{=}[llluuu]^-{#1}}
\newcommand{\rllluuuem}[1]{\ar@{=}[llluuu]_-{#1}}
\newcommand{\rllluuuen}[1]{\ar@{=}[llluuu]|-{#1}}

\newcommand{\rlllde}{\ar@{=}[llld]}
\newcommand{\rllldeh}[1]{\ar@{=}[llld]^-{#1}}
\newcommand{\rllldem}[1]{\ar@{=}[llld]_-{#1}}
\newcommand{\rlllden}[1]{\ar@{=}[llld]|-{#1}}

\newcommand{\rllldde}{\ar@{=}[llldd]}
\newcommand{\rlllddeh}[1]{\ar@{=}[llldd]^-{#1}}
\newcommand{\rlllddem}[1]{\ar@{=}[llldd]_-{#1}}
\newcommand{\rllldden}[1]{\ar@{=}[llldd]|-{#1}}

\newcommand{\rlllddde}{\ar@{=}[lllddd]}
\newcommand{\rllldddeh}[1]{\ar@{=}[lllddd]^-{#1}}
\newcommand{\rllldddem}[1]{\ar@{=}[lllddd]_-{#1}}
\newcommand{\rlllddden}[1]{\ar@{=}[lllddd]|-{#1}}

\newenvironment{choice}{\left\{\begin{array}{ll}}{\end{array}\right.}

\newcommand{\adjunction}[3]{
  \ar@<.4pc>[#1]^-{#2}
  \ar@{}[#1]|-*=0[@]{\bot}
  \ar@<-.4pc>@{<-}[#1]_-{#3}
}

\newcommand{\cornerUL}[1][dr]{\save*!/#1-1.2pc/#1:(-1,1)@^{|-}\restore}

\newcommand{\attention}[1]{\textcolor{red}{#1}}
\renewcommand{\attention}[1]{#1}

\newcommand{\BRel}[1]{\Bf{BRel}(#1)}
\newcommand{\brelc}{\BRel\CC}
\newcommand{\brelf}[1]{\brelfn{#1}:\BRel{#1}\to{#1}^2}
\newcommand{\brelfn}[1]{p_{#1}}

\newcommand{\nm}[1]{\lceil{#1}\rceil}
\newcommand{\unm}[1]{\lfloor{#1}\rfloor}

\newcommand{\Div}[2]{\Bf{Div}_{#1}({#2})}
\newcommand{\Divfn}[2]{V_{#1,#2}}
\newcommand{\Divf}[2]{\Divfn{#1}{#2}:\Div{#1}{#2}\to{#2}}

\newcommand{\Meas}{\Bf{Meas}}
\newcommand{\Span}[1]{\Bf{Span}(#1)}

\newcommand{\QBS}{\Bf{QBS}}
\newcommand{\ev}{\mathrm{ev}}

\newcommand{\mDiv}[4]{\Bf{Div}(#1,#4,#2,#3)}

\newcommand{\tmem}[1]{{\em #1\/}}
\newcommand{\tmop}[1]{\ensuremath{\operatorname{#1}}}

\newcommand{\coden}[2]{#1^{[#2]}}
\newcommand{\codeng}[3]{\attention{#1^{[#2],#3}}}
\newcommand{\ap}{\mathbin{\bul}}

\newcommand{\qRp}{\mathcal R^+}
\newcommand{\qRm}{\mathcal R^\times}
\newcommand{\qRg}[1]{\mathcal R^+_{#1}}
\newcommand{\qRa}{\mathcal R}
\newcommand{\qB}{\mathcal B}
\newcommand{\qN}{\mathcal N}
\newcommand{\qZ}{\mathcal Z}
\newcommand{\qQ}{\mathcal Q}
\newcommand{\qQs}{\mathcal Q_s}

\newcommand{\asgn}{{\mathsf\Delta}}
\newcommand{\divntv}{\mathsf{TV}} 
\newcommand{\divnkl}{\mathsf{KL}} 
\newcommand{\divnhd}{\mathsf{HD}} 
\newcommand{\divnchi}{\mathsf{Chi}} 
\newcommand{\divndp}{\mathsf{DP}} 
\newcommand{\divnpwdp}{\mathsf{pwDP}} 
\newcommand{\divnrn}{\mathsf{Re}} 
\newcommand{\divnnc}{\mathsf{NC}} 
\newcommand{\divnzcdp}{\mathsf{zCDP}} 
\newcommand{\divntcdp}{\mathsf{tCDP}} 
\newcommand{\divnf}{{}^{f}\mathsf{Div}} 
\newcommand{\divnOpfun}[2]{\langle{#2}\rangle^*#1} 

\newcommand{\divnCostDiv}[2]{{\mathsf C}(#1,#2)} 

\newcommand{\asg}[2]{\asgn^{#1}_{#2}} 
\newcommand{\divtv}[1]{\divntv_{#1}} 
\newcommand{\divdp}[2]{\divndp_{#2}^{#1}} 
\newcommand{\divpwdp}[2]{\divnpwdp_{#2}^{#1}} 
\newcommand{\divrn}[2]{{}^{#1}\divnrn_{#2}} 
\newcommand{\divnc}[1]{\divnnc_{#1}} 
\newcommand{\divtcdp}[2]{{}^{#1}\divntcdp_{#2}} 
\newcommand{\divf}[1]{{}^{f}\mathsf{Div}_{#1}} 
\newcommand{\asgnpc}{\asgn_{c}}
\newcommand{\divOpfun}[4]{(\divnOpfun {#1} {#2})^{#3}_{#4}} 

\newcommand{\divCostDiv}[3]{\divnCostDiv{#1}{#2}_{#3}} 

\newcommand{\pwr}{P}
\newcommand{\dist}{D}
\newcommand{\giry}{G}
\newcommand{\sgiry}{G_{s}}

\newcommand{\apRHL}{{\sf apRHL}}

\newcommand{\mF}{\textcolor{black}{F}}

\newcommand{\mEQ}{\textcolor{black}{\mathrm{Eq}}}
\newcommand{\mT}{\textcolor{black}{\mathrm{Top}}}

\newcommand{\extplus}{\mathbin{\bar{+}}}

\newcommand{\adjL}{L}
\newcommand{\adjR}{K}
\newcommand{\probqbs}{P}

\newcommand{\fm}[3]{\{#3\}_{#2\in #1}}

\newcommand{\mcite}[1]{(\cite{#1})}
\newcommand{\kl}{^\sharp}

\newcommand{\aquick}{\mathsf{qsort}}
\newcommand{\ainsert}{\mathsf{isort}}

\allowdisplaybreaks[4]

\newif\ifdebugappendix
\debugappendixfalse 
\ifdebugappendix

\renewenvironment{toappendix}{**********BEGIN APPENDIX**********}{**********END APPENDIX********** \\}
\renewenvironment{appendixproof}{\begin{proof}}{\end{proof}}
\else
\fi

\theoremstyle{plain}
\newtheorem{therm}{Theorem}
\newtheorem{proposition}{Proposition}
\newtheorem{lemma}{Lemma}
\newtheorem{corollary}{Corollary}
\theoremstyle{definition}
\newtheorem{definition}{Definition}

\title{Divergences on Monads for Relational Program Logics}

\author{Tetsuya Sato \and Shin{-}ya Katsumata}

\begin{document}






\maketitle

\begin{abstract}
  Several relational program logics have been introduced for
  integrating reasoning about relational properties of programs and
  measurement of quantitative difference between computational
  effects.  Towards a general framework for such logics, in this
  paper, we formalize quantitative difference between computational
  effects as {\em divergence on monad}, then develop a relational
  program logic acRL that supports generic computational effects and
  divergences on them.  To give a categorical semantics of acRL
  supporting divergences, we give a method to obtain {\em graded
    strong relational liftings} from divergences on monads.  We derive
  two instantiations of acRL for the verification of 1) various differential
  privacy of higher-order functional
  probabilistic programs and 2) difference of distribution of costs
  between higher-order functional programs with
  probabilistic choice and cost counting operations.
\end{abstract}


\section{Introduction}

Comparing behavior of programs is one of the fundamental activities in
the verification and analysis of programs, and many concepts,
techniques and formal systems have been proposed for this purpose,
such as {\em product program construction} \mcite{BGBHS2017POPL}, {\em
  relational Hoare logic} \mcite{BentonRelationalHoare2004}, {\em
  higher-order relational refinement types}
\mcite{DBLP:conf/popl/BartheGAHRS15} and so on.

Several recent relational program logics integrate compositional
reasoning about relational properties of programs and
over-approximation of {\em quantitative difference} between
computational effects of programs; the latter is done in the style of
{\em effect system} \mcite{DBLP:conf/popl/LucassenG88}.  One
successful logic of this kind is Barthe et al.'s {\em approximate
  probabilistic relational Hoare logic} (\apRHL{} for short) designed
for verifying \emph{differential privacy} of probabilistic programs
\mcite{DBLP:conf/popl/BartheKOB12}. A judgment of \apRHL{} is of the
form $c\sim_{\epsilon,\delta}c':\Phi\Arrow\Psi$, and its intuitive
meaning is that for any state pair $(\rho,\rho')$ related by $\Phi$,
the $\epsilon$-distance between two probability distributions of final
states $\sem c\rho$ and $\sem{c'}{\rho'}$ is below $\delta$, and final
states satisfy $\Psi$.  Another relational program logic that measures
the difference between computational effects of programs is
\c{C}i\c{c}ek et al.'s {\sf RelCost}
\mcite{Cicek:2017:RCA:3093333.3009858}.  The target of the reasoning
is a higher-order programming language equipped with cost counting
effect. When we derive a judgment
$\Delta; \Psi; \Gamma \vdash M_1 \ominus M_2 \precsim n \colon \Phi$
in {\sf RelCost}, the sound semantics ensures that the difference of
cost counts by $M_1$ and $M_2$ is bound by $n$.

A high-level view on these relational program logics is that they
integrate the feature of measuring quantitative difference between
computational effects into relational program logic. We are interested
in extracting mathematical essence of this design and making
relational program logics versatile. Towards this goal, we contribute the
following development.
\begin{itemize}
\item We introduce a structure called {\em divergence on monad} for
  measuring quantitative difference between computational effects
  (Section \ref{sec:divergence_of_monads},
  \ref{sec:divergence:examples}). This generalizes various statistical
  divergences, such as Kullback-Leibler divergence and total variation
  distance on probability distributions. After exploring examples of
  divergence on monads, we introduce a method to transfer divergences
  on a monad to those on another monad through monad opfunctors.
  
\item The key structure to integrate divergences on monads and
  relational program logics is something called {\em graded strong
    relational lifting} of monads that extends given divergences. We
  present a general construction of such liftings from divergences on
  monads in Section \ref{sec:codensity_lifting}. This generalization
  shows that the development of relational program logics with
  quantitative measurement on computational effects can be done with
  various combinations of monads and divergences on them.
  
\item We introduce a generic relational program logic (called acRL)
  over Moggi's computational metalanguage (the simply-typed lambda
  calculus with monadic types) in Section \ref{sec:acRL}. Inside acRL,
  we can use graded strong relational liftings constructed from
  divergences on a monad, and reason about relational properties of
  programs together with quantitative difference of computational
  effects. To illustrate how the reasoning works in acRL, we
  instantiate it with the computational metalanguage having effectful
  operations for continuous random sampling (Section \ref{ex:prob})
  and cost counting operation (Section \ref{sec:probcost}).
\end{itemize}

\section{Preliminaries}
\label{sec:setting}

We assume basic knowledge about category theory \mcite{cwm2} and
Moggi's model of computational effects \mcite{moggicomputational}.
The definition of monad \cite[Chapter VI]{cwm2} and Kleisli category
\cite[Section VI.5]{cwm2} are omitted. 

In this paper, a Cartesian category (CC for short) is specified by a
category $\CC$ with a designated final object $1$ and a binary product
functor $(\times)\colon \CC^2\to\CC$. The associated pairing operation
and projection morphisms are denoted by
$\langle-,-\rangle,\pi_1,\pi_2$, respectively. The unique morphism to the
terminal object is denoted by $!_I:I\to 1$. A Cartesian closed category (CCC for
short) is a CC $(\CC,1,(\times))$ with a specified exponential functor
$(\Arrow)\colon \CC^{\mathrm{op}}\times\CC\to\CC$. The associated
evaluation morphism and currying operation is denoted by
$\ev,\lambda(-)$ respectively.

Let $(\CC,1,(\times))$ be a CC. A {\em global element} of $I\in\CC$ is
a morphism of type $1\arrow I$. For a category $\CC$, we define the
functor $U^\CC:\CC\to\Set$ by $U^\CC=\CC(1,-)$. When $\CC$ is obvious,
$U^\CC$ is denoted by $\CC$.
Morphisms in $\CC$ act on global elements
by the composition. To emphasize this action, we introduce a dedicated
notation $(\ap)$ whose type is $\CC(I,J)\times UI\to UJ$. Of course,
$f\ap x\triangleq f\circ x=(Uf)(x)$. We also define the partial
application of a binary morphism $f \colon I \times J \rightarrow K$ to a
global element $i\in UI$ by
$f_i \triangleq f \circ \langle i \circ !_J, \tmop{id}_J \rangle:J\to
K$.
When $\CC$ is a CCC, there is an evident isomorphism $\unm{-}\colon U(I\Arrow J)\cong\CC(I,J)$.
We write $\nm{-}$ for its inverse.

A monad $(T,\eta,\mu)$ on a category $\CC$ determines the operation
$(-)\kl:\CC(I,TJ)\to\CC(TI,TJ)$ called {\em Kleisli extension}. It is
defined by $f\kl\triangleq\mu_J\circ Tf$. A monad may be given as a
{\em Kleisli triple} \cite[Definition 1.2]{moggicomputational}. A
{\em strong monad} on a CC $(\CC,1,(\times))$ is a pair of a monad
$(T,\eta,\mu)$ and a natural transformation
$\theta_{I,J}\colon I\times TJ\to T(I\times J)$ called {\em strength}. It
should satisfy four axioms; see \cite[Definition
3.2]{moggicomputational} for detail.

In a CC-SM $(\CC,1,(\times),T,\eta,\mu,\theta)$, the
application of the strength to a global element can be expressed by
the unit and the Kleisli extension of $T$ \cite[Proof of Proposition 3.4]{moggicomputational}:
\begin{equation}
 \label{eq:stun}
  \theta_{I, J} \ap\langle i,c\rangle
 = ((\eta_{I\times J})_i)^{\sharp}\ap c \quad (i\in UI,c\in U(TJ)).
\end{equation}
We will use this fact in Proposition
\ref{prop:divergence:cost:projection} and Proposition
\ref{pp:grastrrellift}.

There are plenty of examples of C(C)Cs. For the models of
probabilistic computation, we will later use CC $\Meas$ of
measurable spaces and CCC $\QBS$ of quasi-Borel spaces
\mcite{HeunenKSY17}. Their definitions are deferred to Section
\ref{sec:qbs}.

\subsection{Category of Binary Relations}\label{sec:brel}

We next introduce the category $\brelc$ of binary relations over
$\CC$-objects. This category is equivalent to {\em subscones} of
$\CC^2$ \mcite{DBLP:conf/csl/MitchellS92}. It offers an underlying
category for relational reasoning about programs interpreted in $\CC$.
\begin{itemize}
\item An object in $\brelc$ is a triple $(I_1,I_2,R)$ where
 $R\subseteq UI\times UJ$.
\item A morphism from $(I_1,I_2,R)$ to $(J_1,J_2,S)$ in $\brelc$ is a pair of
  $\CC$-morphisms $f_1\colon I_1\arrow J_1$ and $f_2\colon I_2\arrow J_2$ such that for any
 $(i_1,i_2)\in R$, we have $(f_1\ap i_1,f_2\ap i_2)\in S$.
\end{itemize}
When $X$ is a name of a $\brelc$-object, by $X_1,X_2$ we mean its
first and second component, and by $R_X$ we mean its third component; so
$X=(X_1,X_2,R_X)$. By $(x_1,x_2)\in X$ we mean $(x_1,x_2)\in R_X$.
For objects $X,Y\in\brelc$ and a morphism
$(f_1,f_2)\colon (X_1,X_2)\to (Y_1,Y_2)$ in $\CC^2$, by
\begin{displaymath}
 (f_1,f_2)\colon X\darrow Y
\end{displaymath}
we mean that $(f_1,f_2)\in\brelc(X,Y)$, that is, for any
$(x_1,x_2)\in X$, we have $(f_1\ap x_1,f_2\ap x_2)\in Y$. We say that
$X\in\brelc$ is an {\em endorelation} (over $I$) if $X_1=X_2(=I)$.

We next define the forgetful functor $\brelf\CC$ by
\begin{displaymath}
 \brelfn\CC X\triangleq(X_1,X_2),\quad \brelfn\CC(f_1,f_2)\triangleq(f_1,f_2).
\end{displaymath}
For $(I_1,I_2)\in\CC^2$, by $\brelc_{(I_1,I_2)}$
we mean the complete boolean
algebra $\{X\in\brelc~|~X_1=I_1\wedge X_2=I_2\}$ with the order given by
$X\le Y\iff R_X\subseteq R_Y$.

When $\CC$ is a C(C)C, so is $\brelc$
\cite[Proposition 4.3]{DBLP:conf/csl/MitchellS92}. We specify a final object, a binary
product functor and an exponential functor (in case $\CC$ is a CCC) on
$\brelc$ by:
\begin{align*}
  \dot 1
 & \triangleq
    (1,1,\{(\id_1,\id_1)\}) \\
 X\dtimes Y
  & \triangleq
  (X_1\times Y_1, X_2\times Y_2,
    \{(\langle x_1,y_1\rangle,\langle x_2,y_2\rangle)~|~(x_1,x_2)\in X,(y_1,y_2)\in Y \})\\
 X\dArrow Y
  & \triangleq
  (X_1\Arrow Y_1,X_2\Arrow Y_2,
    \{(f_1,f_2)~|~\fa{(x_1,x_2)\in X}(\ev\circ\langle f_1,x_1\rangle,\ev\circ\langle f_2,x_2\rangle)\in Y\}
  ).
\end{align*}

\section{Divergences on Objects}

We introduce the concept of {\em divergence} on objects in a CC
$\CC$. Major differences between divergence and metric are threefold:
1) it is defined over objects in $\CC$, 2) no axioms is imposed on it,
and 3) it takes values in a partially ordered monoid called {\em
 divergence domain}, which we define below.
\begin{definition}
  A {\em divergence domain} $\qQ=(Q,\le,0,(+))$ is a partially ordered
 commutative monoid whose poset part is a complete lattice.
\end{definition}
The monoid addition $(+)$ is only required to be monotone; no
interaction with the sup / inf is required. We reserve the letter
$\qQ$ to denote a general divergence domain. Examples of divergence
domains are:
\begin{align*}
  \qN&=(\mathbb{N}\cup\{\infty\},\leq,0,(+)),&
 \qRp&=([0,\infty],\le,0,(+)),\\
  \qRm&=([0,\infty],\leq,1,(\times)),&
 \qRg 1&=([0,\infty],\le,0,\lam{(p,q)}p+q+pq),\\
  \qZ&=(\mathbb{Z}\cup\{\infty,-\infty\},\leq,0,(\extplus)),&
 \qRa&=([-\infty,\infty],\le,0,(\extplus))
\end{align*}
Here, $\extplus$ is an extension of the addition by
$r \extplus (-\infty) = (-\infty) \extplus r=-\infty$.
\begin{definition}
  Let $\CC$ be a CC.  A {\em $\qQ$-divergence} on an object
 $I\in\CC$ is a function $d\colon (UI)^2\to \qQ$.
\end{definition}
A suitable notion of morphism between $\CC$-objects with divergences
is {\em nonexpansive morphism}.
\begin{definition}
  Let $\CC$ be a CC.  We define the category $\Div\qQ\CC$ of
  $\qQ$-divergences on $\CC$-objects and nonexpansive morphisms
  between them by the following data.
  \begin{itemize}
  \item An object is a pair $(I,d)$ of an object $I\in\CC$ and a
    $\qQ$-divergence $d$ on $I$.
  \item A morphism from $(I,d)$ to $(J,e)$ is a $\CC$-morphism
    $f\colon I\to J$ such that for any $x_1,x_2\in UI$,
    $e(f\ap x_1,f\ap x_2)\le d(x_1,x_2)$ holds.
  \end{itemize}
  For an object $X\in\Div\qQ\CC$, by $d_X$ we mean its $\qQ$-divergence part.
  We also define the forgetful functor $\Divf \qQ\CC$ by
  $\Divfn \qQ\CC(I,d)\triangleq I$ and $\Divfn \qQ\CC(f)\triangleq f$.
\end{definition}
We remark that the
forgetful functor $\Divf\qQ\Set$ is a (Grothendieck) {\em fibration},
and the functor $\ol U\colon \Div\qQ\CC\to\Div\qQ\Set$ defined by
$\ol U(I,d)\triangleq (UI,d)$ and $\ol U(f)\triangleq f$ makes the
following commutative square a pullback in $\CAT$ (the large category
of categories and functors between them):
\begin{displaymath}
  \xymatrix@C=2cm{
  \Div\qQ\CC \rdm{\Divfn\qQ\CC}\cornerUL \rrh{\ol U} & \Div\qQ\Set \rdh{\Divfn\qQ\Set} \\
    \CC \rrm{U} & \Set
 }
\end{displaymath}
Therefore this pullback diagram asserts that $\Divf\qQ\CC$ arises from
the {\em change-of-base} of the fibration $\Divfn\qQ\Set$ along the
global section functor $U:\CC\to\Set$ \mcite{jacobscltt}.

\section{Divergences on Monads}
\label{sec:divergence_of_monads}

We introduce the concept of {\em divergence on monad} as a
quantitative measure of difference between computational effects.
This is hinted from Barthe and Olmedo's composable divergences on
probability distributions
\mcite{DBLP:conf/icalp/BartheO13}. Divergences on monads are defined
upon two extra data called {\em grading monoid} and {\em basic
  endorelation}.

\begin{definition}
  A {\em grading monoid} is a partially ordered 
  monoid $(M,\le,1,(\cdot))$.
\end{definition}
\begin{definition}
  A {\em basic endorelation} is a functor $E\colon \CC\to\brelc$ such
  that $E I$ is an endorelation on $I$.
\end{definition}
Grading monoids will be used when formulating
$(\varepsilon,\delta)$-differential privacy as a divergence on a
monad. Basic endorelations specify which global elements are regarded
as identical.  Any CC $\CC$ has at least two basic endorelations of
{\em equality relations} and {\em total relations}:
\begin{align*}
  \mEQ I&\triangleq(I,I,\{(i,i)~|~i\in UI\})
  & \mT I&\triangleq(I,I,UI\times UI).
\end{align*}
Other examples of basic endorelations can be found in concrete
categories.
\begin{itemize}
\item The category $\Div \qQ\CC$ of $\qQ$-divergences on $\CC$-objects has a basic
  relation $E_\delta$ parameterized by $\delta\in \qQ$. It collects all
  pairs of global elements whose divergence is bound by $d$.
  That is, $E_\delta(I,d)\triangleq (I,I,\{(x_1,x_2)~|~d(x_1,x_2)\le\delta\})$.
  
\item The category of preorders and monotone functions has the basic
  endorelation $E_{eq}$ collecting equivalent global elements:
  $E_{eq}(I,\le)\triangleq (I,I,\{(x,y)~|~x\le y\wedge y\le x\})$.

\end{itemize}
\begin{definition}\label{def:div}
  Let $(\CC,1,(\times),T,\eta,\mu,\theta)$ be a CC-SM, $\qQ$ be a
  divergence domain, $(M,\le,1,(\cdot))$ be a grading monoid and $E\colon \CC\to\brelc$
  be a basic endorelation.  An {\em $E$-relative $M$-graded
    $\qQ$-divergence} (when $M=1$, we drop ``$M$-graded'') on the
  monad $T$ is a doubly-indexed family of $\qQ$-divergences
  $ \asgn = \{\asg m I \colon  (U(TI))^2 \to \qQ\}_{m\in M,I\in\CC} $
  satisfying the following conditions:
  \begin{description}
  \item[Monotonicity] For any $m\le m'$ in $M$, $I\in\CC$ and
    $c_1,c_2\in U(TI)$,
    \begin{displaymath}
      \asg m I(c_1,c_2)\geq \asg {m'} I(c_1,c_2).
    \end{displaymath}

  \item[$E$-unit Reflexivity] For any $I\in\CC$,
    \begin{displaymath}
      \sup_{(x_1,x_2) \in E I}\asg 1 I(\eta_I\ap x_1,\eta_I\ap x_2)\le 0.
    \end{displaymath}

  \item[$E$-composability] For any $m_1,m_2\in M$, $I,J\in\CC$,
    $c_1,c_2\in U(TI)$ and $f_1,f_2\colon I\to TJ$,
    \begin{displaymath}
      \asg {m_1\cdot m_2} J(f_1\kl\ap c_1,f_2\kl\ap {c_2}) \le \asg {m_1}
      I(c_1,c_2) + \sup_{(x_1,x_2) \in E I}\asg{m_2}J(f_1\ap x_1,f_2\ap x_2).
    \end{displaymath}
  \end{description}
  We write $\mDiv TM\qQ E$ for the collection of $E$-relative
  $M$-graded $\qQ$-divergences on $T$.
  We introduce a partial order $\preceq$ on $\mDiv TM\qQ E$ by:
  \begin{displaymath}
    \asgn_1\preceq \asgn_2
    \iff
    \fa{m\in M,I\in\CC,c_1,c_2\in U(TI)}
    (\asgn_1)^m_I(c_1,c_2) \ge (\asgn_2)^m_I(c_1,c_2).
  \end{displaymath}
\end{definition}
The $E$-composability condition is a generalization of the
composability of differential privacy stated as \cite[Theorem
1]{DBLP:conf/icalp/BartheO13}. What is new in this paper is that 1) we
introduce a condition on the monad unit ($E$-unit reflexivity), and
that 2) the sup computed in $E$-unit reflexivity and $E$-composability
scans global elements related by $E$, while \cite{DBLP:conf/icalp/BartheO13}
only considers the case where $E=\mEQ$. We will later show that both
$E$-unit reflexivity and $E$-composability play an important role when
connecting divergences, relational liftings of $T$, and the monad
structure of $T$ - these conditions are necessary and sufficient to
construct {\em strong graded relational liftings} of $T$ satisfying
{\em fundamental property} with respect to divergences (Proposition
\ref{lem:eq}).


\section{Examples of Divergences on Monads}
\label{sec:divergence:examples}

\subsection{Cost Difference for Deterministic Computations}

To aid in understanding the $E$-unit reflexivity and $E$-composability
conditions, we illustrate a few divergences on an elementary
monad: the {\em cost count monad} $T = \NN \times -$ on $\Set$. Its
unit and Kleisli extension are defined by
\begin{align*}
  \eta_I (x) & \triangleq (0,x)
  & f\kl(i,x) & \triangleq (i + \pi_1(f(x)), \pi_2(f(x)))
  &(x\in I,i\in\NN,f\colon I\to TJ).
\end{align*}
The monad $T$ can be used to record the cost incurred by deterministic
computations. For instance, consider the quick sort algorithm $\aquick$
and the insertion sort algorithm $\ainsert$, both of which are modified so
that they tick a count whenever they compare two elements to be
sorted. These two modified sort programs are interpreted as functions
$\sem\aquick,\sem\ainsert\colon \NN^*\to T(\NN^*)$, so that the first
component of $\sem\aquick(x)$ and that of $\sem\ainsert(x)$ report the
number of comparisons performed during sorting $x$.

We first define an $\qN$-divergence
$\mathsf{C}_I $ on $TI$, for each $I\in\Set$, by
\begin{displaymath}
  {\mathsf{C}_I} {((i,x),(j,y))} \triangleq |i - j|.
\end{displaymath}
This divergence $\mathsf{C}_I$ computes the difference of costs between two
computations $(i,x),(j,y) \in TI$, ignoring their return values.  The
family $\mathsf C=\fm \Set I{\mathsf{C}_I}$ \emph{forms} a
$\mT$-relative $\qN$-divergence on $T$.  The $\mT$-unit reflexivity of
$\mathsf{C}$ means that the difference of costs between pure
computations is zero:
\[
\mathsf{C}_I (\eta_I (x),\eta_I (y)) = \mathsf{C}_I ((0,x),(0,y)) = 0.
\]
The $\mT$-composability of $\mathsf{C}$ says that we can limit the
cost difference of two runs of programs $f\kl(i,x)$ and $g\kl(j,y)$ by
the sum of cost difference of the preceding computations $(i,x),(j,y)$
and that of two programs $f,g\colon I\to TJ$. The latter is measured by
taking the sup of cost difference of $f(x)$ and $g(y)$, where $(x,y)$
range over the basic endorelation $\mT I$.
\begin{align*}
\mathsf{C}_I (f\kl(i,x),g\kl(j,y))& = 
\mathsf{C}_I (i + \pi_1(f(x)), \pi_2(f(x)),j + \pi_1(g(y)), \pi_2(g(y)))\\
&\leq
|i - j| + \sup_{x,y \in I} |\pi_1(f(x)) - \pi_1(g(y))|\\
&=
\mathsf{C}_I((i,x),(j,y)) + \sup_{(x,y) \in \mT I} \mathsf{C}_J(f(x),g(y)). 
\end{align*}

We remark that $\mathsf C$ is {\em not} an $\mEQ$-relative
$\qN$-divergence on $T$ because the $\mEQ$-composability fails: when
$f(x)=(0,w)$, $f(y)=(1,w)$ and $f(z) = (0,v)$ (for $z \neq x,y$) we
have $\mathsf{C}_I((0,x),(0,y)) = 0$ and
$ \sup_{(x,y) \in \mEQ I} \mathsf{C}_J(f(x),f(y)) = 0$, but we have
$\mathsf{C}_J (f\kl(0,x),f\kl(0,y)) = \mathsf{C}_J((0,w),(1,w)) = 1$.


Alternatively, we may consider the following $\qN$-divergence
$\mathsf{C}'_I$ on $TI$ for each $I\in\Set$:
\[
\mathsf{C}'_I((i,x),(j,y)) \triangleq
\begin{cases}
|i - j| & x = y \\
\infty & x \neq y
\end{cases}.
\]
This divergence is {\em sensitive} on return values of computations.
When return values of two computations agree, $\mathsf C'$ measures
the cost difference as done in $\mathsf{C}$, but when they do not
agree, the cost difference is judged as $\infty$. This divergence is
an $\mEQ$-relative $\qN$-divergence on $T$.

\begin{toappendix}
  \begin{proposition}
   The family $\mathsf{C}' = \{{\mathsf{C}'_I} \colon (\mathbb{N} \times I)^2 \to \qN \}_{I \in \Set}  $ of $\qN$-divergences defined by
\[
\mathsf{C}'_I((i,x),(j,y)) \triangleq
\begin{cases}
|i - j| & x = y \\
\infty & x \neq y
\end{cases}.
\]
   is a $\mEQ$-relative $\qN$-divergence on the monad $\mathbb{N} \times - $.
  \end{proposition}
  \begin{proof}
  The monotonicity of $\mathsf{C}'$ is obvious.
  
  We show the $\mEQ$-unit-reflexivity of $\mathsf{C}'$.
    For all $(x,y) \in \mEQ I$ (that is, $x = y \in I$), we have 
    \[
      \mathsf{C}'_I (\eta_I (x),\eta_I (y)) = {\mathsf{C}'_I}
      {((0,x),(0,y))} = 0.
    \]
    
    We show the $\mEQ$-composability of $\mathsf{C}'$.
    Let $(i,x), (j,y) \in \mathbb{N} \times I$ and $f,g \colon I \to \mathbb{N} \times J$.
    We write $f(z) = (i_z,f_z)$ and $g(z) = (j_z,g_z)$ for each $z \in Z$.
    \begin{itemize}
    \item
 	If $x = y$ and $x_z = y_z$ for all $z \in I$, we have
    \begin{align*}
      \mathsf{C}'_J (f\kl(i,x),g\kl(j,y))
      &= \mathsf{C}'_J(i +  i_x,f_{x}),(j +  j_x,g_{x}))\\
      &= |(i +  i_x) - (j +  j_x)|\leq |i  - j | + |i_x - j_x| \\
      &\leq \mathsf{C}'_I((i,x),(j,y)) + \sup_{(x,y) \in \mEQ I (\iff x = y \in I)}\mathsf{C}'_J(f(x),g(y))
    \end{align*}
    \item
    If $x \neq y$ or $f_z \neq g_z$ for some $z \in I$, we have 
    \[
    \mathsf{C}'_J (f\kl(i,x),g\kl(j,y))
    \leq
 	\infty
 	=
    \mathsf{C}'_I((i,x),(j,y)) + \sup_{(x,y) \in \mEQ I (\iff x = y \in I)}\mathsf{C}'_J(f(x),g(y)).
    \]
    \end{itemize}
  This completes the proof.
  \end{proof}
\end{toappendix}

\subsection{Cost Difference for Nondeterministic Computations}\label{sec:relcost}

Deterministic and nondeterministic
computations with cost counting can be respectively modeled by the
monads $(\mathbb{N} \times -)$ and $\pwr(\mathbb{N} \times -)$ on
$\Set$.

We define the divergences for cost difference as in Table
\ref{tab:divcost}.  
These divergences extract the upper bound of cost difference
between two computations. 
The divergences $\mathsf{C}$ and $\mathsf{NC}$ measure the usual distance 
of costs for deterministic and nondeterministic computations respectively.
The divergence $\mathsf{NCI}$ measures the \emph{subtraction} of costs of two
nondeterministic computations.
For results of two nondeterministic computations $A, B \in \pwr(\mathbb{N} \times I)$,
the divergence $\mathsf{NCI}_I(A,B)$ is an upper bound of $i - j$ 
for all possible choices of $(i,x) \in A$ and $(j,y) \in B$, where
a lower bound of $i - j$ is also given by $-\mathsf{NCI}_I(B,A)$.
The same idea to measure the difference of costs between two programs
by subtraction also appears in \mcite{Cicek:2017:RCA:3093333.3009858,Radicek:2017:MRR:3177123.3158124}.
If either $A$ or $B$ is empty, we fail to get an information of costs.
We then have $\mathsf{NCI}_I(A,B) = -\infty$.
On the other hand, if both $A$ and $B$ are not empty,
their \emph{cost intervals} are defined by
\[
[l_A,h_A] \triangleq [\inf_{(i,x) \in A} i, \sup_{(i,x) \in A} i],\qquad [l_B,h_B] \triangleq [\inf_{(j,y) \in B} j, \sup_{(j,y) \in B} j].
\]
We then have $\mathsf{NCI}_I(A,B) = h_A - l_B$ and $-\mathsf{NCI}_I(B,A) = l_A - h_B$.

\begin{table}
  \caption{($1$-graded) $\mT$-relative $Q$-divergences for cost counting monads}
  \begin{tabular}{cccl}
    \hline $\asgn \in \mDiv T1\qQ\mT$ & $T$ & $\qQ$ & Definition of ${\asg {} I}(c_1,c_2)$ \\
    \hline $\mathsf{C}$ & $\mathbb{N} \times - $ & $\qN$ & $ {\mathsf{C}_I} {((i,x),(j,y))} = |i - j| $ \\
    \hline $\divnnc$ & $\pwr(\mathbb{N} \times -)$ & $\qN$ & ${\divnc I}{(A,B)} = \sup_{(i,x) \in A, (j,y) \in B} |i - j|$ \\ 
    \hline $\mathsf{NCI}$ & $\pwr(\mathbb{N} \times -)$ & $\qZ$ &
    $\mathsf{NCI}_I(A,B) = \sup_{(i,x) \in A, (j,y) \in B} i - j $\\
    \hline
  \end{tabular}
  \label{tab:divcost}
\end{table}

\begin{toappendix}
 
  \begin{proposition}
   The family $\mathsf{NC} = \{\mathsf{NC}_I \colon (\pwr(\mathbb{N} \times I))^2 \to \qN \}_{I \in \Set}$ of $\qN$-divergences defined by
   \[
    \mathsf{NC}_I(A,B) \triangleq  \sup_{(i,x) \in A, (j,x) \in B} |i - j|\]
   is a $\mT$-relative $\qN$-divergence on the monad $\pwr(\mathbb{N} \times - )$.
  \end{proposition}
  \begin{proof}
  The monotonicity of $\mathsf{NC}$ is obvious.
  
  We show the $\mT$-unit-reflexivity of $\mathsf{NC}$ .
  For all $(x,y) \in \mT I$ (that is, $x,y \in I$), we have 
    \[
      \mathsf{NC}_I (\eta_I (x),\eta_I (y)) = {\mathsf{NC}_I} {(\{
        (0,x) \},\{ (0,y) \} )} = |0 - 0| = 0.
    \]
    
   We show the $\mT$-composability of $\mathsf{NC}$.
   For all $f,g \colon I \to \pwr(\mathbb{N} \times J )$ and $A,B \in \pwr(\mathbb{N} \times I )$, we have
    \begin{align*}
      \mathsf{NC}_J (f\kl A,g\kl B)
      &= \sup \{ |i - j| \mid (i,x) \in f\kl(A), (j,y) \in g\kl(B) \}\\
      &=
        \sup \left\{
        |i_1 + i_2 - j_1 - j_2|~\middle|
        \begin{array}{l@{}}
          (i_1,x) \in A, (j_1,y) \in B, \\
          (i_2,x') \in f(x), (j_2,y') \in g(y)
        \end{array}
      \right\}\\
      &\leq \sup\{ |i_1 - j_1| ~|~ (i_1,x) \in A, (j_1,y) \in B \}\\
      &\qquad + \sup_{(x,y) \in \mT I (\iff x,y \in I)}\{|i_2 - j_2|  ~|~ (i_2,x') \in f(x), (j_2,y') \in g(y) \}\\
      &= \mathsf{NC}_I (A, B) + \sup_{(x,y) \in \mT I (\iff x,y \in I)}\mathsf{NC}_J (f(x),g(y)).
    \end{align*}
  This completes the proof.
  \end{proof}
\begin{proposition}
   The family $\mathsf{NCI} = \{\mathsf{NCI}_I \colon (\pwr(\mathbb{N} \times I))^2 \to \qN \}_{I \in \Set}$ of $\qZ$-divergences defined by
   \[
   \mathsf{NCI}_I(A,B) \triangleq
     \sup_{(i,x) \in A, (j,y) \in B} i - j 
   \]
   is a $\mT$-relative $\qZ$-divergence on the monad $\pwr(\mathbb{N} \times - )$.
  \end{proposition}
  \begin{proof}
  The monotonicity of $\mathsf{NCI}$ is obvious.
  
  We show the $\mT$-unit-reflexivity of $\mathsf{NCI}$ .
  For all $(x,y) \in \mT I$ (that is, $x,y \in I$), we have 
    \[
      \mathsf{NCI}_I (\eta_I (x),\eta_I (y)) = {\mathsf{NCI}_I} {(\{
        (0,x) \},\{ (0,y) \} )} = 0 - 0 = 0.
    \]
    
   We show the $\mT$-composability of $\mathsf{NCI}$.
   For all $f,g \colon I \to \pwr(\mathbb{N} \times J )$ and $A,B \in \pwr(\mathbb{N} \times I )$, we have
    \begin{align*}
      \mathsf{NCI}_J (f\kl A,g\kl B)
      &= \sup \{ i - j \mid (i,x) \in f\kl(A) \land (j,y) \in g\kl(B) \}\\
      &=
        \sup \left\{
        i_1 + i_2 - j_1 - j_2~\middle|
        \begin{array}{l@{}}
          (i_1,x) \in A, (j_1,y) \in B, \\
          (i_2,x') \in f(x), (j_2,y') \in g(y)
        \end{array}
      \right\}\\
      &\leq \sup\{ i_1 - j_1  ~|~ (i_1,x) \in A, (j_1,y) \in B \}\\
      &\qquad + \sup_{(x,y) \in \mT I (\iff x,y \in I)}\{ i_2 - j_2  ~|~ (i_2,x') \in f(x), (j_2,y') \in g(y) \}\\
      &= \mathsf{NCI}_I (A, B) + \sup_{(x,y) \in \mT I (\iff x,y \in I)}\mathsf{NCI}_J (f(x),g(y)).
    \end{align*}
  This completes the proof.
  \end{proof}
\end{toappendix}

\subsection{Divergences for Differential Privacy}

Differential privacy (DP for short) is a quantitative definition of
privacy of randomized queries in databases. DP is based on the idea of
noise-adding anonymization against background-knowledge attacks.
In the study of DP, a query is modeled by a measurable function
$c \colon I \to \giry J$, where $I$ and $J$ are measurable spaces of
inputs and outputs respectively, and $\giry J$ is the measurable space
of all probability measures over $J$; here $\giry$ itself refers to
the {\em Giry monad} (\cite{Giry1982}; see also Section
\ref{sec:meas}).
\begin{definition}[Differential Privacy, \mcite{DworkMcSherryNissimSmith2006}]
  Let $c \colon I \to \giry J$ be a morphism in $\Meas$, representing
  a randomized query.  The query $c$ satisfies
  $(\varepsilon,\delta)$-differential privacy ($\epsilon,\delta\ge 0$
  are reals) if for any adjacent datasets
  $(d_1,d_2) \in R_{\mathrm{adj}}$\footnote{ Strictly speaking,
    differential privacy depends on the definition of adjacency of
    datasets.  The adjacency relation $R_{\mathrm{adj}}$ is usually
    defined as $\{(d_1,d_2) | \rho(d_1,d_2) \leq 1\}$ with a metric
    $\rho$ over $I$.}, the following holds:
  \begin{displaymath}
    \forall S \subseteq_{\mathrm{measurable}} J.~ \Pr[ c(d_1) \in S]
    \leq
    \exp(\varepsilon) \Pr[ c(d_2) \in S] + \delta.
  \end{displaymath}
\end{definition}
To express this definition in terms of divergence on monad, we
introduce a doubly-indexed family of $\qRp$-divergence
$\divndp=\{\divdp\varepsilon J\}_{\varepsilon\in[0,\infty],J\in\Meas}$
on $\giry J$ by
\begin{displaymath}
  \divdp\varepsilon J (\mu_1,\mu_2)\triangleq
  \sup_{S\in \Sigma_J } (\mu_1(S)  - \exp(\varepsilon) \mu_2(S))\quad
  (\mu_1,\mu_2\in \giry J).
\end{displaymath}
Then the query $c \colon I \to \giry J$ satisfies
$(\varepsilon,\delta)$-DP if and only if
\begin{displaymath}
  \fa{(d_1,d_2)\in R_{\mathrm{adj}}}
  \divdp\varepsilon J(c(d_1),c(d_2))\le\delta.
\end{displaymath}
The pair $(\varepsilon,\delta)$ indicates the difference between
output probability distributions $c(d_1)$ and $c(d_2)$ of the
query $c$ for given datasets $d_1$ and $d_2$.  Intuitively, the
parameter $\varepsilon$ is an upper bound of the ratio
$\Pr[c(d_1) = s]/\Pr[c(d_2) = s] $ of probabilities which indicates
the leakage of privacy.  If $\varepsilon$ is large, attackers can
distinguish the datasets $d_1$ and $d_2$ from the outputs of the query
$c$. The parameter $\delta$ is the probability of
failure of privacy protection.

The family $\divndp$ forms an $\mEQ$-relative $\qRp$-graded
$\qRp$-divergence on the Giry monad $\giry$ \cite[Lemma
6]{DBLP:conf/lics/SatoBGHK19}. This is proved by extending the
composability of the divergence for DP on discrete probability
distributions shown as \cite[Lemmas 3 and
6]{DBLP:conf/popl/BartheKOB12} and \cite[Proposition
5]{DBLP:conf/icalp/BartheO13}, based on the \emph{composition theorem}
of DP \cite[Section 3.5]{DworkRothTCS-042}.

The conditions in Definition \ref{def:div} on $\divndp$ corresponds to
the following basic properties of DP:
\begin{description}
\item[(monotonicity)] The monotonicity of $\divndp$ corresponds to
  weakening the differential privacy of queries: if $c$ satisfies
  $(\varepsilon,\delta)$-DP and $\varepsilon \leq \varepsilon'$ and
  $\delta \leq \delta'$ holds, then $c$ satisfies
  $(\varepsilon',\delta')$-DP.
  
\item[($\mEQ$-unit reflexivity)] The $\mEQ$-unit reflexivity of
  $\divndp$ implies
  $\divdp 0 J(\eta_J \circ h(x),\eta_J \circ h(x)) = 0$ for any
  measurable function $h \colon I \to J $ and $x \in I$.  This,
  together with the composability below, ensures the {\em robustness}
  of DP of a query $c\colon I\to\giry J$ with respect to deterministic
  postprocessing:
  \begin{equation}
    \label{eq:dpostpro}
    \fa{h:J\to K}
    \text{$c$ is $(\epsilon,\delta)$-DP}
    \implies
    \text{$\giry h\circ c$ is $(\epsilon,\delta)$-DP}.
  \end{equation}
  In fact, the divergence $\divndp$ is reflexive: we have
  $\divdp 0J(\mu,\mu) = 0$ for every $\mu \in \giry J$. Therefore
  $h\colon J\to K$ and $\giry h$ in \eqref{eq:dpostpro} can be replaced by
  $h\colon J\to\giry K$ and $h\kl$; the replaced condition states the {\em
    robustness} of DP of a query with respect to probabilistic
  postprocessing.
  
\item[($\mEQ$-composability)] The $\mEQ$-composability of $\divndp$
  corresponds to the known property of DP called the {\em sequential
    composition theorem} \mcite{DworkRothTCS-042}.  If
  $c_1 \colon I \to \giry J'$ and $c_2 \colon J' \to \giry J$ are
  $(\varepsilon_1,\delta_1)$-DP and $(\varepsilon_2,\delta_2)$-DP
  respectively, then the sequential composition
  $c_2\kl\circ c_1 \colon I \to \giry J$ of the queries $c_1$ and
  $c_2$ is $(\varepsilon_1+\varepsilon_2,\delta_1+\delta_2)$-DP.
\end{description}

\paragraph*{A Non-Example: Pointwise Differential Privacy.}
We stated above that a parameter $(\varepsilon, \delta)$ of DP
intuitively gives an upper bound of the probability ratio
$\Pr[c(d_1) = s]/\Pr[c(d_2) = s]$ and the probability of failure of
privacy protection.  However, strictly speaking, there is a gap
between the definition of $(\varepsilon,\delta)$-DP and this intuition
of $\varepsilon$ and $\delta$.  \emph{Pointwise differential privacy}
(\cite[Definition 3.2]{Prasad803anote} and \cite[Proposition
1.2.3]{Hall_newstatistical}) is a finer definition of DP that is
faithful to the intuition.
\begin{definition}
  A measurable function $c \colon I \to \giry J$ (regarded as a query)
  is {\em pointwise $(\varepsilon,\delta)$-differentially private} if
  whenever $d_1$ and $d_2$ are adjacent, for some $A \in \Sigma_J$
  with $\Pr[ c(d_1) \notin A] \leq \delta$, we have
\[
\forall s \in A.~\Pr[c(d_1) = s] \leq \exp(\varepsilon) \Pr[c(d_2) = s],
\]
which is equivalent to \footnote{ Remark that $\Pr[c(d_1) = s]$ and
  $\Pr[c(d_2) = s]$ are Radon-Nikodym derivatives of $c(d_1)$ and
  $c(d_2)$ with respect to a measure $\nu$ such that $c(d_1), c(d_2) \ll \nu$.
  [$\implies$] Obvious. [$\impliedby$] By Radon-Nikodym theorem
  we can take the Radon-Nikodym derivatives $\Pr[c(d_1) = s]$ and $\Pr[c(d_2) = s]$
  with respect to $\nu = c(d_1) + c(d_2)$.
  The inequality does not depend on the choice of $\nu$.
  }
\[
\forall S \subseteq_{\mathrm{measurable}} A.~\Pr[c(d_1) \in S] \leq \exp(\varepsilon) \Pr[c(d_2) \in S].
\]
\end{definition}
To express this definition in terms of divergence on monad, we
introduce a doubly-indexed family of $\qRp$-divergences
$\divnpwdp= \{\divpwdp\varepsilon
J\}_{\varepsilon\in \qRp,J\in\Meas}$ called \emph{pointwise
  indistinguishability}:
\[
  \divpwdp\varepsilon J(\mu_1,\mu_2) \triangleq \inf
  \left\{\mu_1(J \setminus A) ~|~ A \in \Sigma_X \wedge
    (\forall{S \in \Sigma_J}. S \subseteq A \implies \mu_1(S) \leq
    \exp(\varepsilon) \mu_2(S)) \right\}.
\]
Then  $c \colon I \to \giry J$ is pointwise $(\varepsilon,\delta)$-differentially private
if and only if
\begin{displaymath}
  \fa{(d_1,d_2)\in R_{\mathrm{adj}}}
  \divpwdp\varepsilon J(c(d_1),c(d_2)) \leq \delta.
\end{displaymath}
The family $\divnpwdp$ is obviously reflexive:
$\divpwdp\varepsilon J(\mu,\mu) = 0$ holds for any $\mu \in \giry J$
and $\varepsilon \geq 0$. Hence it is $\mEQ$-unit reflexive too.
However, it is not $\mEQ$-composable. We let $3 = \{0,1,2\}$ and
$2 = \{0,1\}$ be discrete spaces, and let
$\alpha = \exp(\varepsilon)$.  We define two probability distributions
$\mu_1,\mu_2 \in \giry 3$ by
\[
\mu_1 \triangleq \frac{1}{10}\mathbf{d}_0 + \frac{9}{10}\mathbf{d}_1, \qquad 
\mu_2 \triangleq \frac{9}{10\alpha}\mathbf{d}_1 + (1 - \frac{9}{10\alpha})\mathbf{d}_2.
\]
We then have $\divpwdp\varepsilon 3(\mu_1,\mu_2) = \frac{1}{10}$ with $A = \{1,2\}$
since $\frac{1}{10} > \exp(\varepsilon) \cdot 0$,  $\frac{9}{10} \leq \exp(\varepsilon) \cdot \frac{9}{10\alpha}$, and $0 \leq \exp(\varepsilon) \cdot (1 - \frac{9}{10\alpha})$.
Next, we define $f \colon 3 \to \giry 2$ 
by
\[
f(0) \triangleq \frac{1}{10}\mathbf{d}_0 + \frac{9}{10}\mathbf{d}_1,\quad
f(1) \triangleq \frac{9}{10}\mathbf{d}_0 + \frac{1}{10}\mathbf{d}_1, \quad
f(2) \triangleq \mathbf{d}_1.
\]
We then calculate
\[
f\kl(\mu_1) = \frac{82}{100}\mathbf{d}_0 + \frac{18}{100}\mathbf{d}_1, \quad 
f\kl(\mu_2) = \frac{81}{100\alpha}\mathbf{d}_0 + (\frac{100\alpha- 90 + 9}{100\alpha})\mathbf{d}_1.
\]
Then, we obtain $\divpwdp\varepsilon 2(f\kl(\mu_1),f\kl(\mu_2)) = \frac{82}{100}$ with $A = \{0\}$ since $ \frac{82}{100} >  \exp(\varepsilon) \frac{81}{100\alpha}$.
Hence $\divpwdp\varepsilon 2(f\kl(\mu_1),f\kl(\mu_2))  = \frac{82}{100} > \frac{1}{10} = \divpwdp\varepsilon 3(\mu_1,\mu_2)$.
Thus $\divnpwdp$ is not $\mEQ$-composable, because by the reflexivity of $\divnpwdp$, we have $\sup_{(x,y) \in \mEQ_3} \divpwdp 0{}(f(x),f(y)) = 0$.

\paragraph*{Various Relaxations of Differential Privacy}
Since the seminal work on DP by \cite{DworkMcSherryNissimSmith2006},
various relaxations of differential privacy have been proposed:
{\em R\'enyi DP} \mcite{MironovCSF17},
{\em zero-concentrated DP} \mcite{BSTCC16} and
{\em truncated zero-concentrated DP} \mcite{BDRSSTOC18}.
They give tighter bounds of differential privacy.
These relaxations of differential privacy can be expressed by suitable
divergences on the Giry monad $\giry$ and sub-Giry monad $\sgiry$; see
Table \ref{tab:divdp} for their definitions. There,
$\alpha,w\in(1,\infty)$ are non-grading parameters for $\divnrn$ and
$\divntcdp$.  Each row of the table represents that $\asgn$ is an
$\mEQ{}$-relative $\qQ$- (resp. $\qQs$-) divergences on $\giry$
(resp. $\sgiry$), and the definition of $\asg {} I(\mu_1,\mu_2)$
follows.
\begin{table}
  \caption{$\mEQ{}$-relative $M$-graded $\qQ$- ($\qQs$-)divergences on $\giry$ ($\sgiry$)}
  \begin{tabular}{ccccll} \hline $\asgn$ & $M$ & $\qQ$ & $\qQs$ &
                                                                      Definition of ${\asg m I}(\mu_1,\mu_2)$ & Composability proof \\
    \hline $\divndp$ & $\qRp$ & $\qRp$ & $\qRp$ &
                                                  $\sup_{S
                                                  \in\Sigma_I}
                                                  (\mu_1(S) -
                                                  \exp(\varepsilon)
                                                  \mu_2(S))$ &
                                                               \mcite{DBLP:conf/icalp/BartheO13}
    \\ \hline $\divrn\alpha{}$ & 1 & $\qRp$ & $\qRa$ &
                                                       $\frac{1}{\alpha
                                                       - 1} \log
                                                       \int_I \left(
                                                       \frac{\mu_1(x)}{\mu_2(x)}
                                                       \right)^\alpha
                                                       \mu_2(x)~dx.$ &
                                                                       \mcite{MironovCSF17}
    \\ \hline $\divnzcdp$ & $\qRp$ & $\qRp$ & $\qRa$ &
                                                       $\sup_{1 <
                                                       \alpha}
                                                       \frac{1}{\alpha}(\divrn
                                                       \alpha
                                                       I(\mu_1,\mu_2)
                                                       - m)$ &
                                                               \mcite{BSTCC16}
    \\ \hline $\divtcdp w{}$ & 1 & $\qRp$ & $\qRa$ &
                                                     $\sup_{1 < \alpha
                                                     < w}
                                                     \frac{1}{\alpha}(\divrn
                                                     \alpha
                                                     I(\mu_1,\mu_2))$ & \mcite{BDRSSTOC18} \\ \hline
  \end{tabular}
  \label{tab:divdp}
\end{table}

\subsection{Statistical Divergences and Composablity of  $f$-Divergences}
\label{sec:sdcfd}

Apart from differential privacy, various distances between
(sub-)probability distributions are introduced in probability theory.
They are called {\em statistical divergences}. Examples include: {\em
  total variation distance} $\divntv$, {\em Hellinger distance}
$\divnhd$, {\em Kullback-Leibler divergence} $\divnkl$, and {\em
  $\chi^2$-divergence} $\divnchi$; they are defined in Table
\ref{tab:divstat}. These statistical divergences are $\mEQ$-relative
divergences on the Giry monad $\giry$ (and $\sgiry$ for $\divntv$); see
the same table for their divergence domains.  Question marks in the
column of $\qQs$ means that we do not know with which monoid structure
the $\mEQ$-composability holds. We remark that these divergences are
also reflexive, that is, $\asgn(c,c)=0$. $\mEQ$-composability of
these divergences in discrete form are proved
in~\mcite{DBLP:conf/icalp/BartheO13,olmedo2014approximate}.  Later,
\cite{DBLP:conf/lics/SatoBGHK19} extends their results to
the composability of divergences in continuous form.
\begin{table}
  \caption{Statistical divergences that are $\mEQ{}$-relative $\qQ$-
    (resp. $\qQs$-) divergences on $\giry$ (resp. $\sgiry$)}
  \begin{tabular}{lcccl} \hline Name & $\asgn$ & $\qQ$ & $\qQs$ &
                                                                      Definition of ${\asg m I}(\mu_1,\mu_2)$\\
    \hline Total variation distance&
    
    $\divntv$ & $\qRp$ & $\qRp$ &
                                             $\frac{1}{2}\int_I|\mu_1(x)
                                             - \mu_2(x)|~dx$ 
    \\ \hline Kullback-Leibler divergence & $\divnkl$ & $\qRp$ & ? &
                                           $\int_I \mu_1(x) \log
                                           \left(
                                           \frac{\mu_1(x)}{\mu_2(x)}\right)~dx$                                                                  
    \\ \hline Hellinger distance & $\divnhd$ & $\qRp$ & ? &
                                           $\frac{1}{2}\int_I \left(
                                           \sqrt{\mu_1(x)} -
                                           \sqrt{\mu_2(x)} \right)^2
                                           ~dx$ 
                                                  \\ \hline
    $\chi^2$-divergence &$\divnchi$ & $\qRg 1$ & ? &
                                    $\int_I \frac{(\mu_1(x) -
                                    \mu_2(x))^2}{\mu_2(x)}~dx$ 
    \\ \hline
  \end{tabular}
  \label{tab:divstat}
\end{table}

Each of four divergences in Table \ref{tab:divstat} can be expressed
as an \emph{$f$-divergence} $\divnf$ (\cite{fdiv,Csiszar67,fdivm}):
\[
  \divf I(\mu_1,\mu_2) \triangleq \int_I
  \mu_2(x)f\left(\frac{\mu_1(x)}{\mu_2(x)}\right)dx.
\]
Here, $f$ is a parameter called {\em weight function}, and has to be a
convex function $f \colon [0,\infty) \to \RR$, continuous at $0$ and
satisfying $\lim_{x\to +0}xf(x)=0$. Weight functions for four
divergences $\divntv, \divnkl,\divnhd,\divnchi$ are in Table
\ref{tab:example_parameters}.  In fact, $\divndp^\varepsilon$ is also
an $f$-divergence with weight function
$f(t)=\max (0,t-\exp(\varepsilon))$; see~\cite[Proposition
2]{DBLP:conf/icalp/BartheO13}. We also remark that R\'enyi divergence
$\divrn \alpha {}$ of order $\alpha$ is the logarithm of the
$f$-divergence with weight function $f(t) = t^\alpha$.

$f$-divergences have several nice properties such as reflexivity,
postprocessing inequality, joint-convexity, duality and
continuity~\mcite{Csiszar67,1705001_2006}. However, the
$\mEQ$-composability of $f$-divergences is not guaranteed in
general. Here we provide a sufficient condition for the $\mEQ$-composability
of $\divnf$ over a specific form of divergence domain.
\begin{proposition}
  \label{prop:composable:f-divergence:weight}
  Let $\gamma\geq 0$ be a nonnegative real number,
  $\qRg\gamma =([0,\infty],\le,0,\lam{(p,q)}p+q+\gamma pq)$ be the
  divergence domain, and $f$ be a weight function such that $f \geq 0$
  and $f(1) = 0$.  If there exists $\alpha, \beta, \beta' \in \RR$
  such that, for all $x,y,z,w \in [0,1]$, the following hold (suppose $0 f (0/0) = 0$):
  \begin{align*}
    0 & \leq (\beta' z + (1-\beta') x) + \gamma x f\left({z}/{x}\right) \\
    xy f\left({zw}/{xy}\right) &\leq
                                 (\beta w + (1-\beta) y) x f\left({z}/{x}\right) \nonumber + (\beta' z + (1-\beta') x) y f\left({w}/{y}\right)\\
      & \quad + \gamma xy f\left({z}/{x}\right)f\left({w}/{y}\right) \nonumber + \alpha (x-z)(w-y),
  \end{align*}
  then $\divnf$ is an $\mEQ{}$-relative $\qRg\gamma$-divergence on the
  Giry monad $\giry$.  When $\alpha = 0$ and $\beta,\beta' \in [0,1]$,
  $\giry$ can be replaced with the sub-Giry monad $\sgiry$.
\end{proposition}
The proof of this proposition generalizes
and integrates the proofs given in
\cite[Section 5.A.2]{olmedo2014approximate}. This proposition is
applicable to prove the composability of divergences in Table
\ref{tab:divstat} by choosing suitable parameters; see Table
\ref{tab:example_parameters}.
\begin{table}
  \caption{Parameters for Proposition
    \ref{prop:composable:f-divergence:weight}}
  \begin{tabular}{cccccc} \hline $\divnf$ & Weight function $f$ &
    $\gamma$ & $\alpha$ & $\beta$ & $\beta'$ \\ \hline $\divntv$ &
    $f(t) = |t - 1| / 2$ & $0$ & $0$ & $1$ & $0$ \\ \hline $\divnkl$ &
    $f(t) = f \log(t) - t + 1$ & $0$ & $-1$ & $1$ & $1$ \\ \hline
    $\divnhd$ & $f(t) = (\sqrt{t} - 1)^2 / 2$ & $0$ & $-1/4$ & $1/2$ &
    $1/2$ \\ \hline $\divnchi$ & $f(t) = (t - 1)^2 / 2$ & $1$ & $-2$ &
    $2$ & $2$ \\ \hline
  \end{tabular}
  \label{tab:example_parameters}
\end{table}
\begin{appendixproof}
  (Proof of Proposition \ref{prop:composable:f-divergence:weight})
  We have $\mEQ{}$-unit reflexivity because the reflexivity $\divf I (\mu,\mu) = 0$ is obtained from $f(1) = 0$.
  We show $\mEQ{}$-composability.
  To show this, we prove a bit stronger statement.
  Consider three positive weight functions $f, f_1, f_2 \geq 0$ with
  $f(1) = f_1(1) = f_2(1) = 0$.  Assume that there are some
  $\alpha, \beta, \beta' \in \RR$ satisfying the following conditions:
\begin{itembox}[l]{{\bf (A')}}
  for all $x,y,z,w \in [0,1]$,
  $0 \leq (\beta' z + (1-\beta') x) + \gamma x
  f_1\left({z}/{x}\right)$ and
  \begin{align*}
    xy f\left({zw}/{xy}\right) &\leq
                                 (\beta w + (1-\beta) y) x f_1\left({z}/{x}\right) + (\beta' z + (1-\beta') x) y f_2\left({w}/{y}\right)\\
                               & \quad + \gamma xy f_1\left({z}/{x}\right)f_2\left({w}/{y}\right) + \alpha (x-z)(w-y).
  \end{align*}
\end{itembox}
  Let $\mu_1, \mu_2 \in \sgiry I$, and let
  $h,k \colon I \to \sgiry J$.  We want to show the composability in
  the sense of \cite[Definition 5.2]{olmedo2014approximate}:
  \begin{equation}\label{inequality:composition:f-div:target0}
  \begin{split}
    &{{}^f\mathsf{Div}}_J(h\kl \mu_1, k\kl \mu_2)\\
    & \leq {{}^{f_1}\mathsf{Div}}_I(\mu_1, \mu_2) +  \sup_{x \in I} {{}^{f_2}\mathsf{Div}}_J(h(x), k(x))
    + \gamma {{}^{f_1}\mathsf{Div}}_I(\mu_1, \mu_2)\cdot \sup_{x \in I} {{}^{f_2}\mathsf{Div}}_J(h(x), k(x)).
  \end{split}
  \end{equation}
  We first fix a \emph{measurable partition} $\{A_i\}_{i = 0}^n$ of $J$, that is
  a family $\{A_i\}_{i = 0}^n$ of measurable subsets $A_i \in \Sigma_J$
  satisfying $i \neq j \implies A_i \cap A_j = \emptyset$ and 
  $\bigcup_{i = 0}^n A_i = J$.
  For each $0 \leq i \leq n$, we fix 
  two monotone increasing sequences $\{h^i_l\}_{l=0}^\infty$ and
  $\{k^i_l\}_{l=0}^\infty$ of simple functions that converge uniformly to
  measurable functions
  $h(-)(A_i) \colon I \to [0,1]$ and $k(-)(A_i) \colon I \to [0,1]$
  respectively.
  The above composability \eqref{inequality:composition:f-div:target0}
  is then equivalent to
\begin{equation}\label{inequality:composition:f-div:target}
\begin{split}
&\lim_{l \to \infty} \sum_{i = 0}^n (\int_X k^i_l~d\mu_2 ) f\left( \frac{\int_X h^i_l~d\mu_1}{\int_X k^i_l~d\mu_2}\right)\\
&\leq {{}^{f_1}\mathsf{Div}}_I(\mu_1,\mu_2) + \sup_{x \in I}{{}^{f_2}\mathsf{Div}}_J(h(x),k(x)) 
+ \gamma {{}^{f_1}\mathsf{Div}}_I(\mu_1,\mu_2)\sup_{x \in I}{{}^{f_2}\mathsf{Div}}_J(h(x),k(x)).
\end{split}
\end{equation}
  We fix $l \in \NN$.
  We suppose
  $h^i_l = \sum_{j = 0}^m \alpha^i_{j}\chi_{B_j}$ and
  $k^i_l = \sum_{j = 0}^m \beta^i_{j}\chi_{B_j}$ for
  some
  $\alpha^i_{j},\beta^i_{j} \in [0,1]$ ($0 \leq j \leq m$)
  and a measurable partition $\{B_j\}_{j = 0}^m$ of $I$.
  
  Thanks to the condition {\bf (A')}, we calculate as follows:
  \begin{align*}
    \lefteqn{\sum_{i = 0}^n (\int_X k^i_l~d\mu_2 ) f\left( \frac{\int_X h^i_l~d\mu_1}{\int_X k^i_l~d\mu_2}\right)}\notag\\
    &\leq \sum_{i = 0}^n \sum_{j = 0}^m \beta^i_{j}\mu_2(B_j) f\left(\frac{\alpha^i_{j}\mu_1(B_j)}{\beta^i_{j}\mu_2(B_j)}\right)\notag\\
    &
      \leq \quad \underbrace{\sum_{i = 0}^n \sum_{j = 0}^m (\beta \alpha^i_j + (1-\beta) \beta^i_j)~ \mu_2(B_j) f_1\left(\frac{\mu_1(B_j)}{\mu_2(B_j)}\right)}_{\triangleq V_1} \\
    &
      \quad + \underbrace{\sum_{i = 0}^n \sum_{j = 0}^m \left(\beta' \mu_1(B_j) + (1-\beta') \mu_2(B_j)) + \gamma \mu_2(B_j) f_1 \left(\frac{\mu_1(B_j)}{\mu_2(B_j)}\right)\right)
      \beta^i_j f_2 \left(\frac{\alpha^i_j }{\beta^i_j }\right)}_{\triangleq V_2}\\
    &
      \quad + \underbrace{\sum_{i = 0}^n \sum_{j = 0}^m \alpha (\mu_2(B_j) - \mu_1(B_j))(\alpha^i_j - \beta^i_j)}_{\triangleq V_3}
  \end{align*}
  We evaluate the above three subexpressions
  $V_1,V_2,V_3$ as follows.

  We evaluate $V_1$ as follows:
  \begin{align*}
     V_1
    & \leq \left(\sup_{0 \leq j \leq m} \sum_{i = 0}^n (\beta \alpha^i_j + (1-\beta) \beta^i_j)\right) \cdot \sum_{j = 0}^m \mu_2(B_j) f_1\left(\frac{\mu_1(B_j)}{\mu_2(B_j)}\right)\\
    & = \sup_{x \in I} \left(\beta \sum_{i = 0}^n h^i_l(x) + (1-\beta) \sum_{i = 0}^n k^i_l(x)\right)\cdot \sum_{j = 0}^m \mu_2(B_j) f_1\left(\frac{\mu_1(B_j)}{\mu_2(B_j)}\right)\\
    & \leq \sup_{x \in I} \left(\beta \sum_{i = 0}^n h^i_l(x) + (1-\beta) \sum_{i = 0}^n k^i_l(x)\right) \cdot {{}^{f_1}\mathsf{Div}}_I(\mu_1,\mu_2)\\
    & \xrightarrow[]{~l \to \infty~} \sup_{x \in I} \left(\beta h(x)(J) + (1-\beta) k(x)(J) \right) \cdot{{}^{f_1}\mathsf{Div}}_I(\mu_1,\mu_2)\\
    & \leq {{}^{f_1}\mathsf{Div}}_I(\mu_1,\mu_2)
  \end{align*}
  Here, 
  the first inequality is given from the non-negativity of each $\mu_2(B_j) f_1\left(\frac{\mu_1(B_j)}{\mu_2(B_j)}\right)$;
  the equality is given by definition of $\alpha_{j}^{i}$ and $\beta_{j}^{i}$;
  the second inequality can be given by the continuity of $ {{}^{f_1}\mathsf{Div}}$
  (\cite[Theorem 16]{1705001_2006}; \cite[Theorem 3]{DBLP:conf/lics/SatoBGHK19} for the sub-Giry monad $\sgiry$):
  \[
  {{}^{f_1}\mathsf{Div}}_I(\mu_1,\mu_2)
  =
  \sup \left \{ \sum_{j = 0}^m \mu_2(B_j)  f_1\left(\frac{\mu_1(B_j)}{\mu_2(B_j)}\right)
  \middle | \{B_j\}_{j = 0}^m \colon \text{measurable partition of } I \right\}
  ;
  \]
  the last inequality is derived 
  by $\beta h(x)(J) + (1-\beta) k(x)(J) \in [0,1]$
  from the assumption that either
  $\beta \in [0,1]$ or $h(x)(J) = k(x)(J)$ for all $x \in I$ holds.

  We next evaluate $V_2$ as follows:
  \begin{align*}
      V_2
      & \leq \left(\sup_{0 \leq j \leq m} \sum_{i = 0}^n \beta^i_j f_2 \left(\frac{\alpha^i_j}{\beta^i_j}\right)\right) \sum_{j = 0}^m \left( \beta' \mu_1(B_j) + (1-\beta') \mu_2(B_j)) + \gamma \mu_2(B_j) f_1\left(\frac{\mu_1(B_j)}{\mu_2(B_j)}\right)\right)\\
      & = \left(\sup_{x \in I} \sum_{i = 0}^n k^i_l(x) f_2 \left(\frac{h^i_l(x)}{k^i_l(x)}\right)\right)
        \left(\beta' \mu_1(I) + (1-\beta') \mu_2(I) +\gamma  \sum_{j = 0}^m
         \mu_2(B_j) f_1\left(\frac{\mu_1(B_j)}{\mu_2(B_j)}\right)
        \right)\\
       & \leq \left(\sup_{x \in I} \sum_{i = 0}^n k^i_l(x) f_2 \left(\frac{h^i_l(x)}{k^i_l(x)}\right)\right)
       \left(\beta' \mu_1(I) + (1-\beta') \mu_2(I) + \gamma {{}^{f_1}\mathsf{Div}}_I(\mu_1,\mu_2)\right)\\
      & \xrightarrow[]{~l \to \infty~}
        \left(\sup_{x \in I} \sum_{i = 0}^n k(x)(A_i) f_2 \left(\frac{h(x)(A_i)}{k(x)(A_i)}\right)\right)
        \left(\beta' \mu_1(I) + (1-\beta') \mu_2(I) +\gamma  {{}^{f_1}\mathsf{Div}}_I(\mu_1,\mu_2) \right)\\
      & \leq \sup_{x \in I}{{}^{f_2}\mathsf{Div}}_J(h(x),k(x))
        \left(\beta' \mu_1(I) + (1-\beta') \mu_2(I) + \gamma {{}^{f_1}\mathsf{Div}}_I(\mu_1,\mu_2)\right)\\
      &\leq \sup_{x \in I}{{}^{f_2}\mathsf{Div}}_J(h(x),k(x))
       \cdot \gamma {{}^{f_1}\mathsf{Div}}_I(\mu_1,\mu_2)\\
       &=\gamma {{}^{f_1}\mathsf{Div}}_I(\mu_1,\mu_2) \cdot 
       \sup_{x \in I}{{}^{f_2}\mathsf{Div}}_J(h(x),k(x)).
    \end{align*}
  Here,
  the first inequality is derived from the non-negativity of each 
  \begin{equation}\label{fdiv:comp:nonnegativity}
   (\beta' \mu_1(B_j) + (1-\beta') \mu_2(B_j)) + \gamma \mu_2(B_j) f_1\left(\frac{\mu_1(B_j)}{\mu_2(B_j)}\right);
  \end{equation}
  the first equality is given by definition of $\alpha_{j}^{i}$ and $\beta_{j}^{i}$ and the countable additivity of $\mu_1$ and $\mu_2$;
  the second inequality is given by the continuity of ${{}^{f_1}\mathsf{Div}}$ and $0 \leq \gamma$;
  the last inequality is derived 
  by $\beta' \mu_1(I) + (1-\beta') \mu_2(I) \in [0,1]$
  from the assumption that either $\beta' \in [0,1]$ or $\mu_1(I) = \mu_2(I)$ holds.
  We prove the third inequality.  
  Since $f_2$ is convex function, and sequences  
  $\{h^i_l(x)\}_{l = 0}^\infty$ and $\{k^i_l(x)\}_{l = 0}^\infty$
  are monotone increasing at each $x \in I$,
  By Jensen's inequality,  
  the sequence
  $\left\{\sum_{i = 0}^n k^i_l(x) f_2
    \left({h^i_l(x)}/{k^i_l(x)}\right)\right\}_{l = 0}^\infty$
  is monotone increasing for each $x \in I$.
  Then, the sequence $\left\{ \sup_{x \in I} \sum_{i = 0}^n k^i_l(x) f_2
  \left({h^i_l(x)}/{k^i_l(x)}\right)\right\}_{l = 0}^\infty$ of supremums
  is also monotone increasing, because each $\sum_{i = 0}^n k^i_{l + 1}(x) f_2
  \left({h^i_{l + 1}(x)}/{k^i_{l + 1}(x)}\right)$ is always greater than $\sum_{i = 0}^n k^i_l(x) f_2
  \left({h^i_l(x)}/{k^i_l(x)}\right)$. 
  Hence,
  \begin{align*}
    \lim_{l \to \infty} \sup_{x \in I} \sum_{i = 0}^n k^i_l(x) f_2 \left(\frac{h^i_l(x)}{k^i_l(x)}\right)
    &= \sup_{l \in \mathbb{N}}  \sup_{x \in I} \sum_{i = 0}^n k^i_l(x) f_2 \left(\frac{h^i_l(x)}{k^i_l(x)}\right)\\
    &= \sup_{x \in I} \sup_{l \in \mathbb{N}} \sum_{i = 0}^n k^i_l(x) f_2 \left(\frac{h^i_l(x)}{k^i_l(x)}\right)\\
    &= \sup_{x \in I} \sum_{i = 0}^n k(x)(A_i) f_2 \left(\frac{h(x)(A_i)}{k(x)(A_i)}\right)\\
    &\leq \sup_{x \in I}{{}^{f_2}\mathsf{Div}}(h(x),k(x)).
  \end{align*}

  Finally, we evaluate $V_3$
  as follows:
  \begin{align*}
   V_3
    & = \sum_{j = 0}^m \alpha (\mu_2(B_j) - \mu_1(B_j))(\sum_{i = 0}^n \alpha^i_j - \beta^i_j)\\
    & = \alpha\left(\int_I  h^i_l ~d\mu_2  - \int_I  k^i_l ~d\mu_2 + \int_I  k^i_l ~d\mu_1 - \int_I  h^i_l ~d\mu_1\right)\\
    & \xrightarrow[]{~l \to \infty~}
      \alpha\left(\int_I  h(-)(J) ~d\mu_2  - \int_I  k(-)(J) ~d\mu_2 + \int_I  k(-)(J) ~d\mu_1 - \int_I  h(-)(J) ~d\mu_1\right).
  \end{align*}
  Here, if either $\alpha = 0$ or $h(x)(J) = k(x)(J)$ for any
  $x \in I$ holds then the limit will be $0$.
  To sum up the above evaluations of $V_1,V_2,V_3$, we obtain the inequality
  (\ref{inequality:composition:f-div:target}) if we have either
  \begin{enumerate}
  \item $\mu_1(I) = \mu_2(I) = 1$ and $\forall x \in I.~h(x)(J) = k(x)(J) = 1$, or
  \item $\alpha = 0$ and $\beta,\beta \in [0,1]$.
  \end{enumerate}
  This completes the proof.
\end{appendixproof}
\begin{toappendix}
  Parameters for Proposition \ref{prop:composable:f-divergence:weight} for 
  for weight functions of $\divntv$, $\divnkl$, $\divnhd$ and $\divnchi$
  are shown in Table \ref{tab:example_parameters}.
Below, we check the conditions in Proposition \ref{prop:composable:f-divergence:weight}.
  \begin{itemize}
  \item For the weight function $f(t) = |t - 1| / 2$ of $\divntv$,
  the tuple $(\gamma,\alpha,\beta,\beta') = (0,0,1,0)$ satisfies for all $x,y,z,w \in [0,1]$,  
    we have 
    \begin{align*}
      0 & \leq w + xf({z}/{x}),\\
      xy f(zw/xy) & =| {zw} - xy | / 2 \leq  | {zw} - {wx}| + |{xw} - xy | /2 = wxf({z}/{x}) + xf(|{w}/{y}) /2.
    \end{align*}
  \item 
  For the weight function $f(t) = t\log(t) - t + 1$ of $\divnkl$,
  the tuple $(\gamma,\alpha,\beta,\beta') = (0,-1,1,1)$ satisfies for all $x,y,z,w \in [0,1]$,
  we have 
    \begin{align*}
      0 & \leq  z  + x f(z/x), \\
      \lefteqn{xy((zw/xy)\log(zw/xy) - zw/xy + 1)}\\
      &= zw\log(w/y) + zw\log(z/x) - zw + xy \\
        & = xw((z/x)\log(z/x) - z/x + 1)  + zy((w/y)\log(w/y) - w/y  + 1) - (x-z)(w-y).
    \end{align*}
  \item For the weight function $f(t) = (\sqrt{t} - 1)^2 / 2$ of $\divnhd$,
     the tuple 
    $(\gamma,\alpha,\beta,\beta') = (0,-1/4,1/2,1/2)$ satisfies for all $x,y,z,w \in [0,1]$,
    \begin{align*}
      0 & \leq  (z + x) / 2  + f(z/x), \\
      xyf(zw/xy)
        & = (zw + xy)/2 - ((x + z) - (\sqrt{x} - \sqrt{z})^2)((y + w) - (\sqrt{y} - \sqrt{w})^2)/ 4 \\
        & = (zw + xy)/2 - ((x + z) - x f(z/x))((y + w) - yf(w/y))/ 4\\
        & \leq (y + w)/ 2 \cdot x f(z/x)  +  (x + z)/2 \cdot y f(w/y)  - (x - z)(w - y)/4.
    \end{align*}
  \item For the weight function $f(t) = ({t} - 1)^2 / 2$ of $\divnchi$,
  The tuple 
  $(\gamma,\alpha,\beta,\beta') = (1,-2,2,2)$ satisfies for all $x,y,z,w \in [0,1]$,
    \begin{align*}
      0 & \leq (2z - x) + x f(z/x) = (2z - x) + ((z/x) - 1)(z - x) = z + (z^2/x),  \\
      \lefteqn{xy f(zw/xy)
      = z^2w^2/xy + xy - 2zw}\\
        &=
          (xf(z/x) + 2z - x)(yf(w/y) + 2w - y) - 2zw + xy \\
        &=(2w - y) xf(z/x) + (2z - x)yf(w/y) + xyf(z/x)f(w/y) -2(x-z)(w-y).
    \end{align*}
  \end{itemize}
\end{toappendix}

\subsection{Divergences on the Probability Monad on QBS via Monad
  Opfunctors.}
\label{subsection:divergence:QBS}

We have seen various divergences on the Giry monad $\giry$.  It would
be nice if they are transferred to the probability monad $\probqbs$ on
$\QBS$ (Section \ref{sec:qbs}). For this, we first develop a generic
method for transferring divergences on monads.

Let $(\CC,S)$ and $(\DD,T)$ be two CC-SMs. A \emph{monad opfunctor}
\cite[Section 4]{STREET1972149} is a functor $p \colon \CC \to \DD$
together with a natural transformation
$\lambda\colon p\circ S \to T\circ p$ making the following diagrams
commute:
\begin{displaymath}
  \xymatrix{
    p \rdm{p\circ \eta^S} \rrdh{\eta^T\circ p} & & p\circ S\circ S \rrh{\lambda\circ  S} \rdm{p\circ \mu^S} & T\circ p\circ S \rrh{T\circ \lambda} & T\circ T\circ p \rdh{\mu^T\circ p} \\
    p\circ S \rrm{\lambda} & T\circ p & p\circ S \rrrm{\lambda} & & T\circ p
  }
\end{displaymath}

\begin{proposition}
  \label{prop:divergence:monad_opfunctor_a}
  Let $(\CC,S),(\DD,T)$ be two CC-SMs,
  $(p\colon \CC\to\DD, \lambda\colon p\circ S\to T\circ p)$ be a monad opfunctor,
  and assume that $U^\DD \circ p = U^\CC$ holds, and basic
  endorelations $F \colon \CC \to \BRel \CC$ and
  $E \colon \DD \to \BRel \DD$ satisfy $R_{FpI} = R_{EI}$ for all
  $I \in \CC$ (we here use $U^\DD\circ p = U^\CC$).
  Then for any
  $\asgn\in\mDiv TM\qQ E$, the following doubly-indexed family of
  $\qQ$-divergences
  $\divnOpfun \asgn {p,\lambda} =\{ \divOpfun \asgn {p,\lambda} m {I} \}_{m\in M,I\in\CC}$ on $SI$ is
  an $F$-relative $M$-graded $\qQ$-divergence on $S$:
  \begin{align*}
    \divOpfun \asgn {p,\lambda} m {I} (\nu_1,\nu_2) &\triangleq \asgn^m_{pI} (\lambda_{I} \ap
                                         \nu_1,\lambda_{I} \ap \nu_2)
                                         = \asgn^m_{pI} ((U^\DD
                                         \lambda_{I})(\nu_1),(U^\DD \lambda_{I})(\nu_2)).
  \end{align*}
\end{proposition}
\begin{appendixproof}
(Proof of Proposition \ref{prop:divergence:monad_opfunctor_a})

We first show the monotonicity of $\divnOpfun \asgn {p,\lambda}$. 
Assume $m \leq m'$. From the monotonicity of the original $\asgn$, 
we obtain for each $\nu_1,\nu_2 \in U^\CC (SI))$,
\begin{align*}
\divOpfun \asgn {p,\lambda}  m {I} (\nu_1,\nu_2)
&=
\asgn^m_{pI} ((U^\DD \lambda_{I})(\nu_1),(U^\DD \lambda_{I})(\nu_2))\\
&\geq 
\asgn^{m'}_{pI} ((U^\DD \lambda_{I})(\nu_1),(U^\DD \lambda_{I})(\nu_2))\\
&= \divOpfun \asgn {p,\lambda}   {m'} {I} (\nu_1,\nu_2).
\end{align*}

Second, we show the $F$-unit-reflexivity of $\divnOpfun \asgn {p,\lambda}$.
    For $FI = (I,I,R_{FI})$, we have $EpI = (p = (pI,pI,R_{FI})$ for all $I \in \CC$.
    We can calculate for all $(x,y) \in R_F$,
    \begin{align*}
      \divOpfun \asgn {p,\lambda}   {1_M} {I} (\eta^S_{I} \ap x,\eta^S_{I} \ap y)
      &= \asg {1_M} {pI} (U^\DD \lambda_{I} \circ U^{\CC} \eta^S_{I} \circ x,U^\DD \lambda_{I} \circ U^{\CC} \eta^S_{I} \circ y)\\
      &= \asg {1_M} {pI} ((\lambda_{I} \circ p \eta^S_{I}) \ap x,(\lambda_{I} \circ p \eta^S_{I}) \ap y)\\
      &= \asg {1_M} {pI} (\eta^T_{pI} \ap x,\eta^T_{pI} \ap y) \leq 0.
    \end{align*}

Finally, we show the $\mF$-composability of $\divnOpfun \asgn {p,\lambda}$.
    For all $J \in \CC$,
    $c_1,c_2 \in U^\CC T I$, and $f_1,f_2 \colon I \to S J$ we can calculate
    \begin{align*}
      \divOpfun \asgn {p,\lambda}   {mn} {J} (f_1\kl \ap c_1, f_2\kl \ap c_2)
      &= 
        \asg  {mn} {pJ} (
        U^\DD \lambda_{J} \circ U^\DD p (f_1\kl) \circ c_1,
        U^\DD \lambda_{J} \circ U^\DD p (f_2\kl) \circ c_2)\\
      &= 
        \asg  {mn} {pJ} (
        U^\DD ((\lambda_J \circ p f_1 )\kl) \circ U^\DD \lambda_I \circ c_1,
        U^\DD ((\lambda_J \circ p f_2 )\kl) \circ U^\DD \lambda_I \circ c_2)\\
      &= 
        \asg {mn} {pJ} (
        (\lambda_J \circ p f_1 )\kl  \ap (\lambda_I \ap c_1),
        (\lambda_J \circ p f_2 )\kl  \ap (\lambda_I \ap c_2))\\
      &\leq
        \asg {m} {pI} (\lambda_I \ap c_1,\lambda_I \ap c_2)+
        \sup_{(x, y) \in EpI}
        \asg {n} {pJ} (
        (\lambda_J \circ p f_1 ) \ap x,
        (\lambda_J \circ p f_2 ) \ap x)
      \\
      &=
        \divOpfun \asgn {p,\lambda}   {m} {I} (c_1,c_2) +
        \sup_{(x, y) \in FI}
        \divOpfun \asgn {p,\lambda}   {n} {J} (f_1 \ap x,f_2 \ap y).
    \end{align*}    
    To prove the second equality, we calculate
    \begin{align*}
      U^\DD \lambda_{J} \circ U^\DD p (f_i\kl)
      &=
        U^\DD (\lambda_{J} \circ p \mu^S_J \circ p S f_i)
        =
        U^\DD (\mu^T_{pJ}\circ T\lambda_{J} \circ \lambda_{SJ} \circ p S f_i) \\
      &=
        U^\DD (\mu^T_{pJ}\circ T\lambda_{J} \circ T p f_i \circ \lambda_{I} )
        =
        U^\DD ((\lambda_{J} \circ p f_i)\kl \circ \lambda_{I} ).
    \end{align*}
This completes the proof.
\end{appendixproof}

The left adjoint $\adjL \colon \QBS \to \Meas$ of the adjunction
$\adjL \dashv \adjR \colon \Meas\arrow\QBS$ and the natural
transformation $l \colon \adjL\probqbs \Rightarrow \giry L$ defined by
$l_X([\alpha,\mu]_{\sim_X})= \mu(\alpha^{-1}(-))$ forms a monad
opfunctor from the probability monad $\probqbs$ on $\QBS$ to the Giry
monad $\giry$ on $\Meas$~\cite[Prop. 22 (3)]{HeunenKSY17}.  Through this monad
opfunctor $(\adjL,l)$, we can convert $\mEQ{}$-divergences on $\giry$
to those on $\probqbs$. This conversion can be applied to all the
statistical divergences in Table \ref{tab:divdp} and \ref{tab:divstat}.

In addition, for any standard Borel space, we can
view such converted divergences $\divnOpfun \asgn {L,l} $ as the same thing as
the original $\asgn$.  When $\Omega \in \Meas$ is standard Borel, we
have an equality $\adjL \adjR \Omega=\Omega$, and $l_{\adjR\Omega}$ is
an isomorphism. Therefore we obtain an isomorphism
$l_{\adjR\Omega}\colon \adjL \probqbs \adjR \Omega \cong \giry \adjL \adjR \Omega=
\giry\Omega$~\cite[Prop. 22 (4)]{HeunenKSY17}. A concrete description
of its inverse is
$l_{\adjR\Omega}^{-1}\ap\mu=
[\gamma',\mu(\gamma^{-1}(-))]_{\sim_{\adjR
    \Omega}}$, where $\gamma' \colon \RR \to \Omega$ and
$\gamma \colon \Omega \to \RR$ are a section-retraction pair
(i.e. $\gamma' \circ \gamma = \mathrm{id}_\Omega$)
that exists for any standard Borel $\Omega$.

\begin{therm}
\label{thm:divergence_QBS:stdBorel}
  For any $\asgn\in\mDiv\giry M\qQ\mEQ$ and standard Borel
  $\Omega \in \Meas$,
  \[ \divOpfun \asgn {L,l} m {\adjR \Omega}
    (l_{\adjR\Omega}^{-1}\ap\mu_1,l_{\adjR\Omega}^{-1}\ap\mu_2)
    =
    \asg m\Omega(\mu_1,\mu_2) \quad (\mu_1,\mu_2 \in U(\giry \Omega)).
  \]
\end{therm}

\subsection{Divergences on State Monads}

The state monad $T_S\triangleq S \Rightarrow (- \times S)$ with a
state space $S$ is used to represent programs that update the state.  We
construct divergences on $T_S$ using divergences $d_S$ on the state
space $S$ in several ways.

\subsubsection{Lipschitz Constant on States}

We first consider the state monad $T_S$ on $\Set$. We also consider a
function $d_S \colon S^2 \to [0,\infty]$ satisfying $d_S(s,s) = 0$.
The following $\qRm$-divergence $\asgn^{\mathsf{lip},d_S}_I(f_1,f_2)$
on $T_SI$ measures how much the function pair
$(\pi_2\circ f_1, \pi_2\circ f_2)$ extends the distance between two
states before updated. In short, $\asgn^{\mathsf{lip},d_S}$
measures the Lipschitz constant on state transformers.
\begin{proposition}
  \label{prop:divergence:statemonad:lipscitz}
  The family $\asgn^{\mathsf{lip},d_S}=\{\asgn^{\mathsf{lip},d_S}_I\}_{I\in\Set}$ of
  $\qRm$-divergences on $T_SI$ defined by
  \[
    \asgn^{\mathsf{lip},d_S}_I (f_1, f_2) \triangleq \sup_{s_1, s_2 \in S}
    \frac{d_S(\pi_2(f_1(s_1)), \pi_2(f_2(s_2)))}{d_S(s_1,s_2)}
    \quad 
    (f_1,f_2 \in T_S I, \text{ we suppose } 0/0 = 1)
  \]
  is a $\mT$-relative $\qRm$-divergence on $T_S$.
\end{proposition}
\begin{appendixproof}
  (Proof of Proposition \ref{prop:divergence:statemonad:lipscitz})
  
  It suffices to show $\mT$-unit reflexivity and $\mT$-composability:
  \begin{align*}
    &\asgn^{\mathsf{lip},d_S}_I (\eta_I(x), \eta_I(y))
      =
      \sup_{s', s \in S}\frac{d_S(\pi_2 (s,x), \pi_2 (s',y))}{d_S(s,s')} = \frac{d_S(s,s')}{d_S(s,s')} = 1,\\
    &\asgn^{\mathsf{lip},d_S}_J (F_1\kl (f_1),F_1\kl (f_2))\\
    &=
      \sup_{s', s \in S} \frac{d_S(\pi_2(F_1(\pi_1 f_1(s))(\pi_2 f_1(s))),\pi_2(F_2(\pi_1 f_2(s'))(\pi_2 f_2(s'))) ) }{d_S(s,s')}\\
    &= \sup_{s', s \in S}
      \frac{d_S(\pi_2 f_1(s) ,\pi_2 f_2(s') )}{d_S(s,s')}
      \cdot
      \frac{d_S(\pi_2(F_1(\pi_1 f_1(s))(\pi_2 f_1(s))),\pi_2(F_2(\pi_1 f_2(s'))(\pi_2 f_2(s'))) )  }{d_S(\pi_2 f_1(s),\pi_2 f_2(s'))}\\
    & \leq
      \sup_{s', s \in S}
       \frac{d_S(\pi_2 f_1(s),\pi_2 f_2(s'))}{d_S(s,s')}
      \cdot
      \sup_{t', t \in S}
      \frac{d_S(\pi_2 (F_1(\pi_1 f_1(s))(t)),\pi_2 (F_2(\pi_1 f_2(s'))(t')) )  }{d_S(t,t')}\\
    &\leq 
      \asgn^{\mathsf{lip},d_S}_I(f_1,f_2) \cdot
      \sup_{x, y \in I} \asgn^{\mathsf{lip},d_S}_J(F_1(x),F_2(y))
  \end{align*}
  Here $F_1,F_2 \colon I \to T_S J$ and
  $f_1,f_2 \in T_S I$.
\end{appendixproof}
For state transformers $f_1,f_2 \in T_S I$, their state-updating part
is given as functions
$\pi_2\circ f_1, \pi_2\circ f_2 \in S \Rightarrow S$.
When $f_1 = f_2 = g$, $\asgn^{\mathsf{lip},d_S}_I(g,g)$ is exactly the
Lipschitz constant of $\pi_2\circ g$.
\subsubsection{Distance between State Transformers with the Same Inputs}
Suppose that the function $d_S$ also satisfies the triangle
inequality. The following $\qRp$-divergence
$\asg {\mathsf{met},d_S} I(f_1,f_2)$ on $T_SI$ estimates the distance
between updated states after the state transformers $f_1$ and $f_2$ are
applied to the same input.
\begin{proposition}
  \label{prop:divergence:statemonad:metric}
  Suppose that the function $d_S$ also satisfy the
  triangle-inequality.  The family
  $\asgn^{\mathsf{met},d_S}=\{\asgn^{\mathsf{met},d_S}_I\}_{I\in\Set}$ of $\qRp$-divergences on
  $T_SI$ defined by:
  \begin{align*}
    \asgn^{\mathsf{met},d_S}_I (f_1, f_2)
    &\triangleq 
    \begin{cases}
      \sup_{s \in S}\! d_S(\pi_2(f_1(s)), \pi_2(f_2(s))) &
      \pi_1\circ f_1 = \pi_1\circ f_2~\text{and}\\
      &\pi_2\circ f_1,\pi_2\circ f_2 \colon \text{nonexpansive} \\
      \infty & \text{otherwise}      
    \end{cases}
  \end{align*}
  is an $\mEQ {}$-relative $\qRp$-divergence on $T_S$.
\end{proposition}
\begin{appendixproof}
  (Proof of Proposition \ref{prop:divergence:statemonad:metric})
  
  It suffices to show $\mEQ$-unit reflexivity and $\mEQ$-composability:
  \begin{align*}
    \asgn^{\mathsf{met},d_S}_I (\eta_I(x), \eta_I(x))
    &=
      \sup_{s \in S} d_S(\pi_2 (x,s), \pi_2 (x,s)) = \sup_{s \in S} d_S(s,s) = 0.\\
    \asgn^{\mathsf{met},d_S}_J (F_1\kl (f_1),F_1\kl (f_2))
    &= \sup_{s \in S} d_S(\pi_2 (F_1(\pi_1 f_1(s))(\pi_2 f_1(s))),\pi_1 (F_2(\pi_1 f_2(s))(\pi_2 f_2(s))) ) \\
    &\leq \sup_{s \in S} d_S(\pi_2 (F_1(\pi_1 f_1(s))(\pi_2 f_1(s))),\pi_2 (F_2(\pi_1 f_1(s))(\pi_2 f_1(s))) ) \\
    & \quad + \sup_{s \in S} d_S(\pi_2 (F_2(\pi_1 f_1(s))(\pi_2 f_1(s))),\pi_2 (F_2(\pi_1 f_1(s))(\pi_2 f_2(s))) )\\
    & \leq \sup_{x \in I} \asgn^{\mathsf{met},d_S}_J(F_1(x),F_2(x)) + \asgn^{\mathsf{met},d_S}_I (f_1,f_2)
  \end{align*}
  Here $F_1,F_2 \colon I \to T_S J$ and
  $f_1,f_2 \in T_S I$.  Without loss of
  generality, we may assume $\pi_1 f_1 = \pi_1 f_2$ holds and
  $\pi_2 f_1$ and $\pi_2 f_2$ are nonexpansive, and for every
  $x \in I$, $\pi_1 F_1(x) = \pi_1 F_2(x)$ holds and $\pi_2 F_1(x)$
  and $\pi_2 F_2(x)$ are nonexpansive.
\end{appendixproof}

\newcommand{\Gum}{\mathbf{Gum}}

\subsubsection{Sup-Metric on the State Monad on the Category of Generalized Ultrametric Spaces}
The category $\Gum$ of generalized ($[0,1]$-valued) ultrametric spaces\footnote{Recall that an ultrametric space $(I,d_I)$ is a set $I$
  together with a function $d_I \colon I^2 \to [0,1]$ such that
  $d_I(x,x) = 0$ and $d_I(x,z) \leq \max(d_I(x,y),d_I(y,z))$.}  and
nonexpansive functions is Cartesian closed~\cite[Section 2.2]{Rutten1996}.
We consider the state monad $T_S = S \Rightarrow (- \times S)$ on $\mathbf{Gum}$
for a fixed space $(S,d_S)\in\Gum$.
From the definition of exponential objects in $\Gum$,
$T_S(I,d_I)$ consists of the set of nonexpansive state transformers with the sup metric between them. 
In fact, the metric part of all $T_S(I,d_I)$ forms a divergence on $T_S$.

\begin{proposition}\label{prop:divergence:statemonad:metric:Gum}
  The family $\{ d_{T_S I} \colon (T_S(I,d_I))^2 \to [0,1] \}_{(I,d_I) \in \Gum} $ consisting of 
  the metric part of the spaces $T_S(I,d_I)$, given by
  \begin{align*}
    d_{T_S I} (f_1, f_2) \triangleq  \sup_{s \in S}
    \max\left(
    d_I(\pi_1(f_1 (s)),\pi_1(f_2 (s))),
    d_S(\pi_2(f_1 (s)),\pi_2 (f_2 (s)))
    \right)
  \end{align*}
  forms an $\mEQ {}$-relative $([0,1],\leq,\max,0)$-divergence on $T_S$.
\end{proposition}
\begin{appendixproof}
  (Proof of Proposition \ref{prop:divergence:statemonad:metric:Gum})
  We first show the $\mEQ {}$-unit reflexivity of $d^{T_S (-)}$. For any $s \in S$, we calculate
  \begin{align*}
    d_{T_S I} (\eta_I(x), \eta_I(x))
    &=\sup_{s \in S} \max\left( d_I(\pi_1 (x,s), \pi_1 (x,s)),  d_S(\pi_2 (x,s), \pi_2 (x,s)\right) \\
	& = \sup_{s \in S} \max(d_I(x,x),d_S(s,s)) = 0.
  \end{align*}
 We next show the $\mEQ {}$-composability of $d^{T_S (-)}$.
 For any $f_1,f_2 \in T_S (I,d_I)$ and nonexpansive functions $F_1,F_2 \colon (I,d_I) \to T_S (J,d_J)$, we compute
  \begin{align*}
  d^{T_S J} (F_1\kl (f_1),F_1\kl (f_2))
    &= \sup_{s \in S} \max\left(
    	\begin{aligned}
		 d_J(\pi_1 (F_1(\pi_1 f_1(s))(\pi_2 f_1(s))),\pi_1 (F_2(\pi_1 f_2(s))(\pi_2 f_2(s))),\\
 		 d_S(\pi_2 (F_1(\pi_1 f_1(s))(\pi_2 f_1(s))),\pi_2 (F_2(\pi_1 f_2(s))(\pi_2 f_2(s)))
 		 \end{aligned}
       \right)\allowdisplaybreaks[0]\\
  &\leq
  \sup_{s \in S}\max\left(
    	\begin{aligned}
		 d_J(\pi_1 (F_1(\pi_1 f_1(s))(\pi_2 f_1(s))),\pi_1 (F_2(\pi_1 f_1(s))(\pi_2 f_1(s))),\\
		 d_J(\pi_1 (F_2(\pi_1 f_1(s))(\pi_2 f_1(s))),\pi_1 (F_2(\pi_1 f_2(s))(\pi_2 f_2(s))),\\
 		 d_S(\pi_2 (F_1(\pi_1 f_1(s))(\pi_2 f_1(s))),\pi_2 (F_2(\pi_1 f_1(s))(\pi_2 f_1(s))),\\
 		 d_S(\pi_2 (F_2(\pi_1 f_1(s))(\pi_2 f_1(s))),\pi_2 (F_2(\pi_1 f_2(s))(\pi_2 f_2(s)))
	 	\end{aligned}
       \right)\allowdisplaybreaks[1]\\
  &=
  \sup_{s \in S}\max\left(
    	\begin{aligned}
    	 d_J(\pi_1 (F_2(\pi_1 f_1(s))(\pi_2 f_1(s))),\pi_1 (F_2(\pi_1 f_2(s))(\pi_2 f_2(s))),\\
 		 d_S(\pi_2 (F_2(\pi_1 f_1(s))(\pi_2 f_1(s))),\pi_2 (F_2(\pi_1 f_2(s))(\pi_2 f_2(s))),\\
		 d_J(\pi_1 (F_1(\pi_1 f_1(s))(\pi_2 f_1(s))),\pi_1 (F_2(\pi_1 f_1(s))(\pi_2 f_1(s))),\\
 		 d_S(\pi_2 (F_1(\pi_1 f_1(s))(\pi_2 f_1(s))),\pi_2 (F_2(\pi_1 f_1(s))(\pi_2 f_1(s)))
	 	\end{aligned}
       \right)\allowdisplaybreaks[0]\\
  &\leq
  \sup_{s \in S}\max\left(
    	\begin{aligned}
    	 		 d_I(\pi_1 (f_1(s)),\pi_1 (f_2(s))),\\
 		 d_S(\pi_2 (f_1(s)),\pi_2 (f_2(s))),\\\
		 \sup_{x \in I}\sup_{s' \in S}\max\left(
		 \begin{aligned}
		 d_J(\pi_1 (F_1(x)(s')),\pi_1 (F_2(x)(s')),\\
 		 d_S(\pi_2 (F_1(x)(s')),\pi_2 (F_2(x)(s'))
 		 \end{aligned}\right)
 		 	 	\end{aligned}
       \right)\allowdisplaybreaks[1]\\
  &=
  \max\left(
    	\begin{aligned}
    	 		\sup_{s \in S}\max( d_I(\pi_1 (f_1(s)),\pi_1 (f_2(s))),d_S(\pi_2 (f_1(s)),\pi_2 (f_2(s))) ),\\
		 \sup_{x \in I}\sup_{s' \in S}\max\left(
		 \begin{aligned}
		 d_J(\pi_1 (F_1(x)(s')),\pi_1 (F_2(x)(s')),\\
 		 d_S(\pi_2 (F_1(x)(s')),\pi_2 (F_2(x)(s'))
 		 \end{aligned}
 		 \right)
 		 	 	\end{aligned}
       \right)\allowdisplaybreaks[0]\\
 &=
  \max( d_{T_S I}(f_1,f_2),\sup_{x \in I} d^{T_S J}(F_1(x),F_2(x)).
  \end{align*}
  We note here that the nonexpansivity of  $F_2 \colon (I,d_I) \to (S,d_S) \Rightarrow (S,d_S) \times (J,d_J)$
  is equivalent to the one of its uncurrying $\overline{F_2} \colon (S,d_S) \times (I,d_I) \to (S,d_S) \times (J,d_J)$.
\end{appendixproof}
In the category $\Gum$, instead of $\mEQ$, there is another basic
endorelation $\mathsf{Dist}_0$:
\[
\mathsf{Dist}_0 (I,d_I) \triangleq  \{ (x_1,x_2)~|~d_I(x_1,x_2) =0 \}.
\]
By modifying the divergence $d_{T_S (-)}$, we obtain a
$\mathsf{Dist}_0$-relative $([0,1],\leq,\max,0)$-divergence as below:
\begin{proposition}\label{prop:divergence:statemonad:metric:Gum:2}
  The following
  forms a $\mathsf{Dist}_0 {}$-relative $([0,1],\leq,\max,0)$-divergence on $T_S$.
  \[
    \asgn^{\mathsf{Dist}_0}_{(I,d_I)} (f_1, f_2) \triangleq  \sup_{d_S(s_1,s_2) = 0} \max( d_S(\pi_1
    (f_1 (s_1)),\pi_1(f_2 (s_2))), d_I(\pi_2(f_1 (s_1)),\pi_2(f_2 (s_2)))).
  \]
\end{proposition}
\begin{appendixproof}
  (Proof of Proposition \ref{prop:divergence:statemonad:metric:Gum:2})
  We first show the $\mathsf{Dist}_0 {}$-unit reflexivity of $\asgn^{\mathsf{Dist}_0}$. 
  For $(x_1,x_2) \in \mathsf{Dist}_0 (I,d_I)$ (i.e. $d_I (x_1,x_2) = 0$), we calculate
  \begin{align*}
    \asgn^{\mathsf{Dist}_0}_{(I,d_I)} (\eta_I(x_1), \eta_I(x_2))
    &=
      \sup_{d_S(s_1,s_2) = 0} \max\left( d_I(\pi_1 (x_1,s_2), \pi_1 (x_2,s_2)),  d_S(\pi_2 (x_1,s_1), \pi_2 (x_2,s_2)\right) \\
	& = \sup_{d_S(s_1,s_2) = 0} \max(d_I(x_1,x_2),d_S(s_1,s_2)) = 0.
   \end{align*}
   
   Next, we show the $\mathsf{Dist}_0 {}$-composability of $\asgn^{\mathsf{Dist}_0}$.
   For any $f_1,f_2 \in T_S (I,d_I)$ and nonexpansive functions $F_1,F_2 \colon (I,d_I) \to T_S (J,d_J)$, we compute
 \begin{align*}
  \lefteqn{ \asgn^{\mathsf{Dist}_0}_J (F_1\kl (f_1),F_1\kl (f_2))}\\
    &= \sup_{d_S(s_1,s_2) = 0} \max\left(
    	\begin{aligned}
		 d_J(\pi_1 (F_1(\pi_1 f_1(s_1))(\pi_2 f_1(s_1))),\pi_1 (F_2(\pi_1 f_2(s_2))(\pi_2 f_2(s_2))),\\
 		 d_S(\pi_2 (F_1(\pi_1 f_1(s_1))(\pi_2 f_1(s_1))),\pi_2 (F_2(\pi_1 f_2(s_2))(\pi_2 f_2(s_2)))
 		 \end{aligned}
       \right)\allowdisplaybreaks[0]\\
  &\leq
  \sup_{d_S(s_1,s_2) = 0}\max\left(
    	\begin{aligned}
		 d_J(\pi_1 (F_1(\pi_1 f_1(s_1))(\pi_2 f_1(s_1))),\pi_1 (F_2(\pi_1 f_1(s_1))(\pi_2 f_1(s_1))),\\
		 d_J(\pi_1 (F_2(\pi_1 f_1(s_1))(\pi_2 f_1(s_1))),\pi_1 (F_2(\pi_1 f_2(s_2))(\pi_2 f_2(s_2))),\\
 		 d_S(\pi_2 (F_1(\pi_1 f_1(s_1))(\pi_2 f_1(s_1))),\pi_2 (F_2(\pi_1 f_1(s_1))(\pi_2 f_1(s_1))),\\
 		 d_S(\pi_2 (F_2(\pi_1 f_1(s_1))(\pi_2 f_1(s_1))),\pi_2 (F_2(\pi_1 f_2(s_2))(\pi_2 f_2(s_2)))
	 	\end{aligned}
       \right)\allowdisplaybreaks[1]\\
  &=
  \sup_{d_S(s_1,s_2) = 0}\max\left(
    	\begin{aligned}
 		 d_J(\pi_1 (F_2(\pi_1 f_1(s_1))(\pi_2 f_1(s_1))),\pi_1 (F_2(\pi_1 f_2(s_2))(\pi_2 f_2(s_2))),\\
 		 d_S(\pi_2 (F_2(\pi_1 f_1(s_1))(\pi_2 f_1(s_1))),\pi_2 (F_2(\pi_1 f_2(s_2))(\pi_2 f_2(s_2))),\\
		 d_J(\pi_1 (F_1(\pi_1 f_1(s_1))(\pi_2 f_1(s_1))),\pi_1 (F_2(\pi_1 f_1(s_1))(\pi_2 f_1(s_1))),\\
 		 d_S(\pi_2 (F_1(\pi_1 f_1(s_1))(\pi_2 f_1(s_1))),\pi_2 (F_2(\pi_1 f_1(s_1))(\pi_2 f_1(s_1)))
	 	\end{aligned}
       \right)\allowdisplaybreaks[0]\\
  &\leq
  \sup_{d_S(s_1,s_2) = 0}\max\left(
    	\begin{aligned}
    	 		 d_I(\pi_1 (f_1(s_1)),\pi_1 (f_2(s_2))),\\
 		 d_S(\pi_2 (f_1(s_1)),\pi_2 (f_2(s_2))),\\\
		 \sup_{(x_1,x_2) \in \mathsf{Dist}_0 (I,d_I)}\sup_{d_S(s'_1,s'_2) = 0}\max\left(
		 \begin{aligned}
		 d_J(\pi_1 (F_1(x)(s'_1)),\pi_1 (F_2(x)(s_2')),\\
 		 d_S(\pi_2 (F_1(x)(s'_1)),\pi_2 (F_2(x)(s'_2))
 		 \end{aligned}\right)
 		 	 	\end{aligned}
       \right)\allowdisplaybreaks[1]\\
  &=
  \max\left(
    	\begin{aligned}
    	 		\sup_{d_S(s'_1,s'_2) = 0}\max( d_I(\pi_1 (f_1(s_1)),\pi_1 (f_2(s_2))),d_S(\pi_2 (f_1(s_1)),\pi_2 (f_2(s_2))) ),\\
		 \sup_{(x_1,x_2) \in \mathsf{Dist}_0 (I,d_I)}\sup_{d_S(s'_1,s'_2) = 0}\max\left(
		 \begin{aligned}
		 d_J(\pi_1 (F_1(x_1)(s'_1)),\pi_1 (F_2(x_2)(s'_2)),\\
 		 d_S(\pi_2 (F_1(x_1)(s'_1)),\pi_2 (F_2(x_2)(s'_2))
 		 \end{aligned}
 		 \right)
 		 	 	\end{aligned}
       \right)\allowdisplaybreaks[0]\\
 &=
  \max( \asgn^{\mathsf{Dist}_0}_I (f_1,f_2),\sup_{(x_1,x_2) \in \mathsf{Dist}_0 (I,d_I)} \asgn^{\mathsf{Dist}_0}_J (F_1(x),F_2(x)).
  \end{align*}
This completes the proof.
\end{appendixproof}
  
\subsection{Combining Divergence with Cost}\label{sec:ccount}
In Section \ref{sec:relcost}, we have introduced a divergence on the
monad $P(\NN\times -)$ modeling nondeterministic choice and cost
counting. In this section we construct a divergence on the combination
of a general computational effect and cost counting.

Let $(\CC,T)$ be a CC-SM and $\asgn\in\mDiv T1\qQ \mEQ$ be a
divergence and $(N,1_N\colon 1\to N,(\star)\colon N\times N\to N)$ be
a monoid object in $\CC$ (for cost counting). Then the composite
$T(N \times - )$ of the monad $T$ and the monoid action monad
$N \times (-)$ again carries a monad structure.  We now define a
family
$\divnCostDiv \asgn N = \{ \divCostDiv \asgn N I\colon (U(T(N\times
I)))^2 \to \qQ \}_{I \in \CC} $ of $\qQ$-divergences by
\begin{align*}
  \divCostDiv \asgn N I (c_1, c_2)\triangleq 
  \begin{cases}
    \asg {}N (T \pi_1 \ap c_1, T \pi_1 \ap c_2) &
    \asg {}{N \times I} (c_1, c_2) \leq \asg{}{N} (T \pi_1 \ap c_1, T \pi_1 \ap c_2)
    \\
    \top_{\qQ}  & \text{otherwise}
  \end{cases}.
\end{align*}
\begin{proposition}\label{prop:divergence:cost:projection}
The family $\divnCostDiv \asgn N $ is an $\mEQ$-relative $\qQ $-divergence on $T(N\times -)$.
\end{proposition}
\begin{appendixproof}
    
  (Proof of Proposition \ref{prop:divergence:cost:projection})
  The monotonicity of $\divnCostDiv \asgn N $ is obvious since $M = 1$.
  
  We show the $\mEQ$-unit reflexivity of $\divnCostDiv \asgn N $. 
  For all $x \in UI$, we have
  \begin{alignat*}{2}
  T \pi_1 \ap (\eta^{T(N \times - )}_I \ap x)
  &=(T \pi_1 \circ \eta^{T(N \times - )}_I  )\ap x
  &&=(T \pi_1 \circ T\eta^{(N \times -)}_I  \circ \eta^{T}_I )\ap x\\
  &=(T (1_N \circ !_I ) \circ \eta^{T}_I )\ap x
  &&=(\eta^{T} \circ 1_N \circ !_I ) \ap x\\
  &= \eta^{T} \ap ((1_N \circ !_I ) \ap x)
  &&= \eta^{T} \ap (1_N \ap (!_I \ap x))\\
  &= \eta^{T} \ap 1_N
  \end{alignat*}
  Hence,
    \begin{alignat*}{2}
      \divCostDiv \asgn N I (\eta^{T(N \times - )}\ap x,\eta^{T(N \times - )}\ap x)
      &=\divCostDiv \asgn N I (\eta^{T(N \times - )}\ap x,\eta^{T(N \times - )}\ap x) \\ 
      &=\asg {}N (T \pi_1 \ap (\eta^{T(N \times - )}\ap x),T \pi_1 \ap (\eta^{T(N \times - )}\ap x)\\
      &=\asg {}N (\eta^{T} \ap 1_N,\eta^{T} \ap 1_N) \\
      &\leq 0_\qQ.
    \end{alignat*}
    
    We next show the $\mEQ$-composability of $\divnCostDiv \asgn N $.
    For any
    $f \colon I \to T(N \times I)$, we define
    $h_f \colon N \times I \to T(N)$ by
    $h_f = T(\star) \circ \theta_{N,N} \circ (N \times (T \pi_1 \circ
    f))$.  Then, we have
    $T \pi_1 \ap f^{\sharp(T(N \times I))}\ap \nu = h_f^{\sharp T}\ap \nu$
    for any $\nu \in U(T(N \times I))$.
    First, for all $m,n \in UN$, we have 
   \begin{alignat*}{2}
     (T(\star) \circ (\eta_{N \times N})_n) \ap m 
     &= T(\star) \ap (\eta_{N \times N} \ap \langle n,m\rangle)
     &&= (T(\star) \circ \eta_{N \times N}) \ap \langle n,m\rangle)\\
      &= (\eta_{N}\circ \star) \ap \langle n,m\rangle)
      &&= \eta_{N} \ap (\star \ap \langle n,m\rangle)\\
      &= \eta_{N} \ap (\star_n \ap m)
     &&= (\eta_{N} \circ \star_n ) \ap m.
     \end{alignat*}  
    From this and the equality (\ref{eq:stun}), we can calculate as follows:
    \begin{alignat*}{2}
    \lefteqn{h_{f} \ap \langle n,i\rangle}\\
    &= (T(\star) \circ \theta_{N,N} \circ (N \times (T \pi_1 \circ
    f))) \ap \langle n,i\rangle
    &&= (T(\star) \circ \theta_{N,N}) \ap ((N \times (T \pi_1 \circ
    f)) \ap \langle n,i\rangle )\\
    &= (T(\star) \circ \theta_{N,N}) \ap (U(N \times (T \pi_1 \circ
    f)) (n,i))
    &&= (T(\star) \circ \theta_{N,N}) \ap ((U(N) \times U(T \pi_1 \circ
    f)) (n,i)) \\
    &= (T(\star) \circ \theta_{N,N}) \ap \langle U(N)(n),U(T \pi_1 \circ
    f)(i) \rangle
    &&= (T(\star) \circ \theta_{N,N}) \ap \langle n,(T \pi_1 \circ
    f)\ap i \rangle\\
    &= T(\star) \ap (\theta_{N,N} \ap \langle n,(T \pi_1 \circ
    f)\ap i \rangle)
    &&= T(\star) \ap ((\theta_{N,N})_n \ap ((T \pi_1 \circ
    f)\ap i))\\
    &= T(\star) \ap (((\eta_{N \times N})_n) \kl \ap ((T \pi_1 \circ
    f)\ap i))
    &&= (T(\star) \circ ((\eta_{N \times N})_n)\kl) \ap ((T \pi_1 \circ
    f)\ap i)\\
    &=
    (T(\star) \circ (\eta_{N \times N})_n)\kl \ap ((T \pi_1 \circ
    f)\ap i)
    &&= (\eta_N \circ (\star_n) )\kl \ap ((T \pi_1 \circ
    f)\ap i).
    \end{alignat*}
    From the assumption $\asg {}{N \times I} (c_1, c_2) \leq \divCostDiv \asgn N I (c_1, c_2)$, the $\mEQ$-unit-reflexivity and $\mEQ$-composability of the original divergence $\asgn$, we obtain the $\mEQ$-composability of $\divnCostDiv \asgn N $ as follows:
    \begin{align*}
      \lefteqn{\divCostDiv \asgn N J (f_1^{\sharp(T(N \times I))} \ap c_1,f_2^{\sharp(T(N \times I))} \ap c_2)}\\
      &= 
        \asg {}N(T \pi_1 \ap f_1^{\sharp(T(N \times I))} \ap c_1,T \pi_1 \ap f_2^{\sharp(T(N \times I))} \ap c_2) \\
      &= 
        \asg {}N( h_{f_1}^{\sharp T} \ap c_1, h_{f_2}^{\sharp T} \ap c_2)\\
      &\leq 
        \asg {}{N \times I} (c_1,c_2) +
        \sup_{\langle n,i\rangle \in U(N \times I)}
        \asg {}{N} (h_{f_1} \ap \langle n,i\rangle, h_{f_2} \ap \langle n,i\rangle)\\
      &=
        \asg {}{N \times I} (c_1,c_2)\\
      &\qquad +
        \sup_{\langle n,i\rangle \in U(N \times I)}
        \asg {}{N}((\eta_N \circ (\star)_n)^{\sharp T} \ap  ((T \pi_1 \circ f) \ap i),(\eta_N \circ (\star)_n)^{\sharp T} \ap  ((T \pi_1 \circ f) \ap i) )\\
      &
\leq
        \asg {}{N \times I} (c_1,c_2)\\
      &\qquad +
        \sup_{\langle n,i\rangle \in U(N \times I)}
        \left(
        \begin{aligned}
        &\asg {}{N} ((T \pi_1 \circ f) \ap i,(T \pi_1 \circ f) \ap i)\\
        &\qquad + \sup_{m \in UN} 
        \asg{}{N}(\eta_N \circ (\star)_n) \ap m,(\eta_N \circ (\star)_n) \ap m)
        \end{aligned}
        \right)
        \\
      &\leq \asg {}{N \times I} (c_1,c_2)
      +
        \sup_{\langle n,i\rangle \in U(N \times I)}
        \asg {}{N} ((T \pi_1 \circ f) \ap i,(T \pi_1 \circ f) \ap i)\\
      &=
        \asg {}{N \times I} (c_1,c_2) +
        \sup_{i \in UI}
        \asg {}{N} ((T \pi_1 \circ f_1) \ap i, (T \pi_1 \circ f_2) \ap i)\\
      &\leq 
        \divCostDiv \asgn N I (c_1,c_2)
        +
        \sup_{i \in UI}
        \divCostDiv \asgn N J (f_1 \ap i,f_2 \ap i).
    \end{align*}
This completes the proof.
\end{appendixproof}
For example, the divergence $\divnCostDiv {\divnkl} {\RR}$ on the composite monad
$\giry(\RR \times -)$ on $\Meas$ describes Kullback-Leibler divergence
between distributions of costs in the probabilistic computations with
real-valued costs.  Intuitively, the side condition
$\divnkl_{\RR \times I} (\mu_1, \mu_2) \leq \divnkl_{\RR} (\giry \pi_1
\ap \mu_1, \giry \pi_1 \ap \mu_2)$ in the definition of $\divnCostDiv {\divnkl} {\RR}$
means that the difference between $\mu_1$ and $\mu_2$ lies only in the
costs.


\subsection{Preorders on Monads}
To explore the generality of our framework, we look at the case where
the divergence domain is $\qB=(\{0\ge 1\},1,\times)$; here $\times$ is
the numerical multiplication.
We identify an indexed family
$\asgn = \{\asg {} I \colon (U(TI))^2 \to \qB \}_{I \in \CC}$ of
$\qB$-divergences and a family of adjacency relations
$\tilde \asgn (1) I \triangleq \{ (c_1, c_2)~|~ \asg {}I (c_1, c_2)
\leq 1 \}$.

We point out a connection between $\mEQ$-relative $\qB$-divergences and
{\em preorders on monads} studied in
\mcite{preordermonad,DBLP:journals/entcs/Sato14}. A preorder on a
monad $T$ on $\Set$ assigns a preorder $\sqsubseteq_I$ on $TI$ for
each $I\in\Set$, and this assignment satisfies:
\begin{description}
\item[Substitutivity] For any function $f\colon I\to TJ$ and $c_1,c_2\in TI$,
  $c_1\sqsubseteq _Ic_2$ implies $f\kl (c_1)\sqsubseteq_J f\kl (c_2)$.
\item[Congruence] For any function $f_1,f_2\colon I\to TJ$, if
  $f_1(x)\sqsubseteq_Jf_2(x)$ holds for any $x\in I$, then
  $f_1\kl (c)\sqsubseteq_Jf_2\kl (c)$ holds for any $c\in TI$.
\end{description}
\begin{proposition}\label{pp:preorder_on_monad_is_divergence}
  A preorder on a monad $T$ on $\Set$ bijectively corresponds to an
  $\mEQ$-relative $\qB$-divergence $\asgn$ on $T$ such that each
  $\tilde\asgn(1)I$ is a preorder.
\end{proposition}
\begin{appendixproof}
  (Proof of Proposition \ref{pp:preorder_on_monad_is_divergence})

  We consider a preorder $\sqsubseteq$ on a monad $T$.  We define the
  $\qB$-divergence $\asg \sqsubseteq {}$ on $TI$ by
  \[
    \asg \sqsubseteq I (c_1,c_2) \triangleq
    \begin{cases}
      0 & c_1 \not\sqsubseteq_I c_2 \\
      1 & c_1 \sqsubseteq_I c_2
    \end{cases}
  \]
  Each $\tilde\asgn(1)I$ is a preorder because 
  $\tilde\asgn(1)I = {\sqsubseteq_I}$ holds for each $I$.

  The $\mEQ$-unit reflexivity of $\asg \sqsubseteq {}$ is derived from the reflexivity of $\sqsubseteq$. For all set $I$ and $c \in TI$,
  \[
  (\asg \sqsubseteq I (c,c) \leq 1) \iff (c \sqsubseteq_I c).
  \]
  Since $\sqsubseteq$ is a preorder on $T$,
  for all set $I,J$, $c_1,c_2 \in TI$ and $f,g \colon I \to TJ$,
  \begin{align*}
  &(\asg \sqsubseteq I (c_1, c_2) \times \sup_{x \in I}\asg \sqsubseteq J (f(x), g(x)) ) = 1\\
  &\iff (\asg \sqsubseteq I (c_1, c_2) = 1) \land (\sup_{x \in I}\asg \sqsubseteq J (f(x), g(x)) = 1)\\
  &\iff (c_1 \sqsubseteq_I c_2) \land (\forall{x \in I}.~ f(x) \sqsubseteq_J g(x))\\
  &\implies (f\kl (c_1) \sqsubseteq_J f\kl (c_2)) \land (f\kl (c_2) \sqsubseteq_J g\kl (c_2))\\
  &\implies (f\kl (c_1) \sqsubseteq_J g\kl (c_2))\\
  &\iff(\asg \sqsubseteq J (f\kl (c_1), g\kl (c_2)) = 1)
  \end{align*}
  Hence, we have the $\mEQ$-composability
  \begin{align*}
    \asg \sqsubseteq J (f\kl (c_1), g\kl (c_2))
    \leq \asg \sqsubseteq I (c_1, c_2) \times \sup_{x \in I} \asg \sqsubseteq J (f(x), g(x)).
  \end{align*}

  Conversely, we consider an $\mEQ$-relative $\qB$-divergence $\asgn$ on
  $T$ such that each $\tilde \asgn  (1) I$ is a preorder.
  We show that the family $\sqsubseteq^\asgn = \{ \sqsubseteq^\asgn_I \}_{I \in \Set}$
  defined by ${\sqsubseteq^\asgn_I} \triangleq \tilde \asgn (1) I$ forms a preorder on monad $T$.
  
  Each component $\sqsubseteq^\asgn_I$ of $\sqsubseteq^\asgn$ at set $I$ is a preorder on the set $TI$.
  We here note that the divergence $\asgn$ must be reflexive (i.e. $\asg {} I(c,c) \leq 1$ for all $I \in \Set,c \in TI$)
  because of the reflexivity of $\sqsubseteq^\asgn_I$:
  \[
    (\asg {} I(c,c) \leq 1) \iff (c \sqsubseteq^\asgn_I c), \quad \text{ for all } I\in \Set, c \in TI.
  \]
  From the reflexivity and $\mEQ$-composability of $\asgn$, we have for all
  $c_1,c_2,c \in TI$ and $f,g \colon I \to TJ$, 
  \begin{align}
  \label{divergence_premonad_subst}
    \fa{c_1,c_2 \in TI,f\colon I \to TJ} &\asg {} J (f\kl (c_1), f\kl (c_2)) \leq \asg {} I (c_1, c_2),\\
\label{divergence_premonad_congr}
    \fa{c \in TI, f,g \colon I \to TJ} &\asg {} J (f\kl (c), g\kl (c)) \leq \sup_{x \in I} \asg {} J (f(x),
    g(x)).
  \end{align}
  They are equivalent to 
  the substitutivity and congruence of $\sqsubseteq^\asgn$
  respectively: 
  \begin{align*}
    (\ref{divergence_premonad_subst})  &\iff \fa{c_1,c_2 \in TI,f\colon I \to TJ} (c_1 \sqsubseteq^\asgn_I c_2 \implies f\kl (c_1) \sqsubseteq^\asgn_J f\kl (c_2)),\\
    (\ref{divergence_premonad_congr}) &\iff \fa{c \in TI, f,g \colon I \to TJ} (\forall{x \in I}.~ f(x) \sqsubseteq^\asgn_J g(x) \implies f\kl
    (c) \sqsubseteq^\asgn_J g\kl (c)).
  \end{align*}
  Finally, the above conversions $\asgn^{(-)}$ and
  $\sqsubseteq^{(-)}$ are mutually inverse:
  \begin{align*}
  \asg {\sqsubseteq^{\asgn'}} I (c_1,c_2) \leq 1
    \iff c_1 \sqsubseteq^{\asgn'}_I c_2
   \iff \asgn'_I(c_1,c_2) \leq 1,\\
  c_1 \sqsubseteq^{\asgn^{\sqsubseteq'}}_I c_2
    \iff \asgn^{\sqsubseteq'}_I(c_1,c_2) \leq 1
   \iff c_1 \sqsubseteq'_I c_2.
  \end{align*} 
  This completes the proof.
\end{appendixproof}
For a preorder $\sqsubseteq$ on a monad $T$ on $\Set$, by
$\asg {\sqsubseteq} {}$ we mean the divergence corresponding to
$\sqsubseteq$ by Proposition \ref{pp:preorder_on_monad_is_divergence}
(in fact, we have $\widetilde{\asg {\sqsubseteq} {}}(1)I = {\sqsubseteq_I}$ for all set $I$).


\section{Properties of Divergences on Monads}

\subsection{Divergences on Monads as Structures in $\Div\qQ\CC$}

In this section we examine divergences on monads from the view point
of monoidal structure of $\Div\qQ\CC$. For any CC $\CC$, the category
$\Div\qQ\CC$ has a symmetric monoidal structure, whose unit and tensor
product are given by
\begin{align*}
  \unit&\triangleq(1,\lam{(x_1,x_2)}0),\\
  (I,d)\otimes(J,e)&\triangleq(I\times J,\lam{(\langle x_1,y_1\rangle,\langle x_2,y_2\rangle)}d(x_1,x_2)+e(y_1,y_2)).
\end{align*}
The coherence isomorphisms of this symmetric monoidal structure are
inherited from the Cartesian monoidal structure on $\CC$. Moreover,
$\Divf\qQ\CC$ becomes a {\em symmetric strict monoidal functor} of
type $(\Div\qQ\CC,\unit,\ox)\to (\CC,1,(\times))$.

\subsubsection{Enrichments of Kleisli Categories Induced by Divergences}

Let $(\CC,T)$ be a CC-SM. We first show that a (non-graded) divergence
on a monad $T$ {\em attaches} a $\Div\qQ\Set$-{\em enrichment} on the
Kleisli category $\CC_T$ of $T$. What we mean by attaching an
enrichment to an ordinary category is formulated as follows.
\begin{definition}
  A {\tmem{$\Div\qQ\Set$-enrichment}} of a category $\DD$ is a family
  $\{ d_{I, J}:\DD(I,J)^2\to\qQ \}_{I, J \in \DD}$ of
  $\qQ$-divergences on the homset $\DD (I, J)$ such that the following
  inequalities hold:
  \begin{align}
    & d_{I, I} (\id_I, \id_I) \leq 0, \label{eq:unit} \\
    & d_{I, K} (g_1 \circ f_1, g_2 \circ f_2) \leq d_{J, K} (g_1, g_2) + d_{I, J} (f_1, f_2). \label{eq:comp}
  \end{align}
\end{definition}
Such an enrichment determines a $\Div\qQ\Set$-enriched category
$\DD^d$, whose object collection and homobjects are given by
\begin{displaymath}
  \Obj{\DD^d}\triangleq \Obj{\DD},\quad
  \DD^d(I,J)\triangleq(\DD (I, J),d_{I,J}).
\end{displaymath}
The identity and composition morphisms of $\DD^d$:
\begin{displaymath}
  j_I:\unit\to\DD^d(I,I),\quad
  m_{I,J,K}:\DD^d(J,K)\ox\DD^d(I,J)\to\DD^d(I,K)
\end{displaymath}
are inherited from $\DD$; they are guaranteed to be nonexpansive by
the conditions \eqref{eq:unit} and \eqref{eq:comp}. The
change of base of enrichment of $\DD^d$ by the  symmetric strict
monoidal functor $\Divf\qQ\DD$ coincides with $\DD$. \footnote{The
  underlying category of $\DD^d$ \cite[Section 1.3]{kelly} does {\em
    not} coincide with $\DD$.  }

We relate conditions \eqref{eq:unit} and \eqref{eq:comp} with the unit
reflexivity and composability conditions in the definition of divergence on monad
(Definition \ref{def:div}).

\begin{therm}\label{thm:divergence_is_enrichment}
  Let $(\CC, T)$ be a CC-SM, $E:\CC\to\brelc$ be a basic endorelation
  such that $R_{E 1} \neq \emptyset$ \footnote{
    $R_{E 1} = \emptyset$ happens if and only if $R_{E I} = \emptyset$ for
    any $I \in \CC$. Therefore nontrivial basic endorelations always
    satisfy $R_{E 1} \neq \emptyset$.  } , $\qQ$ be a divergence
  domain and $\asgn=\{\asg{}I:(U(TI))^2\to\qQ\}_{I\in\CC}$ be a family
  of $\qQ$-divergences on $TI$. Define a family
  $d=\{ d_{I, J} \}_{I, J \in \CC}$ of $\qQ$-divergences on the homset
  $\CC_T (I, J)$ of the Kleisli category $\CC_T$ by
  \begin{equation}
    \label{eq:enrich}
    d_{I, J} (f_1, f_2) \triangleq \sup_{(x_1, x_2) \in E I} \asg{}J (f_1
    \ap x_1, f_2 \ap x_2) .
  \end{equation}
  Then $d$ is a $\Div\qQ\Set$-enrichment of $\CC_T$ if and only if
  $\asgn$ is an $E$-relative $\qQ$-divergence on $T$.
\end{therm}
\begin{appendixproof}
  (Proof of Theorem \ref{thm:divergence_is_enrichment})
  First, it is easy to see that the inequality (\ref{eq:unit}) is equivalent to
  $\asgn$ satisfying $E$-unit reflexivity.
  
  We next show that the inequality (\ref{eq:comp}) is equivalent to $\asgn$
  satisfying $E$-composability.
  
  (only if) Since $U 1 = \{ \id_1 \}$, we have
  $R_{E 1} = \{ (\id_1, \id_1) \}$. Therefore it holds
  $d_{1, J} (c_1, c_2) = \asg{}J (c_1, c_2)$. By letting $I = 1$ in
  the inequality (\ref{eq:comp}), we obtain the $E$-composability:
  \begin{eqnarray*}
    &  & d_{1, K} (f_1 \circ_{\CC_T} c_1, f_2 \circ_{\CC_T}
         c_2) \leq d_{J, K} (f_1, f_2) + d_{1, J} (c_1, c_2)\\
    & \iff & \asg{}K (f_1\kl \circ c_1, f_2\kl \circ c_2)
             \leq \sup_{(x_1, x_2) \in E I} \asg{}K (f_1 \ap x_1, f_2 \ap
             x_2) + \asg{}J (c_1, c_2).
  \end{eqnarray*}
  (if) From the $E$-composability, for any $f_1, f_2 : I \to T J$ and
  $g_1, g_2 : J \to T K$ and $(x_1, x_2) \in E I$, we have
  \begin{displaymath}
    \asg{}K (g_1\kl \ap (f_1 \ap x_1), g_2\kl \ap (f_2
    \ap x_2)) \leq d_{J, K} (g_1, g_2) + \asg{}J (f_1 \ap x_1,
    f_2 \ap x_2) .
  \end{displaymath}
  Next, for any $(x_1, x_2) \in E I$, we have $\asg{}J (f_1 \ap x_1, f_2
  \ap x_2) \leq d_{I, J} (f_1, f_2)$. Thus by monotonicity of $(+)$
  we have
  \begin{displaymath}
    \asg{}K (g_1\kl \ap f_1 \ap x_1, g_2\kl \ap f_2 \ap
    x_2) \leq d_{J, K} (g_1, g_2) + d_{I, J} (f_1, f_2) .
  \end{displaymath}
  By discharging $(x_1, x_2) \in E I$, we conclude
  \begin{displaymath}
    d_{I, K} (g_1\kl \circ f_1, g_2\kl \circ f_2) \leq d_{J, K} (g_1, g_2) + d_{I, J} (f_1, f_2) .
  \end{displaymath}
\end{appendixproof}

\subsubsection{Internalizing Divergences as Structures in $\Div\qQ\CC$}
\label{sec:int}
One might wonder how the $\qQ$-divergence \eqref{eq:enrich} given to
each homset of $\CC_T$ arises. Under a strengthened assumption, we
derive it from the {\em closed structure} with respect to the monoidal
product of $\Div\qQ\CC$. This allows us to {\em internalize}
divergences on monads as structures in $\Div\qQ\CC$.

Let $(\CC,T)$ be a CCC-SM and $\qQ$ be a divergence domain whose
monoid operation $(+)$ preserves the largest element $\top\in\qQ$, that is,
$x+\top=\top$. A consequence of this strengthened assumption is the
following:
\begin{lemma}\label{lem:closed}
  Let $(I,d)\in\Div\qQ\CC$ be an object such that $d(x_1,x_2)$ takes
  only values in $\{0,\top\}\subseteq\qQ$. Then the functor
  $(-)\ox(I,d):\Div\qQ\CC\to\Div\qQ\CC$ has a right adjoint, which we
  denote by $(I,d)\multimap(-)$. Moreover, $\Divf\qQ\CC$ is a map
  of adjunction of type:
  \begin{displaymath}
    \Divfn\qQ\CC:((-)\ox(I,d)\dashv(I,d)\multimap(-))\to
    ((-)\times I\dashv I\Arrow(-)).
  \end{displaymath}
\end{lemma}
The proof of this lemma exhibits that the $\qQ$-divergence $h$
associated to the internal hom object $(I,d)\multimap (J,e)$ measures
the divergence between $f_1,f_2\in U(I\Arrow J)$ by
\begin{displaymath}
  h(f_1,f_2)=
  \sup_{x_1,x_2\in UI,d(x_1,x_2)=0}e(\unm{f_1}\bul x_1,\unm{f_2}\bul x_2),
\end{displaymath}
which almost coincides with the sup part of \eqref{eq:enrich}; here
$\unm{-}:U(I\Arrow J)\to\CC(I,J)$ is the bijection given in Section
\ref{sec:setting}. We use this coincidence to characterize the
unit-reflexivity and composability conditions in the definition of
divergence on monad (Definition \ref{def:div}). First, we define the {\em
  internal Kleisli extension morphism}
$kl_{I,J}:TI\times(I\Arrow TJ)\arrow TJ$ by
\begin{equation}
  \label{eq:intkl}
  kl_{I,J}\triangleq
  \xymatrix{
    TI\times(I\Arrow TJ)
    \rrh{\langle\pi_2,\pi_1\rangle}
    &
    (I\Arrow TJ)\times TI
    \rrh{\theta_{I\Arrow TJ,I}}
    &
    T((I\Arrow TJ)\times I)
    \rrh{\ev^\#}
    &
    TJ
  }.
\end{equation}
Next, for a basic endorelation $E:\CC\to\brelc$, we define the functor
$E':\CC\to\Div\qQ\CC$ by
\begin{displaymath}
  E'I\triangleq(I,d_{E'I}),\quad
  E'f\triangleq f,\quad\text{where}\quad
  d_{E'I}(x_1,x_2)\triangleq
  \begin{choice}
    0 & (x_1,x_2)\in E \\
    \infty & (x_1,x_2)\not\in E.
  \end{choice}
\end{displaymath}
\begin{therm}\label{th:nechar}
  Let $(\CC,T)$ be a CCC-SM, $(M,\le,1,(\cdot))$ be a grading 
  monoid, $\qQ$ be a divergence domain whose monoid operation
  $(+)$ satisfies $x+\top=\top$, and $E:\CC\to\brelc$ be a basic
  endorelation.  Let $\asgn=\{\asg m I\}_{m\in M,I\in\CC}$ be a
  doubly-indexed family of $\qQ$-divergences on $TI$, regarded as
  $\Div\qQ\CC$-objects.  Then
  \begin{enumerate}
  \item $\asgn$ satisfies the $E$-unit reflexivity condition if and
    only if for any $I\in\CC$, the following nonexpansivity holds on
    the global element $\nm{\eta_I}:1\to I\Arrow TI$ corresponding to
    the monad unit:
    \begin{displaymath}
      \nm{\eta_I}\in\Div\qQ\CC (\unit, E' I\multimap \asg 1 I).
    \end{displaymath}
  \item \label{comp} $\asgn$ satisfies the $E$-composablity condition
    if and only if for any $I,J\in\CC$ and $m,n\in M$, the following
    nonexpansivity holds on the internal Kleisli extension morphism
    $kl_{I,J}:TI\times(I\Arrow TJ)\arrow TJ$:
    \begin{displaymath}
      kl_{I,J}\in\Div\qQ\CC(\asg m I\ox (E' I\multimap\asg n J), \asg {m\cdot n} J).
    \end{displaymath}
  \end{enumerate}
\end{therm}
\cite{DBLP:conf/lics/AmorimGHK19} formalized families of composable
divergences as {\em parameterized assignment} in {\em weakly closed
  monoidal refinement}. Roughly speaking, they adopted the equivalence
\eqref{comp} of Theorem \ref{th:nechar} as the definition of
parameterized assignment. However, divergence on monads and
parameterized assignments are built on slightly different categorical
foundations, and their generalities are incomparable. Notable
differences from parameterized assignment are: 1) divergences on
monads are defined in relative to basic endorelations, and 2) the
underlying category of divergences on monads is any CCs, while
parameterized assignments requires closed structure on their
underlying category. In this sense divergences on monads are a mild
generalization of parameterized assignments.

\subsubsection{Divergences on Monads and Divergence Liftings of Monads}

\newcommand{\DtoL}[1]{[#1]}
\newcommand{\LtoD}[1]{\langle #1\rangle}

\newcommand{\SGDLift}[4]{\Bf{SGDLift}(#1,#4,#2,#3)}

We next relate graded divergences on monads and monad-like structures
on the category $\Div\qQ\CC$ of $\qQ$-divergences on
$\CC$-objects. What we mean by monad-like structures is {\tmem{graded
    divergence liftings}} of monads on $\CC$, which we introduce
below. It is a graded monad on $\Div\qQ\CC$
\mcite{DBLP:conf/popl/Katsumata14} whose unit and multiplication are
inherited from a monad on $\CC$.

\begin{definition}
  Let $(\CC, T)$ be a CC-SM, $M$ be a grading monoid and $\qQ$ be a
  divergence domain. An $M$-graded $\qQ$-divergence lifting of $T$ is
  an mapping
  $\dot{T} : M \times\Obj{\Div\qQ\CC }\rightarrow\Obj{\Div\qQ\CC}$
  such that (below $V$ stands for the forgetful functor $\Divf\qQ\CC$)
  \begin{enumerate}
  \item $V (\dot{T} m X) = T (V X)$
    
  \item $m \leq n$ implies $\dot{T} m X \leq \dot{T} n X$
    
  \item $\eta_{V X} \in \Div\qQ\CC (X, \dot{T} 1 X)$
    
  \item
    $\mu_{V X} : \Div\qQ\CC (\dot{T} m (\dot{T} n X), \dot{T} (m \cdot
    n) X)$.
  \end{enumerate}
  Let $E : \CC \rightarrow \brelc$ be a basic endorelation. We say
  that an $M$-graded $\qQ$-divergence lifting $\dot{T}$ of $T$ is
  {\tmem{$E$-strong}} if the strength $\theta$ of $T$ satisfies
  \[ \theta_{V X, J} \in \Div\qQ\CC (X \otimes \dot{T} m (E' J),
    \dot{T} m (X \otimes E' J)) . \] We write $\SGDLift TM\qQ E$ for
  the collection of $E$-strong $M$-graded $\qQ$-divergence liftings of
  $T$. We introduce a partial order $\preceq$ on $\SGDLift TM\qQ E$ by
  \[ \dot{T} \preceq \dot{S} \iff
    \fa{m \in M, X \in \Div\qQ\CC, c_1,c_2\in U(T(V X))}
    d_{\dot{T} m X}(c_1,c_2) \geq d_{\dot{S} m X}(c_1,c_2). \]
\end{definition}

We will later see a similar concept of {\tmem{strong graded relational
    lifting}} of monad in Definition
\ref{def:graded_strong_lifting}. Divergence lifting and relational
lifting are actually instances of a common general definition of
{\tmem{strong graded lifting of monad}} \mcite{DBLP:conf/popl/Katsumata14},
but in this paper we omit this general definition.

The following theorem relates that every divergence can be expressed as the
composite of a graded divergence lifting and the divergence corresponding to a
basic endorelation.

\begin{therm}\label{thm:relative}
  Let $(\CC, T)$ be a CC-SM, $M$ be a grading monoid, $\qQ$
  be a divergence domain and $E : \CC \rightarrow \brelc$ be a basic
  endorelation. For any $\asgn \in \mDiv TM\qQ E$, define a mapping
  $\DtoL{\asgn}:M\times\Obj{\Div\qQ\CC}\to\Obj{\Div\qQ\CC}$ by, for $X=(I, d)$,
  $ \DtoL{\asgn} m X \triangleq (T I, d_{\DtoL{\asgn} m X}) $ where
  \[
    d_{\DtoL{\asgn} m X} (c_1, c_2) \triangleq \sup_{J \in \CC, n \in M, f
      \in \Div\qQ\CC (X, \asgn^n_J)} \asgn^{m \cdot n}_J (f\kl \ap
    c_1, f\kl \ap c_2) .
  \]
  Then $\DtoL{\asgn}$ is an $M$-graded $\qQ$-divergence lifting $\dot{T}$
  such that $\asgn^m_I = \DtoL{\asgn} m (E' I)$.
\end{therm}
\begin{appendixproof}
  (Proof of Theorem \ref{thm:relative})
  $\DtoL{\asgn}$ is a graded variant of {\tmem{codensity lifting}}
  performed along the fibration $\Divf\qQ\CC$
  (\cite{DBLP:journals/lmcs/KatsumataSU18}; see also Definition
  \ref{def:graded_strong_lifting}). Proving that it is a graded
  lifting of $T$ is routine. We show $\asgn^m_I = \DtoL{\asgn} m (E'
  I)$. The direction $\DtoL{\asgn} m (E' I) \leq \asgn^m_I$ is easy. We
  show the converse. From the composability of $\asgn$, for any
  $c_1, c_2 \in U (T I)$, $J \in \CC$, $n \in M$ and
  $f \in \Div\qQ\CC (E' I, \asgn^n_J)$, we have
  \[ \asgn^{m \cdot n}_J (f\kl \ap c_1, f\kl \ap c_2) \leq
     \asgn^m_I (c_1, c_2) + \sup_{(x_1, x_2) \in E I} \asgn^n_J (f \ap
     x_1, f \ap x_2) . \]
  Next, the nonexpansivity of $f$ is equivalent to
  \[ \sup_{(x_1, x_2) \in E I} \asgn^n_J (f \ap x_1, f \ap x_2) \le 0.
  \]
  Therefore we conclude $\asgn^{m \cdot n}_J (f\kl \ap c_1, f\kl
  \ap c_2) \leq \asgn^m_I (c_1, c_2)$. By discharging $J, n, f$, we
  conclude the inequality $\DtoL{\asgn} m (E I) \leq \asgn^m_I$.
\end{appendixproof}

When $M = 1$, Theorem \ref{thm:relative} implies that the assignment
$I\mapsto \asgn_I$ extends to the \tmem{$E'$-relative monad}
$\DtoL{\asgn}\circ E':\CC\to\Div\qQ\CC$ in the sense of \cite{DBLP:journals/corr/AltenkirchCU14}.

When we strengthen the assumptions on $(\CC, T)$ and $\qQ$ as done in Section
\ref{sec:int}, we obtain a sharper correspondence between divergences on monads and
strong graded divergence liftings of monads.

\begin{therm}\label{thm:liftdivadj}
  Let $(\CC, T)$ be a CCC-SM, $M$ be a grading monoid, $\qQ$ be a
  divergence domain such that $(+)$ satisfies $x+\top=\top$ and
  $E : \CC \rightarrow \brelc$ be a basic endorelation. Then there
  exists an adjunction between partial orders:
  \[
    \xymatrix@C=2cm{
      (\SGDLift TM\qQ E, \preceq)
       \adjunction{r}{\LtoD-}{\DtoL-}
      & (\mDiv TM\qQ E, \preceq)
    }
  \] where $\LtoD{\dot{T}} m I \triangleq \dot{T} m (E' I)$
\end{therm}
\begin{appendixproof}
  (Proof of Theorem \ref{thm:liftdivadj})
  Let $\asgn \in \mDiv TM\qQ E$. We have already shown that
  $\DtoL{\asgn}$ is an $M$-graded $\qQ$-divergence lifting of $T$. We show
  that $\DtoL{\asgn}$ is $E$-strong (this proof does not need the closedness of
  $\CC$). Let $X \triangleq (I, d) \in \Div\qQ\CC$ and
  $J \in \CC$ be objects. We first rewrite the goal:
  \begin{align*}
    \lefteqn{\theta \in \Div\qQ\CC (X \otimes \DtoL{\asgn} m (E' J), \DtoL{\asgn} m (X \otimes E' J))}\\
    & \iff
    \left(
      \begin{aligned}
      	&\fa{ x_1, x_2 \in U I, c_1, c_2 \in U (T J)}\\
       	&\qquad d_{\DtoL{\asgn} m (X \otimes E' J)} (\theta \ap \langle x_1, c_1 \rangle, \theta \ap
           \langle x_2, c_2 \rangle) \leq d (x_1, x_2) + d_{\DtoL{\asgn} m (E' J)} (c_1,
           c_2)
		\end{aligned}
		\right)\\
    &\iff 
    \left(
      \begin{aligned}
      &\fa{x_1, x_2 \in U I, c_1, c_2 \in U (T J), K \in \CC, n \in M, f \in \Div\qQ\CC (X \otimes E' J, \asgn^n_K) }\\
      & \qquad \asgn^{m \cdot n}_K (f\kl \ap \theta \ap \langle
      x_1, c_1 \rangle, f\kl \ap \theta \ap \langle x_2, c_2 \rangle)
      \leq d (x_1, x_2) + d_{\DtoL{\asgn} m (E' J)} (c_1, c_2)
      \end{aligned}
		\right)\\
    & \mathbin{\stackrel{\dag}{\iff}} \left(
      \begin{aligned}
      &\fa{ x_1, x_2 \in U I, c_1, c_2 \in U (T
                   J), K \in \CC, n \in M, f \in \Div\qQ\CC (X
                   \otimes E' J, \asgn^n_K)}\\
    & \qquad\asgn^{m \cdot n}_K ((f_{x_1})\kl \ap c_1,
      (f_{x_2})\kl \ap c_2) \leq d (x_1, x_2) + d_{\DtoL{\asgn}}^m (c_1,
      c_2)
      \end{aligned}
		\right).
  \end{align*}
  In the step $\mathbin{\stackrel{\dag}{\iff}}$, we used the equality (1). To show this goal,
  we proceed as follows. Let
  $x_1, x_2 \in U I, c_1, c_2 \in U (T J), K \in \CC, n \in M$ and
  $f \in \Div\qQ\CC (X \otimes E' J, \asgn^n_K)$. First, from the
  composability of $\asgn$, we obtain
  \[ \asgn^{m \cdot n}_K ((f_{x_1})\kl \ap c_1, (f_{x_2})\kl \ap c_2)
    \leq \asgn_J^m (c_1, c_2) + \sup_{(y_1, y_2) \in E J} \asgn^n_K
    (f_{x_1} \ap y_1, f_{x_2} \ap y_2) . \] We look at summands of the
  right hand side. First, we easily obtain
  $\asgn^m_J (c_1, c_2) \leq d_{\DtoL{\asgn} m (E' J)} (c_1, c_2)$. Next, from the
  nonexpansivity of $f$, for any $x_1,x_2\in UI,y_1, y_2 \in U J$, we have
  \[
    \asgn^n_K (f_{x_1} \ap y_1, f_{x_2} \ap y_2)=
    \asgn^n_K (f \ap \langle x_1, y_1 \rangle, f \ap \langle x_2, y_2 \rangle)  \leq d
    (x_1, x_2) + E' J (y_1, y_2) . \] Because $x+\top=\top$,
  we obtain
  \[ \fa{x_1, x_2 \in U I} \sup_{(y_1, y_2) \in E J} \asgn^n_K
    (f_{x_1} \ap y_1, f_{x_2} \ap y_2) \leq d (x_1, x_2) . \]
  Therefore we obtain the goal:
  \[ \asgn^{m \cdot n}_K ((f_{x_1})\kl \ap c_1, (f_{x_2})\kl \ap c_2)
    \leq d_{\DtoL{\asgn} m (E' J)} (c_1, c_2) + d (x_1, x_2) = d (x_1, x_2) +
    d_{\DtoL{\asgn} m (E' J)} (c_1, c_2) . \]
  
  Next, let $\dot{T} \in \SGDLift TM\qQ E$. We show that
  $\LtoD{\dot{T}} \in \mDiv TM\qQ E$.
  
  The unit law of $\dot{T}$ immediately entails
  \[ \eta_I \in \Div\qQ\CC (E' I, \dot{T} 1(E' I)) . \] Next, under
  the assumption on $(\CC, T)$ and $\qQ$, in $\Div\qQ\CC$ the functor
  $(-) \otimes E' I$ has a right adjoint $E' I \multimap (-)$ above
  the adjunction $(-) \times I \dashv I \Rightarrow (-)$ (Lemma
  \ref{lem:closed}).  Therefore each component of the internal Kleisli
  extension morphism $kl$ given in \eqref{eq:intkl} are nonexpansive
  morphisms in $\Div\qQ\CC$:
  \begin{displaymath}
    \xymatrix{
      \LtoD{\dot{T}} m (E' I) \otimes (E' I \multimap \LtoD{\dot{T}} n (E' J)) \rdh{\langle \pi_2, \pi_1 \rangle}\\
      (E' I \multimap \LtoD{\dot{T}} n (E' J)) \otimes \LtoD{\dot{T}} m (E' I) \rdh{\theta} \\
      \LtoD{\dot{T}} m ((E' I \multimap \LtoD{\dot{T}} n (E' J)) \otimes E' I) \rdh{\tmop{ev}\kl} \\
      \LtoD{\dot{T}} (m \cdot n) (E' J)
    }
  \end{displaymath}
  Therefore we conclude
  \[ \tmop{kl} \in \Div\qQ\CC (\LtoD{\dot{T}} m (E' I) \otimes (E' I
    \multimap \LtoD{\dot{T}} n (E' J)), \LtoD{\dot{T}} (m \cdot n) (E' J)) . \]
  We also easily have monotonicity:
  $\LtoD{\dot{T}} m (E' I) \leq \LtoD{\dot{T}} n (E' I)$ for $m \leq n$ by
  condition 1 of graded divergence lifting. We thus conclude that
  $\LtoD{\dot{T}} m E' I \in \mDiv TM\qQ E$.
  
  We finally show $\dot{T} \preceq \DtoL{\LtoD{\dot{T}}}$. Let
  $c_1, c_2 \in U (T I)$.  We show
  \begin{equation}
    \label{eq:goal}
    \sup_{n \in M, J \in \CC, f \in \Div\qQ\CC (X,
      \dot{T} n (E' J))} d_{\dot{T} (m \cdot n) (E' J)} (f\kl (c_1), f\kl
    (c_2)) \leq d_{\dot{T} m X} (c_1, c_2) .
  \end{equation}
  Let $n \in M, J \in \CC, f \in \Div\qQ\CC (X, \dot{T} n (E'
  J))$. Since $\dot{T}$ is an $M$-graded $\qQ$-divergence lifting of
  $T$, we obtain
  \[ f\kl \in \Div\qQ\CC (\dot{T} m X, \dot{T} (m \cdot n) (E' J))
    . \] This implies the inequality
  $d_{\dot{T} (m \cdot n) (E' J)} (f\kl (c_1), f\kl (c_2)) \leq
  d_{\dot{T} m X} (c_1, c_2)$ in $\qQ$. By taking the sup for
  $n, J, f$, we obtain the inequality \eqref{eq:goal}.
\end{appendixproof}

\subsection{Generation of Divergences}
\label{sec:generatedness_of_divergence}

It has been shown that DP can be interpreted as hypothesis testing
(\cite{WassermanZ10,KairouzOV15}).  Given a query $c\colon I\to \giry J$ and
adjacent datasets $(d_1,d_2) \in R_{\mathrm{adj}}\subseteq I^2$, we
consider the following hypothesis testing with the null and
alternative hypotheses:
\begin{align*}
  &H_0 \colon \text{The output $y$ comes from the dataset $d_1$},\\
  &H_1 \colon \text{The output $y$ comes from the dataset $d_2$}.
\end{align*}
For any rejection region $S \in \Sigma_J$, the Type I and Type II
errors are then represented by $\Pr[c(d_1) \in S]$ and
$\Pr[c(d_2) \notin S]$, respectively.  \cite{KairouzOV15} showed that
$c$ is $(\ve,\delta)$-DP \emph{if and only if} for any
adjacent datasets $(d_1,d_2) \in R_{\mathrm{adj}}\subseteq I^2$, the
pair of Type I error and Type II error lands in the \emph{privacy
  region} $R(\ve,\delta)$:
\begin{displaymath}
  \fa{S \in \Sigma_J} (\Pr[c(d_1) \in S],\Pr[c(d_2) \notin S]) \in
  \underbrace{\{ (x,y) \in [0,1]^2 | (1-x) \leq \exp(\ve) y +
    \delta \}}_{\triangleq R(\ve,\delta)} .
\end{displaymath}
They also showed that this is equivalent to
the testing using probabilistic decision rules~\cite[Corollary 2.3]{KairouzOV15}:
\[
  \fa{k \colon J \to \giry \{\mathsf{Acc},\mathsf{Rej}\}} (\Pr[ k\kl c(d_1)  = \mathsf{Acc}],\Pr[k\kl c(d_2) = \mathsf{Rej}]) \in R(\ve,\delta).
\]
Later \cite{DBLP:journals/corr/abs-1905-09982} generalized this
probabilistic variant of hypothesis testing to general statistical
divergences, and arrived at a notion of {\em $k$-generatedness} of
statistical divergences ($k \in \NN \cup \{\infty\}$). Following their
generalization, we introduce the concept of {\em
  $\Omega$-generatedness} of divergences on monads.
\begin{definition}\label{def:gen}
  Let $\Omega \in \CC$.  A divergence $\asgn\in\mDiv TM\qQ E$ is {\em
    $\Omega$-generated} if for any $m \in M$, $I\in\CC$ and
  $c_1,c_2 \in U(TI)$,
  \begin{displaymath}
    \asg m I (c_1,c_2) = \sup_{k \colon I \to T\Omega} \asg m \Omega
    (k\kl \ap c_1,k\kl \ap c_2).
  \end{displaymath}
\end{definition}
An equivalent definition of $\asgn \in \mDiv TM\qQ E$ being
$\Omega$-generated is: the following holds for any
$m\in M,I\in\CC,c_1,c_2\in U(TI),v\in\qQ$:
\[
  \asg m I (c_1,c_2) \leq v \iff \fa{k \colon I \to T\Omega} (k\kl\ap
  c_1,k\kl\ap c_2) \in \tilde\asgn (m,v) \Omega.
\]
Here $\tilde\asgn(m,v)\Omega$ is the binary relation
$\{ (c_1, c_2)~|~ \asg m\Omega (c_1, c_2) \leq v \}$; see also
\eqref{eq:divergence:adjacency}.  For an $\Omega$-generated divergence
$\asgn$, its component $\asg {m} \Omega$ at $\Omega$ is an essential
part that determines all components $\asg {m} I$ of $\asgn$.  When a
divergence is shown to be $\Omega$-generated, the calculation of the
codensity lifting $\coden T {\asgn}$ given in Section
\ref{sec:codensity_lifting} will be simplified (Section
\ref{sec:simplification}).

We illustrate $\Omega$-generatedness of various divergences.  First,
we show the $\Omega$-generatedness of divergences on the Giry monad
$\giry$ in Tables \ref{tab:divdp} and \ref{tab:divstat}.
\begin{itemize}
\item Divergence $\divndp$ is generated over the two-point discrete 
  space $2$~\cite[Section
  B.7]{DBLP:journals/corr/abs-1905-09982}. The binary relation
  $(\widetilde\divndp(\ve,\delta)2)$ coincides with the privacy
  region $R(\ve,\delta)$.
  
\item Divergence $\divntv$ is also generated over $2$~\cite[Section
  C.1]{DBLP:journals/corr/abs-1905-09982}.
  
\item Divergences $\divnrn^\alpha$, $\divnchi$, $\divnhd$ and
  $\divnkl$ are generated over the countably infinite discrete space $\mathbb{N}$. In
  contrast, they are not $N$-generated for every finite discrete space
  $N$~\cite[Sections B.5 and B.9]{DBLP:journals/corr/abs-1905-09982}.
\end{itemize}
On the sub-Giry monad $\sgiry$, the divergence $\divndp$ is $1$-generated,
and the total variation distance $\divntv$ is $2$-generated.
\begin{proposition}\label{pp:1-generatedness_DP}
  The divergence $\divndp \in \mDiv \sgiry  \qRp \qRp \mEQ$ is $1$-generated.
\end{proposition}
\begin{appendixproof}
  (Proof of Proposition \ref{pp:1-generatedness_DP})
  We write $|1| = \{ \ast \}$.
  We first check the measurable isomorphism $\sgiry 1 \cong [0,1]$.
  The measurable functions $\ev_{\{\ast\}} \colon \sgiry 1 \to [0,1]$ 
  ($\nu \mapsto \nu (\ast)$) and the function $H \colon | [0,1]| \to |\sgiry 1|$
  ($r \mapsto r \cdot \mathbf{d}_{\ast}$) are mutually inverse. 
  For any (Borel-)measurable $U \in \Sigma_{[0,1]}$, we have 
  $ H^{-1}(\ev_{\{\ast\}}^{-1} (U)) = U$ and
  $H^{-1}(\ev_{\emptyset}^{-1} (U)) = [0,1]$ if $0 \in U$ and $H^{-1}(\ev_{\emptyset}^{-1} (U)) = \emptyset$ otherwise.
  Since all generators of $\Sigma_{\sgiry 1}$ are $\ev_{\{\ast\}}^{-1}(U)$ and $\ev_{\emptyset}^{-1} (U)$ where $U \in \Sigma_{[0,1]}$, we conclude the measurability of $H$.
  Thus, $f \colon I \to [0,1]$ corresponds bijectively to $H \circ f \colon I \to \sgiry 1$, and 
  \[
    \int_I f d\nu_1 = \int_I \ev_{\{\ast\}} \circ H \circ f d\nu_1
    = ((H \circ f)\kl \nu_1) (\{\ast \}).
  \]
  We then obtain, for all $I \in \Meas$, $\nu_1,\nu_2 \in \sgiry I$ 
  \begin{align*}
    \divdp \ve I (\nu_1,\nu_2)
    &= \sup_{S \in \Sigma_I}(\nu_1(S) - \exp(\ve)\nu_2(S))\\
    &\leq \sup_{f_S \colon I \to [0,1] }(\int_I f_S d\nu_1 -  \exp(\ve)\int_I f_S d\nu_2)\\
    &= \sup_{f_S \colon I \to \sgiry  [0,1] }(((H \circ f_S)\kl\nu_1)(\ast) -  \exp(\ve)((H \circ f_S)\kl\nu_2)(\ast))\\
    &\leq \sup_{f_S \colon I \to \sgiry  [0,1] }\sup_{S'\in \Sigma_1 (\iff S' = \{\ast\},\emptyset)}((H \circ f_S)\kl\nu_1)(S') -  \exp(\ve)((H \circ f_S)\kl\nu_2)(S'))\\
    &= \sup_{f_S \colon I \to \sgiry  [0,1] } \divdp \ve 1 ((H \circ f_S)\kl\nu_1,(H \circ f_S)\kl\nu_2)\\
    &= \sup_{g \colon I \to \sgiry 1 } \divdp \ve 1 (g\kl\nu_1,g\kl\nu_2)\\
    &\leq \divdp \ve I (\nu_1,\nu_2).
  \end{align*}
  The first inequality is given by $\nu(S) = \int_I \chi_S d\nu $ where $\chi_S \colon I \to [0,1]$ is the indicator function of $S$
  defined by $\chi_S(x) = 1$ when $x \in S$ and  $\chi_S(x) = 0$ otherwise. 
  The last inequality is given by the data-processing inequality which is given by the reflexivity and $\mEQ{}$-composability of $\divndp$.
\end{appendixproof}
\begin{proposition}\label{pp:2-generatedness_TV}
  The divergence $\divntv \in \mDiv \sgiry  1 \qRp \mEQ$ is not $1$-generated but $2$-generated.
\end{proposition}
\begin{appendixproof}
  (Proof of Proposition \ref{pp:2-generatedness_TV})
  We first prove that $\divntv$ is not $1$-generated.
  We write $|2| = \{0,1\}$.
  We define $\nu_1,\nu_2 \in \sgiry 2$ by 
  \[
    \nu_1 = \frac{1}{2}\cdot\mathbf{d}_0 + \frac{1}{2}\cdot\mathbf{d}_1, \qquad \nu_2 = \frac{1}{3}\cdot\mathbf{d}_0 + \frac{2}{3}\cdot\mathbf{d}_1.
  \]
  Then the total variation distance between them is calculated by
  \[
    \divtv 2(\nu_1,\nu_2) = \frac{1}{2}\left(\left|\frac{1}{2}-\frac{1}{3} \right| + \left|\frac{1}{2} - \frac{2}{3} \right|\right) = \frac{1}{6}.
  \]
  On the other hand, for any $f \colon 2 \to \sgiry1$, we have
  \begin{align*}
    \divtv 1(f\kl(\nu_1),f\kl(\nu_2))
    &= \frac{1}{2} \left|\frac{1}{2}f(0) + \frac{1}{2}f(1) - \frac{1}{3}f(0) - \frac{2}{3}f(1)\right|\\
    &= \frac{1}{2} \left|\frac{1}{6}f(0) - \frac{1}{6}f(1)\right|\\
    &= \frac{1}{12} \left| f(0) - f(1) \right| \\
    & \leq \frac{1}{12}.
  \end{align*}
  This implies that $\divntv$ is not $1$-generated.

  Next, we prove that $\divntv$ is $2$-generated.
  From the data-processing inequality $\divntv$  which is given by the reflexivity and $\mEQ{}$-composability of $\divntv$,
  we obtain for any $\nu_1,\nu_2 \in \sgiry I$,
  \[
    \divtv I (\nu_1,\nu_2) \geq \sup_{g \colon I \to \sgiry 2}\divtv 2 (g\kl \nu_1, g\kl \nu_2).
  \]
  We show that the above inequality becomes the equality for some $g$.

  We fix $\nu_1,\nu_2 \in \sgiry I$, a base measure $\mu$ over $I$ satisfying the absolute continuity $\nu_1,\nu_2 \ll \mu$ and
  the Radon-Nikodym derivatives (density functions) $\frac{d\nu_1}{d\mu}, \frac{d\nu_2}{d\mu}$ of $\nu_1,\nu_2$ with respect to $\mu$ respectively.

  Let $A = (\frac{d\nu_1}{d\mu} - \frac{d\nu_2}{d\mu})^{-1}([0,\infty))$ and $B = I \setminus A$.
  We define $g \colon I \to \sgiry2$ by
  $g(x) = \mathbf{d}_0$ if $x \in B$ and $g(x) = \mathbf{d}_1$ otherwise.
  Then for any $\nu \in \sgiry I$ we have
  \[
    (g\kl \nu)(\{0\})
    = \int_I g(-)(\{0\}) d\nu
    = \int_A g(-)(\{0\}) d\nu + \int_B g(-)(\{0\}) d\nu
    = \int_A 1 d\nu + \int_B 0 d\nu
    = \nu(A).
  \]
  Similarly we have $(g\kl \nu)(\{1\}) = \nu(B)$.
  Therefore, we obtain
  \begin{align*}
    \frac{1}{2}\divtv I(\mu_1,\mu_2)
    &=\frac{1}{2}\int_I \left| \frac{d\nu_1}{d\mu}(x) - \frac{d\nu_2}{d\mu}(x) \right|~d\mu(x)\\
    &=\frac{1}{2}\int_{A} \frac{d\nu_1}{d\mu}(x) - \frac{d\nu_2}{d\mu}(x)~d\mu(x) + \frac{1}{2}\int_{B} \frac{d\nu_2}{d\mu}(x) - \frac{d\nu_1}{d\mu}(x) ~ d\mu(x)\\
    &=\frac{1}{2}(\nu_1(A) - \nu_2(A) + \nu_2(B) - \nu_1(B))\\
    &=\frac{1}{2}((g\kl\nu_1)(\{ 0 \}) - (g\kl\nu_2)(\{ 0 \}) + (g\kl\nu_2)(\{ 1 \}) - (g\kl\nu_2)(\{ 1 \}))\\
    &=\frac{1}{2}(| (g\kl\nu_1)(\{ 0 \}) - (g\kl\nu_2)(\{ 0 \})| + |(g\kl\nu_2)(\{ 1 \}) - (g\kl\nu_2)(\{ 1 \}) |)\\
    &= \divtv 2(g\kl(\mu_1),g\kl(\mu_2))
  \end{align*}
  We then conclude that $\Delta^{\mathrm{TV}}$ is $2$-generated.
\end{appendixproof}

\paragraph*{$\Omega$-Generatedness of Preorders on Monads}

We relate $\Omega$-generatedness of divergences and preorders on
monads studied in \mcite{preordermonad}.  Let $T$ be a monad on $\Set$
and $\Omega$ be a set. \cite{preordermonad} introduced the concept of
{\em congruent} and {\em substitutive} preorders on $T\Omega$ as
those satisfying:
\begin{description}
\item[Substitutivity] For any function $f\colon \Omega\to T\Omega$ and
  $c_1,c_2\in T\Omega$, $c_1 \leq c_2$ implies
  $f\kl(c_1) \leq f\kl(c_2)$.
\item[Congruence] For any function $f_1,f_2\colon J\to T\Omega$, if
  $f_1(x) \leq f_2(x)$ holds for any $x\in J$, then
  $f_1\kl (c) \leq f_2\kl (c)$ holds for any $c\in T\Omega$.
\end{description}
For instance, any component of a preorder on $T$ at $\Omega$ forms a
congruent and substitutive preorder on $T\Omega$.  We write
$\mathbf{CSPre}(T,\Omega)$ for the set of all congruent and
substitutive preorders on $T\Omega$, and $\mathbf{Pre}(T)$ for the
collection of all preorders on $T$.  \cite{preordermonad} gave a
construction $[-]^\Omega\colon \mathbf{CSPre}(T,\Omega)\to\mathbf{Pre}(T)$
of preorders on $T$ from congruent and substitutive preorders on
$T\Omega$:
\[
  c_1 [\leq]^\Omega_J c_2 \iff \fa{g \colon J \to T\Omega}g\kl
  (c_1) \leq g\kl (c_2)
\]
The constructed preorders on $T$ are $\Omega$-generated in the
following sense:
\begin{proposition}\label{pp:K-generatedness:premonad}
  For any ${\leq} \in \mathbf{CSPre}(T,\Omega)$, the
  $\qB$-divergence $\asg {[\leq]^\Omega} {}$ corresponding to the preorder ${[\leq]}^\Omega$ on
  $T$ is $\Omega$-generated (see Proposition
  \ref{pp:preorder_on_monad_is_divergence} for the correspondence).
\end{proposition}
\begin{appendixproof}
  (Proof of Proposition \ref{pp:K-generatedness:premonad}) For all set
  $J$ and $c_1,c_2 \in TJ$, we have
  \begin{align*}
    \asg {[\leq]^\Omega} J(c_1,c_2) = 1
    &\iff c_1 [\leq]^\Omega_J c_2\\
    &\iff \bigwedge_{g \colon J \to T\Omega} g\kl (c_1) \leq g\kl (c_2)\\
    &\iff \bigwedge_{g \colon J \to T\Omega} g\kl (c_1)  \mathbin{[\leq]^\Omega_\Omega} g\kl (c_2)\\
    &\iff \sup_{g \colon J \to T\Omega} \asgn([\leq]^\Omega)_\Omega (g\kl (c_1),g\kl (c_2)) = 1.
  \end{align*}
  This implies that $\asg {[\leq]^\Omega} {}$ is $\Omega$-generated.
\end{appendixproof}
Applying this proposition, we can determine $\Omega$-generatedness of
preorders on monads:
\begin{itemize}
\item If the monad $T$ has a rank $\alpha$, the construction
  $[-]^\alpha$ is bijective~\cite[Theorem 7]{preordermonad}. Hence for
  such a monad, each preorder on $T$ corresponds to an
  $\alpha$-generated $\qB$-divergence.
\item For the subprobability distribution monad $D_s$ on $\Set$,
  \cite{DBLP:journals/entcs/Sato14} identified all preorders on $D_s$:
  there are 41 preorders on $D_s$. Among them, 25 preorders are
  1-generated, while 16 preorders are 2-generated \cite[Proposition
  6.3]{DBLP:journals/entcs/Sato14}.
\end{itemize}

\subsection{An Adjunction between Quantitative Equational Theories and Divergences}
\label{sec:adj}

\newcommand{\QETC}[3]{\ol{#1}^{\mathrm{QET}({#2},{#3})}}
\newcommand{\Genn}{\mathrm{Gen}}
\newcommand{\Gen}[1]{\Genn(#1)}
\newcommand{\condeq}[3]{{#1}=_{#3}{#2}}

\cite{10.1145/2933575.2934518} introduced a concept of {\em
  quantitative equational theory} as an algebraic presentation of
monads on the category of (pseudo-)metric spaces.  A quantitative
equational theory is an equational theory with indexed equations
$\condeq {t} {u} {\ve}$ having the axioms of pseudometric spaces, plus
suitable axioms reflecting properties of quantitative algebras. A
quantitative equational theory determines a pseudometric on the set of
$\Omega$-terms.

Consider a set $\Omega$ of function symbols of finite arity.  If $n$
is the arity of a function $f \in \Omega$, we write
$f \colon n \in \Omega$.  Let $X$ be a set of variables, and let
$T_\Omega X$ be the $\Omega$-term algebra over $X$.  For
$f \colon n \in \Omega$ and $t_1,\ldots,t_n \in T_\Omega X$, we write
$f(t_1,\ldots,t_n)$ for the term obtained by applying $f$ to
$t_1,\ldots,t_n$.  The construction $X\mapsto T_\Omega X$ forms a (strong) monad
on $\Set$ whose unit sends variables to terms, that is,
$\eta_X(x) = x$, and Kleisli extension
$h\kl \colon T_\Omega I \to T_\Omega X$ of function
$h \colon I \to T_\Omega X$ is defined inductively by
\[
  h\kl (x) \triangleq h(x), \quad
  h\kl (f(t_1,\ldots,t_n)) \triangleq f(h\kl(t_1),\ldots,h\kl(t_n)).
\]
A substitution of $\Omega$-terms over $X$ is a function $\sigma \colon X \to T_\Omega X$. For
$t \in T_\Omega X$, we call $\sigma\kl(t)$ the substitution of $\sigma$
to $t$.  We define the set of 
{\em indexed equations} of terms by
\[
  \VV (T_\Omega X) \triangleq \{  \condeq t u \ve ~|~ t,u \in T_\Omega X, \ve \in \QQ^+ \}.
\]
Here the index $\ve$ runs over non-negative rational numbers.  A {\em
  conditional quantitative equation} is a judgment of the following
form
\[
  \{\condeq {t_i} {u_i} {\ve_i} ~|~ i \in I\} \vdash \condeq t u
  \ve \qquad ( I \colon \text{countable},
  \condeq {t_i} {u_i} {\ve_i},  \condeq t u
  \ve \in \VV (T_\Omega X));
\]
the left hand side of turnstile ($\vdash$) is called hypothesis and the right
hand side conclusion.  We denote by $\EE(T_\Omega X)$ the set of
conditional quantitative equations.
For any countable subset $\Gamma$ of $\VV (T_\Omega X)$ and any substitution
$\sigma \colon X \to T_\Omega X$,
we define
$\sigma(\Gamma) \triangleq \{ \condeq {\sigma\kl(t_i)} {\sigma\kl(u_i)}
{\ve_i} ~|~ \condeq {t_i} {u_i} {\ve_i} \in \Gamma \}$.
\begin{definition}[Quantitative Equational Theory {\cite[Definition 2.1]{10.1145/2933575.2934518}}]
  A quantitative equational theory (QET for short) of type $\Omega$
  over $X$ is a set $U \subseteq \EE(T_\Omega X)$ closed under the
  rules summarized as Figure \ref{fig:qet}.
  \begin{figure}[t]
    \begin{align*}
      \emptyset & \vdash \condeq t t 0  \in U \label{QET:axiom:Ref} \tag{Ref}\\
      \{\condeq t u  \ve \} &\vdash  \condeq u t  \ve  \in U \label{QET:axiom:Sym} \tag{Sym}\\
      \{\condeq t u  \ve, \condeq u v  {\ve'} \} &\vdash  \condeq t v {\ve + \ve'}  \in U \label{QET:axiom:Tri} \tag{Tri}\\
      \fa{\ve' \in \QQ^+ } \{\condeq t u  \ve \} &\vdash  \condeq t u  {\ve + \ve'}  \in U \label{QET:axiom:Max} \tag{Max}\\
      \fa{\ve \in \QQ^+}  \{\condeq t u {\ve'} | \ve < \ve' \} &\vdash  \condeq t u \ve  \in U  \label{QET:axiom:Arch} \tag{Arch}\\
      \fa{f \colon n \in \Omega}  \{\condeq {t_i} {u_i}  {\ve} | 1 \leq i \leq n \} &\vdash  \condeq {f(t_1,\ldots,t_n)} {f(u_1,\ldots,u_n)} {\ve}  \label{QET:axiom:Nonexp} \tag{Nonexp}\\
      \fa{\sigma \colon X \to T_\Omega X}
      {\Gamma \vdash \condeq t u  \ve \in U} &\implies {\sigma(\Gamma)  \vdash  \condeq  {\sigma\kl(t)} {\sigma\kl(u)}  \ve \in U}
                                               \label{QET:axiom:Subst} \tag{Subst}\\
      \Gamma' \vdash \condeq t u  \ve \in U\land \fa{\psi \in \Gamma'}\Gamma \vdash \psi\in U & \implies {\Gamma \vdash \condeq t u  \ve \in U} \label{QET:axiom:Cut}\tag{Cut}\\
      {\condeq t u  \ve \in \Gamma} &\implies {\Gamma \vdash \condeq t u  \ve \in U} \label{QET:axiom:Assumpt}\tag{Assumpt}
    \end{align*}
    \caption{Quantitative Equational Theory Rules}
    \label{fig:qet}
  \end{figure}
  We write $\QET(\Omega,X)$ for the set of QETs of type $\Omega$ over
  $X$. We regard it as a poset $(\QET(\Omega,X),\subseteq)$ by the set
  inclusion order. Given a set $U_0$ of conditional quantitative
  equations of type $\Omega$ over $X$, by $\QETC{U_0}\Omega X$ we mean the
  least QET containing $U_0$.
\end{definition}

We state an adjunction between quantitative equational theories and
divergences on free-algebra monads on $\Set$. More specifically, we
construct the following adjunction and isomorphism between posets:
\begin{equation}
  \label{eq:qetdiag}
  \xymatrix@C=1.5cm{
    (\QET(\Omega,X),\subseteq)
    &
    (\CSEPMet(T_\Omega, X),\preceq)
    \adjunction{l}{U[-]}{d[-]}
    \ar@<.3pc>[r]^-{\Genn} \ar@{}[r]|-{\cong}
    &
    (\mathbf{DivEPMet}(T_\Omega,X),\preceq)
    \ar@<.3pc>[l]^-{(-)_X}
  }.
\end{equation}
By combining these, a QET of type $\Omega$ over $X$ determines an
$X$-generated $\mEQ$-relative $\qRp$-divergence on
$T_\Omega$ and vice versa.  The poset in the middle is that of {\em
  congruent} and {\em substitutive pseudometrics}, which are a quantitative
analogue of congruent and substitutive preorders.
\begin{definition}\label{def:csepmet}
  Let $T$ be a monad on $\Set$ and $X \in \Set$.
  A congruent and substitutive pseudometric (CS-EPMet for short) on $TX$ 
  is an extended pseudometric\footnote{A function $d \colon A^2 \to \qRp$
  is called an extended pseudometric on $A$ if
  $d(a,a) = 0$ (reflexivity),
  $d(b,a) = d(a,b)$ (symmetry) and
  $d(a,c) \leq d(a,b) + d(b,c)$ (triangle-inequality) hold for all $a,b,c \in A$.
  }
  $d \colon (TX)^2 \to \qRp$ on $TX$
  satisfying
  \begin{description}
  \item[Substitutivity] For all function $fX \to TX$ and
    $c_1,c_2\in TX$, $d(f\kl(c_1),f\kl(c_2)) \leq d(c_1,c_2)$.
  \item[Congruence] For all set $I$, function $f_1,f_2 \colon  I\to TX$ and $c \in TI$, 
    $d(f_1\kl(c),f_2\kl(c)) \leq \sup_{i \in I} d(f_1(i),f_2(i))$. 
  \end{description}
  We denote by $\CSEPMet(T, X)$ the set of CS-EPMets on $TX$.
  We then make it into a poset $(\CSEPMet(T, X),\preceq)$
  by the following pointwise opposite order:
  \[
    d\preceq d' \iff \fa{c_1,c_2 \in TX}d(c_1,c_2) \geq d'(c_1,c_2).
  \]
\end{definition}
\begin{definition}
  Let $T$ be a monad on $\Set$ and $X \in \Set$.  We denote by
  $\mathbf{DivEPMet}(T,X)$ the collection of $X$-generated
  $\mEQ$-relative $\qRp$-divergences $\asgn$ on $T$ such that each
  component $\asg{}I$ is an extended pseudometric. We restrict the
  partial order $\preceq$ on $\mDiv T1\qRp\mEQ$ to
  $\mathbf{DivEPMet}(T,X)$.
\end{definition}
We next introduce various monotone functions appearing in
\eqref{eq:qetdiag}.
\begin{align*}
  &d[U] (t,u)
  \triangleq \inf \left\{ \ve \in \QQ^+ \sep \emptyset
    \vdash \condeq {t}{u}{\ve} \in U \right\}
  & &
    \Gen d_I(c_1,c_2)
  \triangleq
    \sup_{k \colon I \to TX} d(k\kl (c_1),k\kl (c_2)) \\
  & U[d]
  \triangleq
    \QETC{\left\{\emptyset \vdash \condeq t u  \ve \sep \ve \in \QQ^+, d (t,u) \leq \ve \right\}}\Omega X
  & &
    (\asgn)_X
  \triangleq \asgn_X
\end{align*}
\begin{proposition}
  The functions $d[-],U[-],\Genn,(-)_X$ defined above
  are all well-defined monotone functions having types
  given in \eqref{eq:qetdiag}.
\end{proposition}
That $d[U]$ is an extended pseudometric is shown in the beginning of
\cite[Section 5]{10.1145/2933575.2934518}. Here we additionally show
that it enjoys congruence and substitutivity of Definition
\ref{def:csepmet}.  The function $\Genn$ is taken from the right hand
side of the definition of $\Omega$-generatedness (Definition
\ref{def:gen}).  The function $(-)_X$ simply extracts the $X$-th
component of a given divergence.
\begin{toappendix}
\begin{lemma}\label{lem:QET=>CSPEMet}
  For any $U \in \QET(\Omega,X)$,
  the function $d[U] \colon (T_\Omega X)^2 \to \qRp$ defined by
  \[
  d[U] (t,u) \triangleq \inf \left\{ \ve \in \QQ^+ \sep \emptyset
    \vdash \condeq {t}{u}{\ve} \in U \right\}
  \]
  is a CS-EPMet on $T_\Omega X$ such that
  $d[U] (t,u) \in \QQ^+ \implies \emptyset \vdash \condeq {t}{u}{d[U] (t,u)}$.
\end{lemma}
\begin{proof}
  We first check the axioms of extended pseudometric.

  By \eqref{QET:axiom:Ref}, $U$ contains
  $ \emptyset\vdash \condeq {t}{t}{0}$ for each $t \in T_\Omega X$.
  Hence $d[U] (t,t) = 0$ holds for all $t \in T_\Omega X$.

  By \eqref{QET:axiom:Sym} and \eqref{QET:axiom:Cut},
  $ \emptyset\vdash \condeq {t}{u}{\ve}$ if and only if
  $ \emptyset\vdash \condeq {u}{t}{\ve}$.  Hence, for all
  $t,u \in T_\Omega X$,
  \[
    d[U](t,u) = \inf \left\{ \ve \in \QQ^+ \sep \emptyset
      \vdash \condeq {t}{u}{\ve} \right\} = \inf \left\{
      \ve \in \QQ^+ \sep \emptyset\vdash \condeq
      {u}{t}{\ve} \right\} =d[U] (u,t)
  \]

  By \eqref{QET:axiom:Tri} and \eqref{QET:axiom:Cut}, if
  $\emptyset\vdash \condeq {t}{u}{\ve}$ and
  $\emptyset\vdash \condeq {u}{v}{\ve'}$ then
  $\emptyset\vdash \condeq {t}{v}{\ve + \ve'}$.
  Hence, for all $t,u,v \in T_\Omega X$,
  \begin{align*}
    d[U](t,v)
    & = \inf \left\{ \ve^\ast  \in \QQ^+  \sep \emptyset\vdash \condeq {t}{v}{\ve^\ast}  \right\}\\
    & \leq \inf \left\{ \ve+\ve'  \sep  \emptyset\vdash \condeq {t}{u}{\ve}  \land \emptyset\vdash \condeq {u}{v}{\ve'}\right\}\\
    & \leq \inf \left\{ \ve  \in \QQ^+ \sep  \emptyset\vdash \condeq {t}{u}{\ve}   \right\}
      +\inf \left\{ \ve'  \in \QQ^+ \sep \emptyset\vdash \condeq {u}{v}{\ve'} \right\}\\
    &= d[U](t,u) + d[U](u,v).
  \end{align*}

  We next check the substitutivity.  Let $t,u \in T_\Omega X$ and
  $h \colon X \to T_\Omega X$.
  By \eqref{QET:axiom:Subst}, we have
  \[ \emptyset\vdash \condeq{t}{u}{\ve} \in U \implies
    \emptyset \vdash\condeq{h\kl(t)}{h\kl(u)}{\ve} \in U.
  \]
  Since $\ve$ is arbitrary, we conclude the substitutivity as
  follows:
  \begin{align*}
    d[U](h\kl(t),h\kl(u))
    &= \inf \left\{ \ve \in \QQ^+ \sep
      \emptyset\vdash \condeq{h\kl(t)}{h\kl(u)}{\ve} \in U
    \right\} \\
    &\leq \inf \left\{ \ve \in \QQ^+ \sep \emptyset
      \vdash \condeq{t}{u}{\ve} \in U \right\}\\
      &= d[U](t,u).
  \end{align*}

  Next, we check the congruence, Let $t \in T_\Omega I$ and
  $h_1,h_2 \colon I \to T_\Omega X$ By applying
  \eqref{QET:axiom:Nonexp} and \eqref{QET:axiom:Cut} inductively by unfolding the structure of
  $t$,
  \begin{equation}\label{eq:QET=>CSPEMet:congr:1}
    \fa{i \in I} \emptyset\vdash \condeq{h_1(i)}{h_2(i)}{\ve}
    \in U \implies
    \emptyset\vdash\condeq{h_1\kl(t)}{h_2\kl(t)}{\ve} \in U.
  \end{equation}
  If
  $\sup_{i \in I} d[U](h_1(i),h_2(i)) \leq \ve'$ for some
  $\ve' \in \QQ^+ $, then we have  
  $ d[U](h_1(i),h_2(i)) \leq \ve'$ for all $i \in I$.  By
  \eqref{QET:axiom:Max},\eqref{QET:axiom:Cut} and definition of $d_X^U$, we have
  $ \vdash \condeq{h_1(i)}{h_2(i)}{\ve'} \in U$ for all
  $i \in I$.  Hence,
   \begin{equation}\label{eq:QET=>CSPEMet:congr:2}
    \sup_{i \in I} d[U](h_1(i),h_2(i)) \leq \ve' \implies
    \fa{i \in I} \emptyset \vdash
    \condeq{h_1(i)}{h_2(i)}{\ve'} \in U.
   \end{equation}
  From the above two implications \eqref{eq:QET=>CSPEMet:congr:1} and \eqref{eq:QET=>CSPEMet:congr:2}, 
  We conclude the congruence as follows:
  \begin{align*}
    d[U](h_1\kl(t),h_2\kl(t))
    & = \inf \left\{ \ve' \in \QQ^+  \sep \emptyset\vdash\condeq{h_1\kl(t)}{h_2\kl(t)}{\ve'} \in U \right\}\\
    & \leq \inf \left\{ \ve' \in \QQ^+  \sep \fa{i \in I} \emptyset \vdash \condeq{h_1(i)}{h_2(i)}{\ve'} \in U \right\}\\
    & \leq \inf \left\{ \ve' \in \QQ^+  \sep \sup_{i \in I} d_X^U(h_1(i),h_2(i)) \leq \ve' \right\}\\
    & = \sup_{i \in I} d[U](h_1(i),h_2(i)).
  \end{align*}
   
  Finally, we assume $d[U] (t,u) \in \QQ^+$.
  By definition of $d[U](t,u)$,
  for any $\ve \in \QQ^+$ such that $d[U](t,u) < \ve$, there is
  $\ve' \in \QQ^+$ satisfying $d[U](t,u) \leq \ve' < \ve$ and $\condeq {t}{u}{\ve'} \in U$.
  Since $\ve \in \QQ^+$ is arbitrary, by \eqref{QET:axiom:Max} and \eqref{QET:axiom:Cut}, we conclude
  \[
  \fa{\ve \in \QQ^+}(d[U](t,u) < \ve \implies \condeq {t}{u}{\ve} \in U).
  \]
  Since $d[U] (t,u) \in \QQ^+$, by \eqref{QET:axiom:Arch} and \eqref{QET:axiom:Cut}, we have $\condeq {t}{u}{d[U] (t,u)} \in U$.
\end{proof}
The monotonicity of $d[-] \colon (\QET(\Omega,X),\subseteq) \to (\CSEPMet(T_\Omega, X),\preceq)$ is easy to prove: 
\begin{align*}
  U \subseteq V
  &\implies 
  \fa{t,u \in T_\Omega X}
  \inf \left\{ \ve \in \QQ^+ \sep \emptyset
    \vdash \condeq {t}{u}{\ve} \in U \right\}
    \geq
    \inf \left\{ \ve \in \QQ^+ \sep \emptyset
    \vdash \condeq {t}{u}{\ve} \in V \right\}\\
   &
   \iff \fa{t,u \in T_\Omega X} d[U](t,u) \geq d[V](t,u)\\
   &\iff
  d[U] \preceq d[V].
\end{align*}
\begin{lemma}\label{lem:CSPEMet=>divergence}
  Let $T$ be a monad on $\Set$, and let $X \in \Set$.  For any
  $d \in \CSEPMet(T, X)$, the family
  $\Gen d = \{\Gen d_I \colon (TX)^2 \to \qRp \}$ defined by
  \[
    \Gen d_I(c_1,c_2) = \sup_{k \colon I \to TX} d(k\kl (c_1),k\kl (c_2))  
  \]
  is an $X$-generated $\mEQ$-relative $\qRp$-divergence on $T$ where
  each $\Gen d_I$ is a pseudometric.
\end{lemma}
\begin{proof}
   From the reflexivity
  of $d$, we have the reflexivity of $\Gen d_I$: for each $c \in TI$,
  \[
    \Gen d_I (c,c) = \sup_{k \colon I \to TX} d(k\kl (c),k\kl (c)) =
    0.
  \]
  Hence, the $\mEQ$-unit-reflexivity of $\asgn^d$ is already proved
  from the (proper) reflexivity.  From the symmetry of $d$, we have
  the symmetry of $\asgn^d_I$: for each $c_1,c_2 \in TI$,
  \begin{align*}
    \Gen d_I (c_1,c_2) &= \sup_{k \colon I \to TX} d(k\kl (c_1),k\kl
    (c_2))\\
    & = \sup_{k \colon I \to TX} d(k\kl (c_2),k\kl (c_1))\\
    & =
    \Gen d_I (c_2,c_1).
  \end{align*}
  From the triangle-inequality of $d$, we have the triangle-inequality
  of $\Gen d_I$: for all $c_1,c_2,c_3 \in TI$,
  \begin{align*}
    \Gen d_I (c_1,c_3) 
    &= \sup_{k \colon I \to TX} d(k\kl (c_1),k\kl (c_3))\\
    &\leq \sup_{k \colon I \to TX} d(k\kl (c_1),k\kl (c_2)) + d(k\kl (c_2),k\kl (c_3))\\
    &\leq \sup_{k \colon I \to TX} d(k\kl (c_1),k\kl (c_2)) + \sup_{k \colon I \to TX} d(k\kl (c_2),k\kl (c_3))\\
    &= \Gen d_I (c_1,c_2) +\Gen d_I (c_2,c_3) .
  \end{align*}
  
  From the reflexivity, congruence and substitutivity of $d$ and the
  triangle-inequality of $\asgn^d_I$, we next show the composability.
  Let $c_1,c_2 \in TI$ and $f_1,f_2 \colon I \to TJ$. We obtain,
  \begin{align*}
    \lefteqn{\Gen d_J (f_1\kl(c_1),f_2\kl(c_2))}\\
    &\leq \Gen d_J (f_1\kl(c_1),f_1\kl(c_2)) + \Gen d_J (f_1\kl(c_2),f_2\kl(c_2))\\
    &= \sup_{k \colon J \to TX} d((k\kl \circ f_1)\kl(c_1),(k\kl \circ f_1)\kl(c_2)) + \sup_{k \colon J \to TX}d ((k\kl \circ f_1)\kl(c_2),(k\kl \circ f_2)\kl(c_2))\\
    &\leq
      \sup_{k \colon J \to TX} d(f_1\kl(c_1),f_1\kl(c_2))
      +  \sup_{k \colon J \to TX} \sup_{i \in I} d(k\kl \circ f_1(i),k\kl \circ f_2(i))\\
    &=
      d(f_1\kl(c_1),f_1\kl(c_2))
      +  \sup_{k \colon J \to TX} \sup_{i \in I} d(k\kl \circ f_1(i),k\kl \circ f_2(i))\\
    &\leq \sup_{f_1 \colon I \to TX}d(f_1\kl(c_1),f_1\kl(c_2)) +
      \sup_{i \in I} \sup_{k \colon J \to TX} d(k\kl \circ f_1(i),k\kl \circ f_2(i))\\
    &= \Gen d_I (c_1,c_2) + \sup_{i \in I}  \Gen d_J(f_1(i),f_2(i)).
  \end{align*}

  Finally we show the $X$-generatedness of $\Gen d$ by definition
  \begin{align*}
    \Gen dI(c_1,c_2)
    &= \sup_{k \colon I \to TX} d (k\kl (c_1), k\kl (c_2))\\
    &= \sup_{h \colon X \to TX} \sup_{k \colon I \to TX} d (h\kl (k\kl (c_1)), h\kl (k\kl (c_2)))\\
    &= \sup_{k \colon I \to TX} \sup_{h \colon X \to TX} d (h\kl (k\kl (c_1)), h\kl (k\kl (c_2)))\\
    &= \sup_{k \colon I \to TX}  \Gen d_X (k\kl(c_1),k\kl(c_2))
  \end{align*}
  
  This completes the proof.
\end{proof}
The monotonicity of $\Genn \colon  (\CSEPMet(T_\Omega, X),\preceq) \to (\mathbf{DivEPMet}(T_\Omega, X),\preceq)$ is easy to prove:
\begin{align*}
d \preceq d' &
\implies 
\fa{c_1,c_2 \in TI}
\sup_{k \colon I \to TX} d(k\kl (c_1),k\kl (c_2)) \geq \sup_{k \colon I \to TX} d'(k\kl (c_1),k\kl (c_2))\\
&\iff \fa{c_1,c_2 \in TI} \Gen d_I(c_1,c_2) \geq  \Gen {d'}_I(c_1,c_2)\\
&\iff \Gen d \preceq \Gen {d'}.
\end{align*}
\end{toappendix}
\begin{therm}\label{thm:adjunction:QETandCSEPMet}
  For any set $\Omega$ of function symbols with finite arity and set
  $X$, the following holds for the monotone functions in \eqref{eq:qetdiag}:
  \begin{enumerate}
  \item $\Genn$ is the inverse of $(-)_X$.
  \item We have an adjunction satisfying
    $d[U[-]] = \id$:
    \begin{equation}
      \label{eq:qetadj}
      \xymatrix@C=1.5cm{ (\QET(\Omega,X),\subseteq) &
        (\CSEPMet(T_\Omega, X),\preceq) \adjunction{l}{U[-]}{d[-]}}
    \end{equation}
  \end{enumerate}
\end{therm}
\begin{appendixproof}
  (Proof of Theorem \ref{thm:adjunction:QETandCSEPMet})
  We first show $(\Gen {-})_X = \mathrm{id}$.
  Let $d \in \CSEPMet(T_\Omega, X)$.
  We fix arbitrary $t,u \in T_\Omega X$.
  From the substitutivity of $d$, we have $ d (k\kl (t), k\kl (u)) \leq d (t,u)$, but
  we can take $k = \eta_X$, we obtain
  \[
  \Gen dXC(t,u) = \sup_{k \colon X \to TX} d (k\kl (t), k\kl (u)) = d (t,u).
  \]
  Since $d,t,u$ are arbitrary, we conclude $(\Gen {-})_X = \mathrm{id}$.
  
  We show $\Gen {(-)_X} = \mathrm{id}$.
  Let $\asgn \in \mathbf{DivEPMet}(T_\Omega, X)$.
  By the $X$-generatedness of $\asgn$, 
  we have for all set $I$ and $t,u \in T_\Omega I$,
  \[
  \Gen {(\asgn)_X}_I (t,u) = \sup_{k \colon I \to TX} \asgn_X (k\kl (t),k\kl (u)) = \asg {} I (t,u).  
  \]
  Since $\asgn,I,t,u$ are arbitrary, we conclude $(\Gen {-})_X = \mathrm{id}$.
  
  We show the adjointness:
  $U[d] \subseteq V \iff d \geq d[V]$ for any $V \in \QET(\Omega,X)$ and $d \in \CSEPMet(T_\Omega, X)$.
  \begin{align*}
    U[d] \subseteq V
    &\iff 
      \QETC{\left\{\emptyset \vdash \condeq t u  \ve \sep \ve \in \QQ^+, d (t,u) \leq \ve \right\}} \Omega X \subseteq V\\
    &\iff \fa{t,u \in T_\Omega, \ve \in \QQ^+} d (t,u) \leq \ve \implies \emptyset \vdash \condeq t u  \ve \in V\\
    &\iff \fa{t,u \in T_\Omega, \ve \in \QQ^+} d (t,u) \leq \ve \implies \inf\left\{\ve' \in \QQ^+ \sep \emptyset \vdash \condeq t u  {\ve'}  \in V\right\} \leq \ve\\
    &\iff \fa{t,u \in T_\Omega} \inf\left\{\ve' \in \QQ^+ \sep \emptyset \vdash \condeq t u  {\ve'}  \in V\right\} \leq d (t,u)\\
    &\iff d \geq d[V]
  \end{align*}
  We notice that since $V$ is closed under \eqref{QET:axiom:Max}, \eqref{QET:axiom:Arch} and \eqref{QET:axiom:Cut}, we have the equivalence
  \begin{align*}
    \lefteqn{\inf\left\{\ve' \in \QQ^+ \sep \emptyset \vdash \condeq t u  {\ve'} \in V\right\} \leq \ve}\\
    &\implies
      (\fa{\ve' \in \QQ^+} \ve' >  \ve \implies \emptyset \vdash \condeq t u  {\ve'} \in V)\\
    &\implies
      \emptyset \vdash \condeq t u  \ve  \in V\\
    &\implies
    \inf\left\{\ve' \in \QQ^+ \sep \emptyset \vdash \condeq t u  {\ve'} \in V\right\} \leq \ve.
  \end{align*}

  We finally show $d[U[-]] = \mathrm{id}_{\CSEPMet(T_\Omega, X)}$.
  From the adjointness, $d[U[d]] \leq d$ holds for each $d \in \CSEPMet(T_\Omega, X)$.
  We can rewrite $d \leq d[U[d]]$ as follows:
  \begin{align*}
    d \leq d[U[d]]
    &\iff
      \fa{t,u \in T_\Omega} d(t,u) \leq d[U[d]](t,u)\\
    &\iff 
      \fa{t,u \in T_\Omega} d(t,u) \leq \inf \left\{ \ve \in \QQ^+ \sep \emptyset
      \vdash \condeq {t}{u}{\ve} \in U[d] \right\}\\
    &\iff 
      \fa{t,u \in T_\Omega,\ve \in \QQ^+} 
      \implies d(t,u) \leq \ve\\
    &\iff
      \{\emptyset \vdash \condeq {t}{u}{\ve} \in u[d] \}
      \subseteq \{ \emptyset \vdash \condeq {t}{u}{\ve} |  d(t,u) \leq \ve \}.
  \end{align*}
  Thanks to the minimality of $U[d]$, it suffices to have a QET $V \in \QET(\Omega,X)$ such that 
  \[
    \{\emptyset \vdash \condeq {t}{u}{\ve} \in V \} = \{ \emptyset \vdash \condeq {t}{u}{\ve} |  d(t,u) \leq \ve \}.
  \]
  Inspired from the definition of models of QET \mcite{bacci_et_al:LIPIcs.CALCO.2021.7},
  we define $V$ as follows: 
  \begin{align*}
    &{\Gamma \vdash \condeq {t}{u}{\ve}} \in V\\
    &\iff \fa{\sigma \colon X \to T_\Omega X}\left(\left(\fa{\condeq {t'}{u'}{\ve'} \in\Gamma}d(\sigma\kl(t'),\sigma\kl(u')) \leq \ve' \right) \implies d(\sigma\kl(t),\sigma\kl(u)) \leq \ve\right).
  \end{align*}
  By the substitutivity of $d$ and the definition of $V$, we obtain for all $t,u \in T_\Omega X$ and $\ve \in \QQ^+$,
  \[
  {\emptyset \vdash \condeq {t}{u}{\ve}} \in V
  \iff 
  (\fa{\sigma \colon X \to T_\Omega X} d(\sigma\kl(t),\sigma\kl(u)) \leq \ve)
  \iff 
   d(t,u) \leq \ve.
  \]
  We check that $V$ satisfies all rules of QET:
  
  \eqref{QET:axiom:Ref} Immediate from the reflexivity of $d$.
  
  \eqref{QET:axiom:Sym} Immediate from the symmetry of $d$.
  
  \eqref{QET:axiom:Tri} Immediate from the triangle-inequality of $d$.
  
  \eqref{QET:axiom:Max} Immediate from the transitivity of ordering $\leq$ and the monotonicity of $+$.
  
  \eqref{QET:axiom:Arch} Immediate from the Archimedean property and the completeness of $[0,\infty]$.
  
  \eqref{QET:axiom:Nonexp}
  Let $f \colon |I| \in \Omega$.
  We then take a term $t_f \in T_{\Omega} I$ corresponding to $f$.
  Let $t, s : I \rightarrow T_{\Omega} X$ be functions.
  We fix an arbitrary $\sigma \colon X \to T_{\Omega} X$.  
  Assume $d (\sigma\kl(t (i)), \sigma\kl(s (i))) \leq \ve$ for each $i \in I$.
  Then this asserts $\sup_{i \in I} d (\sigma\kl(t (i)), \sigma\kl(s (i))) \leq \ve$.
  From the congruence of $d$, we conclude 
  \[
    d (\sigma\kl (f (t(i)|i \in I)), \sigma\kl(f (s(i)|i \in I) ) )
    =
    d (\sigma\kl(t\kl (t_f)), \sigma\kl(s\kl (t_f)))
    \leq
    \sup_{i \in I} d (\sigma\kl(t (i)), \sigma\kl(s (i))) \leq \ve.
  \]
  
  \eqref{QET:axiom:Subst} Immediate by definition of $V$:
  \begin{align*}
    &{\Gamma \vdash \condeq {t}{u}{\ve}} \in V\\
    &\iff \fa{\sigma \colon X \to T_\Omega X}\left(\left(\fa{\condeq {t'}{u'}{\ve'} \in\Gamma}d(\sigma\kl(t'),\sigma\kl(u')) \leq \ve' \right) \implies d(\sigma\kl(t),\sigma\kl(u)) \leq \ve\right)\\
    &\implies 
\fa{\sigma' \colon X \to T_\Omega X}\fa{\sigma \colon X \to T_\Omega X}\left(
\begin{aligned}
&\left(\fa{\condeq {t'}{u'}{\ve'} \in\Gamma}d(\sigma\kl({\sigma'}\kl(t')),\sigma\kl({\sigma'}\kl(u'))) \leq \ve' \right) \\
&\qquad \implies d(\sigma\kl({\sigma'}\kl(t)),\sigma\kl({\sigma'}\kl(u))) \leq \ve
\end{aligned}
\right)\\
    &\implies 
\fa{\sigma' \colon X \to T_\Omega X}\fa{\sigma \colon X \to T_\Omega X}\left(
\begin{aligned}
&\left(\fa{\condeq {t''}{u''}{\ve'} \in\sigma'(\Gamma)}d(\sigma\kl(t''),\sigma\kl(u'')) \leq \ve' \right) \\
&\qquad \implies d(\sigma\kl({\sigma'}\kl(t)),\sigma\kl({\sigma'}\kl(u))) \leq \ve
\end{aligned}
\right)\\
	&\iff
	\fa{\sigma'\colon X \to T_\Omega X} {\sigma'}(\Gamma) \vdash \condeq {{\sigma'}\kl(t)}{{\sigma'}\kl(u)}{\ve} \in V.
  \end{align*}
  
  \eqref{QET:axiom:Cut} Immediate.
  
  \eqref{QET:axiom:Assumpt} Immediate.
\end{appendixproof}

In the proof of this theorem, we used the definition of models of QET
\mcite{bacci_et_al:LIPIcs.CALCO.2021.7}.  Intuitively, the right
adjoint $d[-]$ extracts the pseudometric on $T_\Omega X$ from a given
QET.  The left adjoint $U[-]$ constructs the least QET containing all
information of a given pseudometric on $T_\Omega X$.  The adjunction
\eqref{eq:qetadj} also implies that we can construct monads on the
category of extended metric spaces from CS-EPMets by Mardare et al.'s
metric term monad construction
\mcite{10.1145/2933575.2934518}. Overall adjunction \eqref{eq:qetdiag}
says that $X$-generated divergences can be axiomatized with QETs whose
variable set is $X$.

The range of $U[-]$ is a subset of $\mathbf{UQET}(\Omega,X)$ of \emph{unconditional QETs}
defined below (See also \cite[Section 3]{MardarePanangadenPlotkinLICS17}):
\[
\mathbf{UQET}(\Omega,X) 
\triangleq \left\{
V \in \QET (\Omega,X)
\sep
\exists S \subseteq \{ \emptyset \vdash\condeq {t}{u}{\ve} ~|~ t,u \in T_\Omega X, \ve \in \QQ^+ \}.~ V = \QETC S \Omega X
\right\}. 
\]
Unconditional QETs of type $\Omega$ over $X$ are \emph{equivalent to} $X$-generated divergence on $T_\Omega$: 
restricting QETs to unconditional QETs, the adjunction \eqref{eq:qetadj} becomes a pair of isomorphisms.
\begin{therm}\label{thm:adjunction:UQETandCSEPMet}
  $
  (\mathbf{UQET}(\Omega,X),\subseteq)
  \cong(\CSEPMet(T_\Omega, X),\preceq)
  \cong(\mathbf{DivEPMet}(T_\Omega,X),\preceq)
  $.
\end{therm}
\begin{appendixproof}
(Proof of Theorem \ref{thm:adjunction:UQETandCSEPMet})
Since the range of $U[-]$ is a subset of $\mathbf{UQET}(\Omega,X)$, 
we may define the following monotone restrictions of $U[-]$ and $d[-]$:
\begin{align*}
U'[-] &\colon (\CSEPMet(T_\Omega, X),\preceq) \to (\mathbf{UQET}(\Omega,X),\subseteq) && U'[d] \triangleq U[d] \quad (d \in  \CSEPMet(T_\Omega, X)),\\
d'[-] &\colon (\mathbf{UQET}(\Omega,X),\subseteq) \to (\CSEPMet(T_\Omega, X),\preceq) && d'[V] \triangleq d[V] \quad (V \in \mathbf{UQET}(\Omega,X)).
\end{align*}
By Theorem \ref{thm:adjunction:QETandCSEPMet}, we have $U'[-] \vdash d'[-]$ and $d'[U'[-]] = \id$.
We show $U'[d'[-]] = \mathrm{id}$.
Let $V \in \mathbf{UQET}(\Omega,X)$. There exists $S \subseteq \{ \emptyset \vdash\condeq {t}{u}{\ve} ~|~ t,u \in T_\Omega X, \ve \in \QQ^+ \}$ such that $V = \QETC S \Omega X$.
We check $U'[d'[V]] = V$.
By the adjunction $U'[-] \dashv d'[-]$, we have $U'[d'[V]] \subseteq V$ which is equivalent to $d'[V] \preceq d'[V]$.
It suffices to check $V \subseteq U'[d'[V]]$.
We have 
\begin{align*}
\lefteqn{\emptyset \vdash\condeq {t}{u}{\ve} \in S}\\
&\implies
\emptyset \vdash\condeq {t}{u}{\ve} \in V\\
&\implies d'[V](t,u)  = \inf\{ \ve' \in \QQ^+~|~ \emptyset  \vdash\condeq {t}{u}{\ve'} \in V \} \leq \ve
\end{align*}
From the monotonicity of the closure $\QETC{(-)}\Omega X$, we conclude
\[
V = \QETC S \Omega X \subseteq \QETC{\{ \emptyset \vdash\condeq {t}{u}{\ve} ~|~ d'[V](t,u) \leq \ve \}}  \Omega X = U'[d'[V]].
\]
Since $V \in \mathbf{UQET}(\Omega,X)$ is arbitrary, we have $U'[d'[-]] = \mathrm{id}$.
\end{appendixproof}
\section{Graded Strong Relational Liftings for Divergences}
\label{sec:codensity_lifting}

We have introduced the concept of divergence on monad for measuring
quantitative difference between two computational effects. To
integrate this concept with relational program logic, we employ a
semantic structure called {\em graded strong relational lifting} of
monad. It is introduced for the semantics of approximate probabilistic
relational Hoare logic for the verification of differential privacy
\mcite{DBLP:conf/popl/BartheKOB12}, then later used in various program
logics
\mcite{DBLP:conf/icalp/BartheO13,DBLP:conf/csfw/BartheGAHKS14,DBLP:conf/popl/BartheGAHRS15,DBLP:journals/entcs/Sato16,DBLP:conf/lics/SatoBGHK19}.
Independently, it is also introduced as a semantic structure for effect
system \mcite{DBLP:conf/popl/Katsumata14}. Liftings introduced in the
study of differential privacy are designed to satisfy a special
property called {\em fundamental property} \cite[Theorem
1]{DBLP:conf/popl/BartheKOB12}: when we supply the equivalence
relation to the lifting, it returns the adjacency relation of the
divergence. This special property is the key to express the
differential privacy of probabilistic programs in relational program
logics.

In this paper, we present a {\em general construction} of graded
strong relational liftings from divergences on monads.  First, we
recall its definition
\mcite{DBLP:conf/popl/Katsumata14,DBLP:conf/esop/GaboardiKOS21}.
\begin{definition}\label{def:graded_strong_lifting}
  Let $(\CC,T)$ be a CC-SM and $(M, \leq, 1, (\cdot))$ be a
  grading monoid. An {\em $M$-graded strong relational
    lifting} $\dot{T}$ of $T$ is a mapping
  $\dot{T} : M \times \Obj{\BRel\CC} \rightarrow \Obj{\BRel\CC}$
  satisfying the following conditions:
  \begin{enumerate}
  \item \label{con:inclusion} $\brelfn\CC(\dot{T} m X)=(T X_1, T X_2)$,
    and $m \leq m'$ implies $\dot{T} m X \leq \dot{T} m' X$.
  
  \item \label{con:unit}
    $(\eta_{X_1}, \eta_{X_2}) : X \darrow \dot{T} 1 (X)$.
  
  \item \label{con:kl} $(f_1, f_2) : X \darrow \dot{T} m (Y)$ implies
    $(f\kl_1, f\kl_2) : \dot{T} m' X \darrow \dot{T} (m \cdot m') Y$.
  
  \item \label{con:str}
    $(\theta_{X_1, Y_1}, \theta_{X_2, Y_2}) : X \dtimes \dot{T} m Y
    \darrow \dot{T} m (X \dtimes Y)$.
  \end{enumerate}
\end{definition}
\renewcommand{\TT}{\top\!\top}

Our interest is in the graded strong relational lifting that carries
the information of a given divergence $\asgn\in\mDiv TM\qQ E$. We
identify such liftings by the following {\em fundamental
  property}. First define the {\tmem{adjacency relation}} of
$\asgn$ by
\begin{equation}\label{eq:divergence:adjacency}
 \tilde{\asgn} (m, v) I \triangleq (T I, T I, \{ (c_1, c_2)~|~ \asg mI (c_1, c_2) \leq v \})
 \quad
 (m \in M, v \in \qQ ,I \in \CC).
\end{equation}
Note that $\tilde\asgn$ is monotone on $m$ and $v$.
\begin{definition}
  We say that an $M\times \qQ $-graded strong relational lifting
  $\dot T$ of $T$ satisfies the {\em fundamental property} with
  respect to $\asgn\in\mDiv TM\qQ E$ if the following holds:
  \begin{displaymath}
    \dot T (m, v) (E I) = \tilde\asgn (m, v) I\quad
    (m\in M,v\in \qQ,I\in\CC).
  \end{displaymath}
\end{definition}
\begin{therm}\label{th:fund}
  Let $(\CC,T)$ be a CC-SM, $(M, \leq, 1, (\cdot))$ be a grading
  monoid, $\qQ$ be a divergence domain and
  $\asgn = \{\asg m I \colon (U(TI))^2\to \qQ \}_{m\in M,I\in\CC}$ be
  a doubly-indexed family of $\qQ$-divergences satisfying monotonicity
  on $m$ (Definition \ref{def:div}).  Define the following mapping
  $\coden T\asgn:(M\times\qQ)\times\Obj{\brelc}\to\Obj{\brelc}$:
  \begin{align*}
    \coden T\asgn (m,v) X\triangleq
    (TX_1,TX_2,\{(c_1,c_2)~|~&\fa{I\in\CC, n\in M, w\in\qQ, (k_1, k_2) : X\dto \tilde\asgn (n, w) I} \\
    &\quad (k_1\kl \ap c_1, k_2\kl \ap c_2) \in \tilde\asgn (m \cdot n, v + w) I\})
  \end{align*}
  \begin{enumerate}
  \item \label{pp:grastrrellift} The mapping $\coden T\asgn$ is an
    $M \times \qQ$-graded strong relational lifting of $T$.
  \item \label{lem:eq} Let $E:\CC\to\brelc $ be a basic
    endorelation. Then
    \begin{align}
      \text{$\asgn$ is $E$-unit-reflexive}
      &\iff
        \fa{I \in \CC,(m,v)\in M\times \qQ }\coden T{\asgn}(m,v)(EI) \le
        \tilde\asgn (m,v) I
        \tag{S} \label{property1:unitref=>soundness}
      \\
      \text{$\asgn$ is $E$-composable}
      &\iff
        \fa{I \in \CC,(m,v)\in M\times \qQ }\coden T{\asgn}(m,v)(EI) \ge
        \tilde\asgn (m,v) I.
        \tag{C} \label{property2:comp=>completeness}
    \end{align}
  \end{enumerate}
\end{therm}
The construction of $\coden T\asgn$ is a graded extension of the
{\tmem{codensity lifting}}
\mcite{DBLP:journals/entcs/Sato16,DBLP:journals/lmcs/KatsumataSU18}.
The remainder of this section is the proof of Theorem \ref{th:fund}.
\begin{toappendix}
  \begin{lemma}\label{lem:grastrrellift-1}
    Let $(\CC,T)$ be a CC-SM and
    $ \asgn = \{\asg m I \colon (U(TI))^2\to \qQ \}_{m\in M,I\in\CC}$
    be a doubly-indexed family of $\qQ$-divergences satisfying
    monotonicity on $m$ (Definition \ref{def:div}).  Then
    $\coden T\asgn$ is an $M \times \qQ $-graded relational lifting of
    $T$ (satisfies conditions \ref{con:inclusion}--\ref{con:kl} of
    Definition \ref{def:graded_strong_lifting}).
  \end{lemma}
  \begin{proof}
    (Condition \ref{con:inclusion}) We first show that
    $(\mathrm{id}_{TX_1},\mathrm{id}_{TX_2}) \in \BRel \CC (\coden
    T\asgn (m,v) X,\coden T\asgn (n,w) X)$ for all $X$ whenever
    $m \leq n$ and $v \leq w$.  From the monotonicity of $\asgn$, for
    all $I \in \CC$, $c'_1, c'_2 \in U(TI)$, $n' \in M$,$w' \in \qQ$,
    we have
    \begin{align*}
      \lefteqn{(c'_1,c'_2) \in  \tilde\asgn(m\cdot n',v + w')I}\\
      &\iff \asg {m \cdot n'}I  (c'_1,c'_2) \leq v + w'\implies \asg {n \cdot n'} I  (c'_1,c'_2) \leq v + w'\implies \asg {n \cdot n'} I  (c'_1,c'_2) \leq w + w'\\
      &\iff (c'_1,c'_2) \in \tilde\asgn(n \cdot n',w + w')I.
    \end{align*}
    Therefore, for any $(c_1,c_2) \in \coden T\asgn (m,v) X$, we
    obtain $(c_1,c_2)\in\coden T\asgn (n,w) X$ as follows:
    \begin{align*}
      \lefteqn{(c_1,c_2)\in\coden T\asgn (m,v) X}\\
      &\iff
        \forall I\in\CC, n'\in M, w'\in\qQ, (k_1, k_2) :X\dto \tilde\asgn (n', w') I~.~(k_1\kl \ap c_1, k_2\kl \ap c_2) \in \tilde\asgn (m \cdot n', v + w') I\\
      &\implies
        \forall I\in\CC, n'\in M, w'\in\qQ, (k_1, k_2) : X\dto \tilde\asgn (n', w') I~.~(k_1\kl \ap c_1, k_2\kl \ap c_2) \in \tilde\asgn (n \cdot n', w + w') I\\
      &\iff (c_1,c_2)\in\coden T\asgn (n,w) X.
    \end{align*}

    (Condition \ref{con:unit}) We next show
    $(\eta_{X_1},\eta_{X_2}) :X\dto\coden T\asgn (1,0) X$.  From the
    definition of morphisms in $\BRel \CC$, for all $(x_1,x_2) \in X$,
    we have
    $(\eta_{X_1} \ap x_1, \eta_{X_2} \ap x_2) \in \coden T\asgn (1,0)
    X$ as follows:
    \begin{align*}
      \lefteqn{(x_1,x_2) \in X}\\
      &\implies \fa{ I\in\CC, n\in M, w\in\qQ, (k_1, k_2) :X\dto \tilde\asgn (n, w) I}
    	(k_1 \ap x_1, k_2 \ap x_2) \in \tilde\asgn (n, w) I
      \\
      & \iff \fa{I\in\CC, n\in M, w\in\qQ, (k_1, k_2) : X\dto \tilde\asgn (n, w) I}
        ((k_1\kl \circ \eta_{X_1}) \ap x_1), (k_2\kl \circ \eta_{X_2}) \ap x_2) \in \tilde\asgn (n, w) I
      \\
      & \iff \fa{I\in\CC, n\in M, w\in\qQ, (k_1, k_2) : X\dto \tilde\asgn (n, w) I}
    	(k_1\kl \ap (\eta_{X_1} \ap x_1), k_2\kl \ap (\eta_{X_2} \ap x_2)) \in \tilde\asgn (n, w) I
      \\
      &\iff(\eta_{X_1} \ap x_1, \eta_{X_2} \ap x_2) \in \coden T\asgn (1,0) X.
    \end{align*}

    (Condition \ref{con:kl}) Finally, we show that
    $(f_1\kl,f_2\kl) : \coden T\asgn (n,w) X\dto\coden T\asgn (n\cdot
    m,w +v) Y$ holds for any $(f_1,f_2) : X\dto\coden T\asgn (m,v) Y$
    and $(n,w)\in M \times \qQ $.  For all
    $(f_1,f_2) : X\dto\coden T\asgn (m,v) Y$, we have
    \begin{align*}
      \lefteqn{(f_1,f_2) :X\dto\coden T\asgn (m,v) Y}\\
      &\iff \fa{(x_1,x_2) \in X}(f_1 \ap x_1,f_2 \ap x_2) \in \coden T\asgn (m,v) Y\\
      &\iff\left(
        \begin{aligned}
          &\fa{(x_1,x_2) \in X, I\in\CC, n'\in M, w'\in\qQ, (k_1, k_2) :Y\dto \tilde\asgn (n', w') I}\\
          & \qquad (k_1\kl \ap (f_1 \ap x_1), k_2\kl \ap (f_2 \ap
          x_2)) \in \tilde\asgn (m \cdot n', v + w') I
        \end{aligned}
            \right)\\
      &\iff
        \left(
        \begin{aligned}
          &\fa{(x_1,x_2) \in X, I\in\CC, n'\in M, w'\in\qQ, (k_1, k_2) :Y\dto \tilde\asgn (n', w') I}\\
          &\qquad((k_1\kl \circ f_1) \ap x_1), (k_2\kl \circ f_2) \ap
          x_2) \in \tilde\asgn (m \cdot n', v + w') I
        \end{aligned}
            \right)\\
      &\iff
 	\left(
        \begin{aligned}
          &\fa{I\in\CC, n'\in M, w'\in\qQ, (k_1, k_2) :Y\dto \tilde\asgn (n', w') I}\\
          &\qquad (k_1\kl \circ f_1,k_2\kl \circ f_2) :
          X\dto\tilde\asgn (m \cdot n', v + w') I
        \end{aligned}
            \right).
            \tag{a}\label{codensity:lifting:conclusion1}
    \end{align*}
    For all $(c_1,c_2) \in \coden T\asgn (n,w) X$, we have
    \begin{align*}
      \lefteqn{(c_1,c_2)\in\coden T\asgn (n,w) X}\\
      &\iff
        \left(
        \begin{aligned}
          &\fa{I\in\CC, n'\in M, w'\in\qQ, (l_1, l_2) :X\dto \tilde\asgn (n', w') I}\\
          & \qquad (l_1\kl \ap c_1, l_2\kl \ap c_2) \in \tilde\asgn (n
          \cdot n', w + w') I
        \end{aligned}
            \right).
            \tag{b}\label{codensity:lifting:conclusion2}
    \end{align*}
  
    We here fix $(f_1,f_2) : X\dto\coden T\asgn (m,v) Y$.  We show
    $(f_1\kl,f_2\kl) \colon \coden T\asgn (n,w) X \dto T\asgn (n \cdot
    m,w + v) Y$.  We also fix $I\in\CC$, $n''\in M$, $w''\in\qQ$ and
    $(k_1, k_2) :Y\dto \tilde\asgn (n'', w'') I$.  From
    \eqref{codensity:lifting:conclusion1}, we obtain
    \[
      (k_1\kl \circ f_1,k_2\kl \circ f_2) : X\dto\tilde\asgn (m \cdot
      n'', v + w'') I.
    \]
    Therefore, by instantiating \eqref{codensity:lifting:conclusion2}
    with $(n',w') = (m \cdot n'',v + w'')$ and
    $(l_1,l_2) = (k_1\kl \circ f_1,k_2\kl \circ f_2)$, for all
    $(c_1,c_2) \in \coden T\asgn (n,w) X$, we have
    \[
      ((k_1\kl \circ f_1)\kl \ap c_1, (k_2\kl \circ f_2)\kl \ap c_2)
      \in \tilde\asgn (n \cdot m \cdot n'', w + v + w'') I.
    \]
    Since $(c_1,c_2) \in \coden T\asgn (n,w) X$, $I\in\CC$,
    $n''\in M$, $w''\in\qQ$ and
    $(k_1, k_2) :Y\dto \tilde\asgn (n'', w'') I$ are arbitrary, we
    conclude
    $(f_1\kl,f_2\kl) \colon \coden T\asgn (n,w) X \dto T\asgn (n \cdot
    m,w + v)$ as follows:
    \begin{align*}
      &\left(\begin{aligned}
          &\fa{(c_1,c_2) \in \coden T\asgn (n,w) X, I\in\CC, m''\in M, v''\in\qQ, (k_1, k_2) : Y \dto\tilde\asgn (m'', v'') I}\\
          &\qquad ((k_1\kl \circ f_1)\kl \ap c_1,(k_2\kl \circ f_2)\kl \ap c_2):X\dto \tilde\asgn (n \cdot m \cdot m'', w + v + v'') I
        \end{aligned}\right)\\
      &\iff
        \left(\begin{aligned}
            &\fa{(c_1,c_2) \in \coden T\asgn (n,w) X, I\in\CC, m''\in M, v''\in\qQ, (k_1, k_2) : Y \dto\tilde\asgn (m'', v'') I}\\
            &\qquad (k_1\kl \ap (f_1\kl \ap c_1),k_2\kl \ap (f_2\kl \ap c_2)):X\dto \tilde\asgn (n \cdot m \cdot m'', w + v + v'') I
          \end{aligned}\right)\\
      &\iff
        \fa{(c_1,c_2) \in \coden T\asgn (n,w) X}(f_1\kl \ap c_1,f_2\kl \ap c_2) \in \coden T\asgn (n \cdot m,w + v) Y\\
      &\iff (f_1\kl,f_2\kl) \colon  \coden T\asgn (n,w) X \dto  T\asgn (n \cdot m,w + v).
    \end{align*}
    This completes the proof.
  \end{proof}
\end{toappendix}
\begin{proof}
  (Proof of (1)) Proving conditions \ref{con:inclusion}-\ref{con:kl}
  of graded strong relational lifting (Definition
  \ref{def:graded_strong_lifting}) are routine generalization of
  \cite{DBLP:journals/lmcs/KatsumataSU18} and \cite[Section
  5]{DBLP:conf/popl/Katsumata14}; thus omitted here (see Lemma
  \ref{lem:grastrrellift-1} in appendix).

  However, condition \ref{con:str} of Definition
  \ref{def:graded_strong_lifting} needs a special attention because in
  general codensity lifting does not automatically lift strength. The
  current setting works because of our particular choice of the
  category of binary relations over $\CC$. We prove condition \ref{con:str} as follows.
  Since $f_i \ap j = f \ap \langle i, j \rangle$ for any $j\in UJ$ holds,
  we have the equivalence
  \begin{align*}
    (f, g) : X \dtimes Y\dto Z
    & \iff  \forall (x, x') \in X, (y, y') \in Y .
      (f\ap{\langle x, y\rangle}, g\ap{\langle x', y'\rangle}) \in Z\\
    & \iff  \forall (x, x') \in X, (y, y') \in Y .
      \left( \left( f_{x} \right)\ap{y}, \left( g_{x'} \right)\ap{y'} \right) \in Z\\
    & \iff  \forall (x, x') \in X.
      (f_{x}, g_{x'}) \colon  Y\dto Z.
  \end{align*}
  From this, condition \ref{con:kl} (law of graded Kleisli extension), and the equation 
  \eqref{eq:stun} on the strength of a CC-SM,
  we prove condition \ref{con:str} from condition \ref{con:unit} (unit law): 
  for all $m \in M$ and $v \in \qQ$, we have 
  \begin{align*}
    &  (\eta_{X_1\times Y_1}, \eta_{X_2\times Y_2}) \colon X\dtimes Y\dto \coden T\asgn (1,0) (X \dtimes Y)\\
    & \iff \fa{(x, x') \in X}
      ((\eta_{X_1\times Y_1})_x, (\eta_{X_2\times Y_2})_{x'}) \colon  Y\dto \coden T\asgn (1,0)  (X \dtimes Y)\\
    & \implies \fa{(x, x') \in X} (((\eta_{X_1\times Y_1})_x)\kl, ((\eta_{X_2\times Y_2})_{x'})\kl) :
      \coden T\asgn (m,v) Y\dto \coden T\asgn (m,v) (X \dtimes Y)\\
    &\iff
    \left(\begin{aligned}
    &\fa{(x, x') \in X, (c_1,c_2) \in \coden T\asgn (m,v) Y}\\
    &\qquad (((\eta_{X_1\times Y_1})_x)\kl\ap c_1, ((\eta_{X_2\times Y_2})_{x'})\kl \ap c_2) \in
     \coden T\asgn (m,v) (X \dtimes Y)
    \end{aligned}\right)\\
    &\iff
    \left(\begin{aligned}
    &\fa{(x, x') \in X, (c_1,c_2) \in \coden T\asgn (m,v) Y} \\
    &\qquad (\theta_{X_1, Y_1}\ap \langle x, c_1\rangle, \theta_{X_2, Y_2}\ap\langle x', c_2\rangle)\in  \coden T\asgn (m,v) (X \dtimes Y)
    \end{aligned}\right)\\
    & \iff  \fa{(x,x') \in X} ((\theta_{X_1, Y_1})_x, (\theta_{X_2, Y_2})_{x'})\colon \coden T\asgn (m,v) Y \dto \coden T\asgn (m,v) (X \dtimes Y)\\
    & \iff  (\theta_{X_1, Y_1}, \theta_{X_2, Y_2}) \colon  X \dtimes \coden T\asgn (m,v) Y\dto \coden T\asgn (m,v) (X\dtimes Y).
  \end{align*}

  (Proof of (2)-\eqref{property1:unitref=>soundness})
  We show the equivalence of $\asgn$ being $E$-unit-reflexive and the
  implication
  \begin{align}
    & \fa{I \in \CC,m\in M,v\in \qQ ,c,c'\in U(TI)} \nonumber \\
    &\qquad( \fa{J \in \CC,m'\in M,v'\in \qQ ,(k,l) : EI\dto\tilde\asgn (m',v') J}
      \asg{m\cdot m'}J(k\kl\ap c,l\kl\ap{c'})\le v+v') \label{eq:middle} \\
    &\qquad\qquad \implies \asg m I(c,c')\le v. \nonumber 
  \end{align}
  We suppose that the above implication holds.  We fix $I \in \CC$.
  Let $(i,j) \in EI$.  By instantiating the whole implication with
  $m=1,v=0,c=\eta_I\ap i,c'=\eta_I\ap j$, the middle part of
  \eqref{eq:middle} becomes
  \begin{displaymath}
    \fa{J \in \CC,m'\in M,v'\in \qQ ,(k,l) : EI\dto\tilde\asgn (m',v')J}
    \asg{m'}J(k\ap i,l\ap j)\le v',
  \end{displaymath}
  which is trivially true. Therefore we conclude
  $\asg m I(\eta_I\ap i,\eta_I\ap j)\le 0$ for any $(i,j) \in EI$,
  that is, $E$-unit reflexivity holds.
  
  Conversely, we suppose that $\asgn$ satisfies the unit-reflexivity.
  We take $I,m,v,c,c'$ of appropriate type and assume the middle part
  of \eqref{eq:middle}.  By instantiating it with
  $J=I,m'=1,v'=0,k=l=\eta_I$, we conclude $\asg m I(c,c')\le v$.
  
  (Proof of (2)-\eqref{property2:comp=>completeness}) We show the
  equivalence of $\asgn$ being $E$-composable and the implication
  $\fa{I \in \CC,m\in M,v\in \qQ}\tilde\asgn I(m,v) \le \coden T \asgn
  I(m,v)(EI)$ as follows:
  \begin{align*}
    \lefteqn{\fa{I \in \CC,m\in M,v\in \qQ}\tilde\asgn I(m,v)\le \coden T \asgn I(m,v)(EI)} \\
    &\iff
      \left(
      \begin{aligned}
        & \fa{I \in \CC,m\in M,v\in \qQ,c,c'\in U(TI)}\\
        & \qquad\asg m I(c,c')\le v \implies\\
        & \qquad\qquad \fa{J \in \CC,m'\in M,v'\in \qQ,(k,l) \colon  EI\dto\tilde\asgn(m',v') J}\\
        & \qquad\qquad \qquad (k\kl\ap c,l\kl\ap c')\in
        \tilde\asgn(m\cdot m',v+v') J
      \end{aligned}\right)\\
    &\iff
      \left(
      \begin{aligned}
        & \fa{I,J \in \CC,m\in M,v\in \qQ,c,c'\in U(TI),m'\in M,v'\in \qQ,k,l\in\CC(I,TJ)}\\
        & \qquad\asg m I(c,c')\le v\implies \\
        & \qquad\qquad(\fa{(i,j) \in EI}(k\ap i,l\ap j)\in \tilde\asgn
        (m',v')I)\implies \asg{m\cdot m'}I(k\kl\ap c,l\kl\ap c')\le
        v+v'
      \end{aligned}\right)\\
    &\iff
      \left(
      \begin{aligned}
        & \fa{I,J \in \CC,m\in M,v\in \qQ,c,c'\in U(TI),m'\in M,v'\in \qQ,k,l\in\CC(I,TJ)}\\
        & \qquad\asg m I(c,c')\le v\implies \\
        & \qquad\qquad\textstyle\sup_{(i,j) \in EI}\asg{m'} J(k\ap
        i,l\ap j)\le v'\implies \asg{m\cdot m'}J(k\kl\ap c,l\kl\ap
        {c'})\le v+v'
      \end{aligned}\right)\\
    &\iff
      \left(
      \begin{aligned}
        & \fa{I,J \in \CC,m\in M,c,c'\in U(TI),m'\in M,k,l\in\CC(I,TJ)}\\
        & \qquad\asg{m\cdot m'}I(k\kl\ap c,l\kl\ap{c'})\le \asg m
        I(c,c')+\textstyle\sup_{(i,j) \in EI}\asg{m'}I(k\ap i,l\ap j).
      \end{aligned}\right).
  \end{align*}
  The first two equivalences are obtained by expanding the definitions
  of $\brelc$, $\coden T \asgn {}$ and $\tilde\asgn$, the last two
  equivalences hold because $\qQ$ is a divergence domain.
\end{proof}

Combining the fundamental property and the strength of $\coden T \asgn$, we 
recover a strength law of divergences. 
\begin{proposition}\label{divergence_strength}
  Let $(\CC,T)$ be a CC-SM, $E:\CC\to\brelc$ be a basic endorelation,
  $(M, \leq, 1, (\cdot))$ be a grading monoid and $\qQ$ be a
  divergence domain.  Suppose also that
  $EI \dtimes EJ \subseteq E(I \times J)$ holds for all $I,J \in \CC$.
  Then each divergence $\asgn\in\mDiv TM\qQ E$ satisfies: for all
  $(x_1,x_2) \in EI$ and $c_1,c_2 \in U(TI)$,
\[
\asg m {I \times J} (\theta_{I,J} \ap \langle x_1,c_1 \rangle ,\theta_{I,J} \ap \langle x_2,c_2 \rangle)
\leq
\asg m J(c_1,c_2).
\]
\end{proposition}
\begin{appendixproof}
(Proof of Proposition \ref{divergence_strength})
By Theorem \ref{th:fund} and the assumption $\fa{I,J \in \CC}EI \dtimes EJ \subseteq E(I \times J)$,
we obtain for all $(x_1,x_2) \in EI$ and $c_1,c_2 \in U(TI)$,
\begin{align*}
&(\langle x_1,c_1 \rangle,\langle x_2,c_2 \rangle) \in EI \dtimes \tilde\asgn (m,v)J\\
&\iff (\langle x_1,c_1 \rangle,\langle x_2,c_2 \rangle) \in EI \dtimes \coden T \asgn (m,v)(EJ)\\
&\implies (\theta_{I,J} \ap \langle x_1,c_1 \rangle ,\theta_{I,J} \ap \langle x_2,c_2 \rangle) \in 
\coden T \asgn (m,v) (EI \dtimes EJ)\\
&\implies (\theta_{I,J} \ap \langle x_1,c_1 \rangle ,\theta_{I,J} \ap \langle x_2,c_2 \rangle) \in 
\coden T \asgn (m,v) E(I \times J)\\
&\iff (\theta_{I,J} \ap \langle x_1,c_1 \rangle ,\theta_{I,J} \ap \langle x_2,c_2 \rangle) \in 
\tilde\asgn (m,v)(I \times J).
\end{align*}
This completes the proof.
\end{appendixproof}

\subsection{Simplifying Codensity Liftings by $\Omega$-Generatedness of Divergences}
\label{sec:simplification}
We here show that for an $\Omega$-generated divergence $\asgn$, the
calculation of the codensity lifting $\coden T \asgn {} $ can be
simplified.  For an object $I\in\CC$, we define $\codeng T \asgn I$ by
\begin{align*}
&(c_1,c_2) \in \codeng T \asgn I (m,v) X\\
&\iff
\fa{n, w, (k_1, k_2) \colon X \dto \tilde\asgn (n, w) I} (k_1\kl \ap c_1, k_2\kl \ap c_2) \in \tilde\asgn (m \cdot n, v + w) I.
\end{align*}
The original calculation of $\coden T \asgn {} $ is a large
intersection $\coden T \asgn = \bigwedge_{I\in \CC} \codeng T \asgn I$
where $I$ runs over all $\CC$-objects, but if $\asgn$ is
$\Omega$-generated, the parameter $I$ can be fixed at $\Omega$.
\begin{proposition}
  For any $\Omega$-generated divergence $\asgn\in\mDiv TM\qQ E$, we
  have $\coden T \asgn {} = \codeng T \asgn \Omega$.
\end{proposition}
\begin{proof}
We show the equivalence $\coden T \asgn X = \codeng T \asgn \Omega X$ for each $X \in \brelc$.

($\supseteq$) Immediate from $\coden T \asgn = \bigwedge_{I\in \CC} \codeng T \asgn I$.

($\subseteq$) By the $\Omega$-generatedness of $\asgn$, we have for all $I \in \CC$ and $c'_1,c'_2 \in U(TI)$,
\[
(c'_1,c'_2) \in \tilde\asgn (m', v') I
\iff 
\fa{k \colon I \to T\Omega}(k\kl \ap c'_1,k\kl \ap c'_2) \in \tilde\asgn (m', v') \Omega
\]
Therefore, for any $(c_2,c_2) \in U(TX_1) \times U(TX_2)$, we have 
\begin{align*}
\lefteqn{(c_1,c_2) \in \codeng T \asgn \Omega X}\\
&\iff \fa{n\in M, w\in\qQ, (k_1, k_2) \colon  X\dto \tilde\asgn (n, w) \Omega} (k_1\kl \ap c_1, k_2\kl \ap c_2) \in \tilde\asgn (m \cdot n, v + w) \Omega\\
&\implies
\left(
\begin{aligned}
\fa{I\in\CC, n\in M, w\in\qQ, (l_1, l_2) \colon  X\dto \tilde\asgn (n, w) I,k \colon I \to T\Omega}\\
\qquad(k\kl\circ l_1\kl \ap c_1, k\kl\circ l_2\kl \ap c_2) \in \tilde\asgn (m \cdot n, v + w) \Omega
\end{aligned}\right)\\
&\iff 
\fa{I\in\CC, n\in M, w\in\qQ, (l_1, l_2) \colon  X\dto \tilde\asgn (n, w) I}
(l_1\kl \ap c_1,l_2\kl \ap c_2) \in \tilde\asgn (m \cdot n, v + w) I\\
& \iff (c_1,c_2) \in \coden T \asgn X.
\end{align*}
This completes the proof.
\end{proof}
For example, the generatedness of $\divndp$ shown in Section
\ref{sec:generatedness_of_divergence} implies that
$\coden {\giry} {\divndp} = \codeng {\giry} {\divndp} 2$ and
$\coden {\sgiry} {\divndp} = \codeng {\sgiry} {\divndp} 1$.  In fact,
the simplification $\codeng {\sgiry} {\divndp} 1$ is equal to the
$(\qRp)^2$-graded relational lifting $\sgiry^{\top\top}$ for DP given
in \cite[Section 2.2]{DBLP:journals/entcs/Sato16}, which is defined by,
for each $(X_1,X_2,R_X) \in \BRel\Meas$,
\begin{align*}
&\sgiry^{{\top\top}}(\varepsilon,\delta)(X_1,X_2,R_X)\\
&\triangleq
(\sgiry(X_1),\sgiry(X_2),\{(\nu_1,\nu_2) \mid \fa{A \in \Sigma_{X_1}, B \in \Sigma_{X_2}} R_X(A) \subseteq B \implies 
\nu_1(A) \leq \exp(\varepsilon) \nu_2(B) +\delta\}).
\end{align*}
For detail, see the proof of equalities (\dag) and (\ddag)
in the proof of \cite[Theorem 2.2(iv)]{DBLP:journals/entcs/Sato16}.

\subsection{Two Lifting Approaches: Codensity and Coupling}

\label{sec:comp}

We briefly compare two lifting approaches: graded codensity lifting
and coupling-based lifting employed in
\mcite{DBLP:conf/popl/BartheKOB12,DBLP:conf/icalp/BartheO13,DBLP:conf/csfw/BartheGAHKS14,DBLP:conf/popl/BartheGAHRS15,DBLP:conf/lics/SatoBGHK19}.
  
We compare the role of the unit-reflexivity and composability in the
codensity graded lifting and the coupling-based graded
lifting. Consider the CCC-SM $(\Set,\dist)$, where $\dist$ is the
probability distribution monad. Given an $\mEQ$-relative $M$-graded
$\qQ$-divergence $\asgn$ on $\dist$, the coupling-based graded lifting
is defined by
\begin{equation}
  \label{eq:coupling}
  \dot{\dist}^{\asgn} (m,v) X \triangleq \{ (\dist p_1 \ap \mu_1,\dist p_2 \ap \mu_2)~|~(\mu_1,\mu_2) \in (\dist R_X)^2, \asg m{R_X} (\mu_1,\mu_2) \leq v  \}  
\end{equation}
where $p_i\colon R_X\arrow X_i$ is the projection ($i=1,2$) from the binary
relation. The pair $(\mu_1,\mu_2)$ of probability distributions
collected in the right hand side of \eqref{eq:coupling} is called a
{\em coupling}.

The fundamental property $\dot\dist^\asgn(\mEQ I)=\tilde\asgn(m,v)I$
immediately follows from the definition of $\dot\dist^\asgn$, while
the composability and unit-reflexivity of $\asgn$ are used to make
$\dot\dist^{\asgn}$ a strong $M\times \qQ$-graded lifting
\cite[Proposition 9]{DBLP:conf/icalp/BartheO13}. On the other hand,
the codensity graded lifting $\coden D {\asgn}$ is always an
$M\times \qQ$-graded lifting; this does not rely on the
unit-reflexivity and composability of $\asgn$ (Proposition
\ref{pp:grastrrellift}). These properties are used to show that
$\coden D {\asgn}$ satisfies the fundamental property (Proposition \ref{lem:eq}).

The coupling-based lifting \eqref{eq:coupling} can be naturally
generalized to any $\Set$-monad $T$. However, at this moment we do not
know how to generalize the coupling technique to any CC-SM
$(\CC,T)$.  As the prior study by \cite{DBLP:conf/lics/SatoBGHK19}
pointed out, there is already a difficulty in extending it to the
CC-SM $(\Meas,\giry)$.

We illustrate how the problem arises. Let $X \in \BRel\Meas$. We would
like to pick two probability measures over $R_X$ as couplings, but
$R_X$ is merely a set. We therefore equip it with the subspace
$\sigma$-algebra of $X_1\times X_2$, and let $H_X$ be the derived
measurable space (hence $|H_X|=R_X$). We write
$p_i:H_X \rightarrow X_i$ for measurable projections ($i=1,2$).  We
then define a candidate $M\times \qQ$-graded lifting of $\giry$ by
\[
  \dot{\giry} (m,v) X = \{ (\giry p_1 \ap \mu_1,\giry p_2
  \ap \mu_2)~|~(\mu_1,\mu_2) \in (U\giry H_X)^2,
  \asg m{H_X} (\mu_1,\mu_2) \leq v \}.
\]
We now verify that $\dot{\giry}$ also lifts the Kleisli extension of
$\giry$, that is,
\begin{displaymath}
  (f, g)\colon Y\dto \dot{\giry} (m',v') X
  \implies
  (f\kl, g\kl) \colon  \dot{\giry} (m,v) Y\to \dot{\giry} (mm', v+v') X.
\end{displaymath}
Let $(f, g)\colon Y\dto \dot{\giry} (m',v') X$ be pair of measurable
functions. Then for each $(x,y) \in R_Y$, we have
$(f \ap  x, g \ap  y) \in R_{\dot \giry(m,v) X}$. Therefore there
exists $(\mu^{(x,y)}_1,\mu^{(x,y)}_2) \in (U\giry H_X)^2$ such that
$\giry\pi_1 \ap  \mu^{(x,y)}_1 = f \ap  x$ and
$\giry\pi_2 \ap  \mu^{(x,y)}_2 = g \ap  y$. Using \emph{the
  axiom of choice}, we turn this relationship into \emph{functions}
$\mu_1,\mu_2\colon R_Y \to U\giry H_X$.  If they {\em were} measurable
functions of type $H_Y \rightarrow \giry H_X$, then from the
composability of $\asgn$, we would have
$\asg{mm'}{H_X}(\mu\kl_1 \ap  w_1, \mu\kl_2 \ap  w_2) \leq v +
v'$ for $w_1, w_2 \in U \giry H_Y$ such that
$\asg{m'}{H_Y}(w_1, w_2) \leq v'$. This gives
$ (f\kl, g\kl) \colon \dot{\giry} (m,v) Y\dto \dot{\giry} (mm', v+v') X $.
However, in general, ensuring the measurability of $\mu_1,\mu_2$ is
not possible, especially because they are picked up by the axiom of
choice.  A solution given in \cite{DBLP:conf/lics/SatoBGHK19} is to
use the category $\Span\Meas$ of spans, that guarantees the existence
of good measurable functions $h_1, h_2 \colon H_Y \rightarrow \giry H_X$.


\newcommand{\D}{\Delta}
\newcommand{\G}{\Gamma}

\section{Approximate Computational Relational Logic}
\label{sec:acRL}


We introduce a program logic called {\em approximate computational
  relational logic} (acRL for short). It is a combination of Moggi's
computational metalanguage and a relational refinement type system
\mcite{DBLP:conf/popl/BartheGAHRS15}. The strong graded relational
lifting of a monad constructed from a divergence will be used to
relationally interpret monadic types, and gradings give upper
bounds of divergences between computational effects caused by two
programs.  acRL is similar to the relational refinement type system
HOARe2 \mcite{DBLP:conf/popl/BartheGAHRS15}, which is designed for
verifying differential privacy of probabilistic programs.  Compared to
HOARe2, acRL supports general monads and divergences, while it does
not support dependent products nor non-termination.

The relational logic acRL adopts the {\em extensional approach}
(cf. \cite[Chapter 9.2]{NielsonNielsonSemantics}):

\begin{itemize}
\item Relational assertions between contexts $\G$ and $\D$ are defined
  as binary relations
  between $U\sem\G$ and $U\sem\D$, or equivalently $\brelc$-objects $\phi$ such that
  $\brelfn\CC(\phi)=(\sem\G,\sem\D)$. Logical connectives and
  quantifications are defined as operations on such $\brelc$-objects.
  This is in contrast to the standard design of logic where assertions
  are defined by a BNF.
  
\item Let $\Gamma \vdash M : \tau$ and $\D \vdash N : \sigma$ be
  well-typed terms, $\phi$ be a relational assertion between
  $\G,\D$, and $\psi$ be an assertion between $\tau, \sigma$. The main
  concern of acRL is the statement
  ``$\forall (\gamma, \delta) \in \phi . (\sem M \ap \gamma, \sem N
  \ap \delta) \in \psi$'' (equivalently
  $(\sem M,\sem N):\phi\darrow\psi$). In this section we denote
  this statement by $\phi \vdash (M, M') : \psi$.
  
\item Inference rules of the logic consists of the {\em facts} about the
  statement $\phi \vdash (M, M') : \psi$. We remark that in the
  standard logic, proving these facts corresponds to the soundness of
  inference rules.
\end{itemize}

\subsection{Moggi's Computational Metalanguage}

\newcommand{\match}[5]{#1 \mathbin{\mathtt{with}}\iota_1(#2).#3\mathbin{\rule{.4pt}{1ex}}\iota_2(#4).#5}
\newcommand{\mret}[1]{\mathop{\mathtt{ret}}(#1)}
\newcommand{\mlet}[3]{\mathop{\mathtt{let}}{#1=#2}\mathbin{\mathtt{in}}#3}
\newcommand{\mop}[3]{\mathop{\mathtt{let}}{#1=#2}\mathbin{\mathtt{in}}#3}
\newcommand{\tT}{\attention{\mathtt{T}}}

\begin{figure}[t]
\caption{Syntax of Types and Raw Terms of the Computational Metalanguage}
  \begin{align*}
    \Typ(B)\ni\tau\mathbin{::=}& b\mid 1\mid \tau\times\tau\mid 0\mid \tau+\tau\mid \tau\Arrow\tau\mid \tT \tau
                       \quad(b\in B)\\ \\
    M\mathbin{::=} & x\mid o(M)\mid c(M)\mid ()\mid (M,M)\mid \pi_1(M)\mid \pi_2(M)\quad (o\in O_v, c\in O_e)\\
    \mid &\iota_1(M)\mid \iota_2(M)\mid \match M{x:\tau}M{x:\tau}M\\
    \mid &(\lam{x:\tau}M)\mid (MM)\mid \mret M\mid \mlet {x:\tau}MM
  \end{align*}
  \label{fig:syn}
\end{figure}

\subsubsection{Syntax of the Computational Metalanguage}
For the higher-order programming language, we adopt Moggi's {\em
  computational metalanguage} \mcite{moggicomputational}. It is an
extension of the simply typed lambda calculus with monadic types.  For
a set $B$, we define the set $\Typ(B)$ of types over $B$ by the first
BNF in Figure \ref{fig:syn}.  We then define the set $\Typ_1(B)$ of
first-order types to be the subset of $\Typ(B)$ consisting only of
$b,1,\times,+$.

We next introduce {\em computational signatures} for specifying
constants in the computational metalanguage. A computational signature
is a tuple $(B,\Sigma_v,\Sigma_e)$ where $B$ is a set of base types,
and $\Sigma_v$ and $\Sigma_e$ are functions whose range is
$\Typ_1(B)^2$.  The domains of $\Sigma_v,\Sigma_e$ are sets of {\em
  value operation} symbols and {\em effectful operation} symbols, and
are denoted by $O_v,O_e$, respectively. These functions assign input
and output types to these operations.

Fix a countably infinite set $V$ of variables. A context is a function
from a finite subset of $V$ to $\Typ(B)$; contexts are often denoted
by capital Greek letters $\G,\D$. For contexts $\G,\D$ such that
$\dom(\G)\cap\dom(\D)=\emptyset$, by $\G,\D$ we mean the join of $\G$ and
$\D$.

The set of raw terms is defined by the second BNF in Figure
\ref{fig:syn}. The type system of the computational metalanguage has
judgments of the form $\G\vdash M:\tau$ where $\G$ is a context, $M$
a raw term and $\tau$ a type. It adopts the standard rules for
products, coproducts, implications and monadic types; see
e.g. \cite{moggicomputational}. The typing rules for value operations
and effectful operations are given by
\begin{displaymath}
  \infer{\G\vdash o(M):b'}{
    o\in O_v &
    \Sigma_v(o)=(b,b') & \G\vdash M:b}
  \qquad
  \infer{\G\vdash c(M):\tT b'}{
    o\in O_e &
    \Sigma_e(c)=(b,b') & \G\vdash M:b}
\end{displaymath}

\newcommand{\subfst}[2]{\attention{\pi_1^{#1,#2}}}
\newcommand{\subsnd}[2]{\attention{\pi_2^{#1,#2}}}
\newcommand{\subproj}[3]{\attention{\pi_{#1}^{#2,#3}}}

A {\em simultaneous substitution} from $\G$ to $\G'$ is a function
$\theta$ from the set $\dom(\G')$ of variables to raw terms such that
the well-typedness $\G\vdash\theta(x):\G'(x)$ holds for each
$x\in\dom(\G')$. The application of $\theta$ to a term
$\G'\vdash M:\tau$ is denoted by $M\theta$, which has a typing
$\G\vdash M\theta:\tau$.  For disjoint contexts $\G_i$ ($i=1,2$), we
define the projection substitutions
$\G_1,\G_2\vdash\subproj i{\G_1}{\G_2}:\G_i$ by
$\subproj i{\G_1}{\G_2}(x)=x$.

\subsubsection{Categorical Semantics of the Computational Metalanguage}
\begin{figure}[t]
  \caption{Data for the Categorical Semantics of Metalanguage}
  \begin{enumerate}
  \item $(\CC,T)$ is a CCC-SM and $\CC$ has finite coproducts.
  \item $\sem b\in\CC$ for each $b\in B$
  \item $\sem o:\sem{b}\to\sem{b'}$ for each $o\in O_v$ such that
    $\Sigma_v(o)=(b,b')$
  \item $\sem c:\sem{b}\to T\sem{b'}$ for each $c\in O_e$ such that
    $\Sigma_e(c)=(b,b')$
  \end{enumerate}
  \label{fig:data}
\end{figure}
The interpretation of the computational metalanguage over a
computational signature $(B,\Sigma_v,\Sigma_e)$ is given by the data
specified by Figure \ref{fig:data}.

We first inductively extend the interpretation of base types to all
types using the bi-Cartesian closed structure and the monad.  Next,
for each context $\G$, we fix a product diagram
$(\sem\G,\{\pi_x:\sem\G\to\sem{\G(x)}\}_{x\in\dom(\G)})$; when
$\dom(\G)=\{x\}$, we assume that $\sem\G=\sem{\G(x)}$ with
$\pi_x=\id$. Lastly we interpret a typing derivation of
$\G\vdash M:\tau$ as a morphism $\sem M:\sem\G\to\sem\tau$ in the
standard way, using the interpretations of operations given in Figure
\ref{fig:data}.  We further extend this to the interpretation of each
simultaneous substitution $\G\vdash\theta:\G'$ as a morphisms
$\sem\theta:\sem\G\to\sem{\G'}$.

\subsection{Approximate Relational Computational Logic}

\newcommand{\asrt}[3]{{}^{#1}_{#2}\vdash #3}
\renewcommand{\u}{\attention{u}}
\renewcommand{\d}{\attention{d}}
\newcommand{\far}[2]{\forall^{#1}_{#2}~.~}
\newcommand{\exr}[2]{\exists^{#1}_{#2}~.~}

\subsubsection{Relational Logic in External Form}

A {\em relational assertion} $\phi$ between disjoint contexts $\G$ and
$\D$ is a binary relation between $U\sem\G$ and $U\sem\D$.  We denote
such a relational assertion by $\asrt\G\D\phi$, and identify it as a
$\brelc$-object $\phi$ such that $\brelfn\CC(\phi)=(\sem\G,\sem\D)$.
Similarly, a relational assertion between types $\tau$ and $\sigma$ is
defined to be a relational assertion $\asrt{\u:\tau}{\d:\sigma}\phi$;
here $\u,\d$ are reserved and fixed variables.

Relational assertions between contexts $\G$ and $\D$ carry a {\em
  boolean algebra structure} $\wedge,\vee,\neg$ given by the
set-intersection, set-union and set-complement (see the boolean
algebra $\brelc_{(\sem\G,\sem\D)}$ in Section \ref{sec:brel}). The
pseudo-complement $\phi\Arrow\psi$ is defined to be
$\neg \phi\vee\psi$. For $\asrt{\G,x:\tau}{\D,y:\sigma}\phi$, by
$\asrt\G\D{\far xy\phi}$ and $\asrt\G\D{\exr xy\phi}$ we mean the
relational assertions defined by the following equivalence:
\begin{align*}
  (\gamma,\delta)\in \far xy\phi
  \iff&
        \fa{\gamma'\in U\sem{\Gamma,x:\tau},\delta'\in U\sem{\Delta,y:\sigma}}
  \\&\quad\quad
  (\sem{\subfst\G{x:\tau}}\ap\gamma'=\gamma)\wedge
  (\sem{\subfst\D{y:\sigma}}\ap\delta'=\delta)
  \Arrow
  (\gamma',\delta')\in\phi\\
  (\gamma,\delta)\in \exr xy\phi
  \iff&
        \ex{\gamma'\in U\sem{\Gamma,x:\tau},\delta'\in U\sem{\Delta,y:\sigma}}
  \\&\quad\quad
  (\sem{\subfst\G{x:\tau}}\ap\gamma'=\gamma)\wedge
  (\sem{\subfst\D{y:\sigma}}\ap\delta'=\delta)\wedge
  (\gamma',\delta')\in\phi
\end{align*}

The boolean algebra structure and the above quantifier operations
allow us to interpret first-order logical formulas as relational
assertions; we omit its detail here.  In addition to these standard
logical connectives, we will use graded strong relational lifting
$\coden T\asgn$ to form relational assertions. That is, for any basic
endorelation $E:\CC\to\brelc$, grading monoid $M$, divergence domain
$\qQ$ and divergence $\asgn\in\mDiv TM\qQ E$, we obtain a relational
assertion
$\asrt{\u:\tT \tau}{\d:\tT \sigma}{\coden T{\asgn}(m,v)\phi}$ from any
$\asrt{\u:\tau}{\d:\sigma}\phi$, $m\in M$ and $v\in \qQ$.

\newcommand{\sub}[3]{#1[{#2};{#3}]}

For substitutions $\G\vdash\theta:\G',\D\vdash{\theta'}:\D'$ and an
assertion $\asrt{\G}{\D}\phi$, by
$\asrt{\G'}{\D'}{\sub\phi\theta{\theta'}}$ we mean the relational
assertion
$\{(\gamma,\delta)~|~(\sem\theta\ap\gamma,\sem{\theta'}\ap\delta)\in\phi\}$.
For disjoint context pairs $\G,\G'$ and $\D,\D'$ and relational
assertions $\asrt\G\D\phi$ and $\asrt{\G'}{\D'}\psi$, by the
juxtaposition $\asrt{\G,\G'}{\D,\D'}{\phi,\psi}$ we mean the
relational assertion
$\asrt{\G,\G'}{\D,\D'}{ \sub\phi{\subfst\G{\G'}}{\subfst\D{\D'}}\wedge
  \sub\psi{\subsnd\G{\G'}}{\subsnd\D{\D'}}}$.

\subsubsection{Inference Rules for acRL}

\newcommand{\judge}[4]{\attention{#1\vdash{}(#2,#3)\colon #4}} For
well-typed computational metalanguage terms $\G\vdash M:\tau$ and
$\D\vdash N:\sigma$, and relational assertions $\asrt\G\D\phi$ and
$\asrt{\u:\tau}{\d:\sigma}\psi$, by the judgment
\begin{displaymath}
  \judge\phi MN\psi 
\end{displaymath}
we mean the inclusion $\phi\subseteq \sub\psi{[M/\u]}{[N/\d]}$ of
binary relations.  This is equivalent to
$(\sem M, \sem N):\phi\dto\psi$.  We show basic facts about judgments
$\judge\phi MN\psi$.
\begin{proposition}\label{theorem:soundness}
  \begin{enumerate}
  \item \label{rule:extensional:semantic_equivalence}
    $\judge{\phi} MN{\psi}$ and $\sem M=\sem{M'}$ and
    $\sem N=\sem{N'}$ implies $\judge{\phi}{M'}{N'}{\psi}$.
  \item \label{rule:consequence}
  	 $\judge \phi MN\psi$ and $\phi'\subseteq \phi$ and
    $\psi\subseteq \psi'$ implies
    $\judge{\phi'}MN{\psi'}$.
  \item \label{rule:weakening_of_grading}
    $\judge \phi MN {\coden T\asgn(m,v)\psi}$ and $m\le n$ and
    $v\le w$ and $\psi\le \psi'$ \\implies
    $\judge \phi MN {\coden T\asgn(n,w){\psi'}}$.
  \item \label{rule:monad:return}
    $\judge \phi{ M}{N}\psi$ implies
    $\judge \phi{\mret M}{\mret{N}}{\coden T\asgn(1,0){\psi}}$.
  \item \label{rule:monad:bind}
    $\judge \phi MN{\coden T\asgn(m,v)\psi}$ and
    $\judge{\phi,\sub{\psi}{[x/\u]}{[x'/\d]}} {M'}{N'}{\coden T\asgn(n,w)\rho}$
    \\implies
    $\judge \phi{\mlet xM{M'}}{\mlet {x'}{N}{N'}}{\coden T\asgn(m\cdot n,v\cdot w)\rho}$.
  \end{enumerate}
\end{proposition}

We  next establish relational judgments on effectful operations. We
present a convenient way to establish such judgments using the
fundamental property of the graded relational lifting $\coden T\asgn$.
\begin{proposition}
  \label{lem:divergence:to:judgment}
  For any $c\in O_e$ such that $\Sigma_e(c)=(b,b')$, relational
  assertion $\asrt{\u:b}{\d:b}{\phi}$ and $m\in M$, putting
  $v= \sup \{\asg m {\sem{b'}} (\sem{c} \ap x, \sem{c} \ap y) ~|~(x,
  y) \in \phi \}$, we have
  $\judge{\phi}{c(\u)}{c(\d)}{\coden T\asgn(m,v)(E\sem{b'})}$.
\end{proposition}
\begin{proof}
  Take an arbitrary pair $(x, y) \in \phi$.
  We have $\asg m {\sem{b'}} (\sem{c} \ap x, \sem{c} \ap y) \leq v$ by definition of $v$.
  Thanks to the fundamental property of $\coden T \asgn$ (Theorem \ref{th:fund}),
  it is equivalent to $(\sem{c} \ap x, \sem{c} \ap y) \in \coden T\asgn (m,v) (E \sem{b'})$.
\end{proof}
\section{Case Study I: Higher-Order Probabilistic Programs}
\label{ex:prob}
\newcommand{\treal}{\mathtt{R}}
\newcommand{\cGauss}{\mathtt{norm}}
\newcommand{\cLap}{\mathtt{lap}}
\newcommand{\ctick}{\mathtt{tick}}
\newcommand{\cntick}{\mathtt{ntick}}
\newcommand{\dGauss}{\attention{\mathcal{N}}}
\newcommand{\dLap}{\attention{\mathrm{Lap}}}

We represent a higher-order probabilistic programming language with
sampling commands from continuous distributions as a computational
metalanguage. For now we assume that the language supports sampling
from Gaussian distribution and Laplace distribution. This
computational metalanguage is specified by the computational
signature:
\[
  \mathcal C= (\{\treal\}, \Sigma_v,
  \{\cGauss:(\treal\times\treal,\treal),~\cLap:(\treal\times\treal,\treal)\}),
\]
where $\Sigma_v$ is some chosen signature for value operations over
reals. We interpret this computational metalanguage by filling Figure
\ref{fig:data} as follows:
\begin{enumerate}
\item for the CCC-SM, we take $(\CC,T)=(\QBS,\probqbs)$ (see Section \ref{sec:qbs}),
\item for the interpretation $\sem\treal$ of $\treal$, we take the quasi-Borel space
  $\adjR\RR$ associated with the standard Borel space $\RR$,
\item the interpretation of value operations is given as expected (we
  omit it here); for example when $\Sigma_v$ contains the real number
  addition operator $+$ as type $(\treal \times \treal,\treal)$, its
  interpretation is the QBS morphism
  $\sem+(x,y)=x + y: \sem{\treal \times \treal}\arrow\sem\treal$,
\item for the interpretation of effectful operations, we put
  \begin{align*}
    \sem\cGauss(x,\sigma)
    &=[\mathrm{id},\dGauss(x,\sigma^2)]_{\sim_{\adjR\RR}},
    & \sem\cLap(x,\lambda)
    &=[\mathrm{id},\dLap(x,\lambda)]_{\sim_{\adjR\RR}}.
  \end{align*}
\end{enumerate}
Here, $\dGauss(x,\sigma^2) \in \giry\RR$ is the Gaussian
distribution with mean $x$ and variance $\sigma^2$.
$\dLap(x,\lambda) \in \giry\RR$ is the Laplacian
distribution with mean $x$ and variance $2\lambda^2$
\footnote{If $\sigma = 0$ (or $\lambda \leq 0$), $\dGauss(x,\sigma^2)$ (resp. $\dLap(x,\lambda)$) is not defined, thus we replace it by the Dirac distribution $\mathbf{d}_x$ at $x$ instead.}.  Every
probability (Borel-)measure $\mu \in \giry \RR$ on $\RR$ can be converted to the
probability measure
$[\mathrm{id},\mu]_{\sim_{\adjR\RR}} \in \probqbs \adjR \RR$ on the
quasi-Borel space $\adjR \RR$ (see Section \ref{subsection:divergence:QBS}).

\begin{toappendix}
\begin{lemma}
The mapping
\[
(x,\sigma) \mapsto 
\begin{cases}
\dGauss(x,\sigma^2) & \sigma^2 > 0\\
\mathbf{d}_{x} & \sigma = 0
\end{cases} 
\]
forms a measurable function of type $\RR \times \RR \to \giry\RR$.
\end{lemma}
\begin{proof}
We show that for all $A \in \Sigma_\RR$, 
the mapping $f_A(x,\sigma) = \dGauss(x,\sigma^2)(A)$ forms a measurable
function of type $\RR \times \RR_{\neq 0} \to [0,1]$ where $\RR_{\neq 0}$
is the subspace of $\RR$ whose underlying set is $\{r\in\RR | r \neq 0 \}$.
We have 
\[
\dGauss(x,\sigma^2)(A)
= \sum_{k \in \mathbb{Z}}\dGauss(x,\sigma^2)(A \cap [k,k+1])
= \sum_{k \in \mathbb{Z}} \int_{A \cap [k,k+1]} \frac{1}{\sqrt{2\pi\sigma^2}} \exp\left(-\frac{(x - r)^2}{\sigma^2}\right)dr
\]
The mapping $h(x,\sigma,r) = \frac{1}{\sqrt{2\pi\sigma^2}} \exp\left(-\frac{(x - r)^2}{\sigma^2}\right)$ forms a continuous function of type
$\RR \times \RR_{\neq 0} \times \RR \to \RR$, hence
it is uniformly continuous on the compact set $I_1 \times I_2 \times [k,k+1]$ where $I_1$ and $I_2$ are arbitrary closed intervals in $\RR$ and
$\RR_{\neq 0}$ respectively.
Then, for all $0 < \varepsilon $, there exists $0 < \delta$ such that 
$|h(x,\sigma,r)  - h(x',\sigma',r') |<\varepsilon$ holds wherever $|x - x'| + |\sigma - \sigma'| + |r -r '| < \delta $.
Hence, for all $0 < \varepsilon $, there is $0 < \delta$ such that
whenever $|x - x'| + |\sigma - \sigma'|< \delta$, 
\[
\left| \int_{A \cap [k,k+1]} h(x,\sigma,r)dr - \int_{A \cap [k,k+1]} h(x',\sigma',r)dr \right| \leq 
|\int_{[k,k+1]} |h(x,\sigma,r) - h(x',\sigma',r')| dr 
\leq \varepsilon.
\]
Since the closed intervals $I_1$ and $I_2$ are arbitrary, we conclude that 
the function $f_{A \cap [k,k+1]} \colon \RR \times \RR_{\neq 0} \to [0,1]$ is continuous, hence measurable.
Hence, the mapping $f_A = \sum_{k \in \mathbb{Z}} f_{A \cap [k,k+1]}$ is measurable.
Since $A$ is arbitrary and $f_A(x,\sigma^2) = \ev_A \circ \dGauss(x,\sigma^2)$, 
the mapping $g(x,\sigma^2) = \dGauss(x,\sigma^2)$ forms a measurable function of type $\RR \times \RR_{\neq 0} \to \giry\RR$.
The rest of proof is routine.
\end{proof}
\begin{corollary}
$\sem{\cGauss} \in \QBS(\adjR\RR\times \adjR\RR,\probqbs\adjR\RR)$.
\end{corollary}
\begin{lemma}[Measurability of $\sem{\cLap}$]
The mapping
\[
(x,\lambda) \mapsto 
\begin{cases}
\dLap(x,\lambda) & \lambda > 0\\
\mathbf{d}_{x} & \lambda \leq 0
\end{cases} 
\]
forms a measurable function of type $\RR \times \RR \to \giry\RR$.
\end{lemma}
\begin{proof}
We have, for all $A \in \Sigma_\RR$,
\[
\dLap(x,\lambda)(A)
=
\int_{A} \frac{1}{2\lambda} \exp\left(-\frac{|x - r|}{\lambda}\right)dr
\]
The density function $h(x,\lambda,r) = \frac{1}{2\lambda} \exp\left(-\frac{|x - r|}{\lambda}\right)$ is continuous 
function of type $\RR \times \RR_{0\leq} \times \RR \to \RR$ where $\RR_{0\leq} $ is the subspace of $\RR$ whose underlying set is 
$\{r\in \RR | 0 \leq r\}$.
The measurability of $\dLap(x,\lambda)$ is proved in the same way as $\dGauss(x,\sigma^2)$.
The rest of proof is routine.
\end{proof}
\begin{corollary}
$\sem{\cLap} \in \QBS(\adjR\RR\times\adjR\RR,\probqbs\adjR\RR)$.
\end{corollary}
\end{toappendix}

\subsection{A Relational Logic Verifying Differential Privacy}

To formulate differential privacy and its relaxations in the quasi-Borel setting, we
convert statistical divergences $\asgn$ on the Giry monad $\giry$
in Table \ref{tab:divdp} to $\mEQ{}$-relative divergences $\divnOpfun \asgn {L,l}$ on
the probability monad $\probqbs$ on $\QBS$ by the construction in Section \ref{subsection:divergence:QBS}.
Then, we construct
the graded relational lifting $\coden \probqbs {\divnOpfun \asgn {L,l}}$ by Theorem
\ref{th:fund}.  Using this, as an instantiation of acRL, we build a
relational logic reasoning about differential privacy and its
relaxations, supporting \emph{higher-order programs} and continuous random
samplings.  Basic proof rules can be given by Proposition
\ref{theorem:soundness}.

\newcommand{\adiff}{\mathrm{diff}}
\newcommand{\asucc}{\mathrm{succ}}

For effectful operations, we import basic proof rules on noise-adding
mechanisms given in prior studies
(\cite{DworkMcSherryNissimSmith2006,DworkRothTCS-042,MironovCSF17,BDRSSTOC18})
via Theorem \ref{thm:divergence_QBS:stdBorel} and Proposition
\ref{lem:divergence:to:judgment}. For example, consider the
$\mEQ$-relative $\qRp$-graded $\qRp$-divergence
$\asgn = \divnOpfun \divndp {L,l}$ on $\probqbs$.  Proposition
\ref{lem:divergence:to:judgment} with an effectful operation $c=\cLap$ and
a relational assertion (below we
identify global elements in $K\RR$ and real numbers)
\begin{displaymath}
  \asrt{\u:\treal\times\treal}{\d:\treal\times\treal}\phi = \{(\langle x,1/\varepsilon\rangle,\langle y,1/\varepsilon\rangle) ~|~ |x-y| \leq 1\},
\end{displaymath}
together with Theorem \ref{thm:divergence_QBS:stdBorel} and the prior
result \cite[Example 1]{DworkMcSherryNissimSmith2006} yields the
following judgment:
\begin{displaymath}
  \judge
  {\phi}
  {\cLap(\u)}
  {\cLap(\d)}
  {\coden \probqbs {\divnOpfun \divndp {L,l}} (0,\epsilon) (\mEQ{\adjR \RR})}.
\end{displaymath}
By letting $\adiff_r$ be the relational assertion
$\asrt{\u:\treal}{\d:\treal}{\{(x,y)~|~|x-y|\le r\}}$, the above
judgment is equivalent to:
\begin{equation}
  \judge{\adiff_1}{\cLap(\u,1/\epsilon)}
    {\cLap(\d,1/\epsilon))} {\coden \probqbs { \divnOpfun \divndp {L,l}}
      (0,\epsilon) (\mEQ{\adjR \RR})}.
      \label{eq:judgment:laplace}
\end{equation}
This rule corresponds to the rule [LapGen] of the program logic \apRHL{}+
(\cite{BGBHS2017POPL}) for differential privacy.  For another example,
by the reflexivity of $\divndp$, $\divnOpfun \divndp {L,l}$ is also
reflexive, hence we obtain the following judgments (below $\asucc_r$ is
the relational assertion
$\asrt{\u:\treal}{\d:\treal}{\{(x,y)~|~y=x+r\}}$):
\begin{align}
  & \judge{\asucc_1}{\cLap(\u,\lambda)} {\cLap(\d,\lambda)} {\coden
    \probqbs { \divnOpfun \divndp {L,l}} (0,0) (\asucc_1)}
    \label{eq:judgment:DP:laplace:slide} \\
  & \judge{\asucc_1}{\cGauss(\u,\sigma)} {\cGauss(\d,\sigma)} {\coden
    \probqbs { \divnOpfun \divndp {L,l}} (0,0) (\asucc_1)}
    \label{eq:judgment:DP:gauss:slide}.
\end{align}
The judgment \eqref{eq:judgment:DP:laplace:slide} correspond to
[LapNull] of \apRHL{}+.  Similarly, the following judgments about the
DP, R\'enyi-DP, zero-concentrated DP of the Gaussian mechanism can be
derived as
(\ref{eq:judgment:gauss:DP})--(\ref{eq:judgment:gauss:zCDP}).
\begin{align}
  & \judge{\adiff_1}{\cGauss(\u,\sigma)}{\cGauss(\d,\sigma)}{
    \coden \probqbs { \divnOpfun \divndp {L,l}} {(\epsilon,\delta)}
    {(\mEQ{\adjR \RR})}}
    \label{eq:judgment:gauss:DP} \\
  & \judge{\adiff_r}{\cGauss(\u,\sigma)}{\cGauss(\d,\sigma)}{
    \coden \probqbs {\divnOpfun  {\divrn\alpha{}} {L,l}} {(\alpha r^2 /
    2\sigma^2)} {(\mEQ{\adjR \RR})}}
    \label{eq:judgment:gauss:RDP} \\
  & \judge{\adiff_r}{\cGauss(\u,\sigma)}{\cGauss(\d,\sigma)}{
    \coden \probqbs {\divnOpfun  \divnzcdp {L,l}} {(0,r^2 / 2\sigma^2)}
    {(\mEQ{\adjR \RR})}}
    \label{eq:judgment:gauss:zCDP}
\end{align}
In \eqref{eq:judgment:gauss:DP} we require
$\sigma \geq \max ((1+\sqrt{3})/2,~\sqrt{2 \log
  (0.66/\delta)}/\epsilon)$.  The derivation is done via Proposition
\ref{lem:divergence:to:judgment}, Theorem
\ref{thm:divergence_QBS:stdBorel} and prior studies
\mcite{DBLP:journals/entcs/Sato16,MironovCSF17,BSTCC16}.

\section{Case Study II: Probabilistic Programs with Costs}
\label{sec:probcost}

We further extend the computational signature $\mathcal C$ in the
previous section with an effectful operation $\ctick$ such that
$\Sigma_e(\ctick)=(\treal,1)$.  The intention of $\ctick(r)$ is to
increase cost counter by $r$ during execution\footnote{To make
  examples simpler, we allow negative costs.}.  To interpret this
extended metalanguage, we fill Figure \ref{fig:data} as follows: 
\begin{enumerate}
\item for the CCC-SM, we take $(\CC,T) = (\QBS,\probqbs_{c})$ where
  $\probqbs_{c}\triangleq\probqbs(\adjR\RR\times-)$ is the monad for
  modeling probabilistic choice and cost counting (see Section
  \ref{sec:ccount}).
\item interpretation of $b \in B$ is the same as Section \ref{ex:prob},
\item interpretation of value operations is also the same as Section \ref{ex:prob},
\item for the interpretation of effectful
  operations, put
  \begin{align*}
    \sem\cGauss(x,\sigma) &= [(0,\mathrm{id}),\dGauss(x,\sigma^2)]_{\sim_{\adjR\RR \times\adjR\RR}},\\
    \sem\cLap(x,\lambda) &= [(0,\mathrm{id}),\dLap(x,\lambda)]_{\sim_{\adjR\RR \times\adjR\RR}},\\
    \sem\ctick (r) &= \eta^\probqbs_{\adjR\RR  \times \sem{1}} (r,\ast) = [\mathrm{const}(r,\ast), \mu]_{\sim_{\adjR\RR \times 1}}.
  \end{align*}
\end{enumerate}
We derive a closed term $\cntick \colon \treal \Arrow \treal \Arrow \tT 1$
for ticking with a cost sampled from Gaussian distribution:
\begin{displaymath}
  \cntick \triangleq (\lambda s. \lambda r.~\mlet x {\cGauss(r,s)}{\ctick(x)}).
\end{displaymath}
The term $\cntick~s~r$ adds cost counter by a random value
sampled from the Gaussian distribution $\cGauss(r,s^2)$.

\subsection{Relational Reasoning on Probabilistic Costs}
\label{sec:casestudy:probcost} 
We convert the total valuation distance
$\divntv\in\mDiv\giry1{\qRp}\mEQ$ to the divergence
$\asgn_{c} \triangleq \divnCostDiv {\divnOpfun \divntv  {L,l}} {K\RR} \in
\mDiv{\probqbs_{c}}1{\qRp}\mEQ$ on $\probqbs_c$ by Propositions
\ref{prop:divergence:monad_opfunctor_a} and
\ref{prop:divergence:cost:projection}.
We also prove basic facts
on effectful operations.  First, the following relational judgments on $\ctick$ can
be easily given:
\begin{align}
  \judge
  {\top}{\ctick(\u)}{\ctick(\d)}
  {\coden {T} {\asgnpc} (1) (\top)}
  \label{eq:costs:tick:top}
  \\
  \judge{\u=\d}
  {\ctick(\u)}{\ctick(\d)}
  {\coden {T} {\asgnpc} (0) (\top)}
  \nonumber
\end{align}
Remark that $\mEQ 1 = \top$ and $\sem{ \ctick(0)} = \sem{\mret \ast}$ holds.
Next, in the similar way as \eqref{eq:judgment:DP:laplace:slide}, by
the reflexivity of $\divntv$, we have the reflexivity of
$\divnOpfun \divntv  {L,l}$, and we obtain, for each real number constant $\sigma,\lambda$,
\begin{align}
  &\judge{\asucc_r}
    {\cGauss(\u,\sigma)}
    {\cGauss(\d,\sigma)}
    {\coden {T} {\asgnpc} (0) (\asucc_r)}
    \notag\\
  \label{eq:costs:laplacian_sliding}
  &\judge{\asucc_r}
    {\cLap(\u,\lambda)}
    {\cLap(\d,\lambda)}
    {\coden {T} {\asgnpc} (0) (\asucc_r)}
\end{align}
We also directly verify the following judgment on $\cntick$ using
Theorem \ref{thm:divergence_QBS:stdBorel} and Proposition
\ref{lem:divergence:to:judgment}:
\begin{align}
  \label{eq:costs:let_gaussian}
  \judge
  {\adiff_1}
  {\cntick~\sigma~\u}
  {\cntick~\sigma~\d}
  {\coden {T} {\asgnpc}{({\textstyle \Pr_{r \sim \dGauss(0,\sigma^2)}}[|r| < 0.5] )}{(\top)}}.
\end{align}

\subsubsection{An Example of Relational Reasoning}
We give examples of verification of difference (of distributions) of costs between
two runs of a probabilistic program whose output and cost depend on the input.
We consider the following program:
\begin{displaymath}
  M \triangleq \lambda r \colon \treal.~ \lambda t \colon \treal \to T1.~\mlet x {\cLap(r,5)} {\mlet {\_} {t(r)} {\mret {x - r}}}.
\end{displaymath}
It first samples a real number $x$ from the Laplacian distribution
centered at the input $r$, call the (possibly effectful) closure $t$
with $r$ and return $x-r$. Since the return type of $t$ is $T1$, it
can only probabilistically tick the counter. We show that the
following two judgments in acRL:
\begin{align}
  \label{proof:probcost:example1}
  &\judge {} {M~0~(\lambda x.\ctick(x))}{M~1~(\lambda x.\ctick(x))} {\coden {T} {\asgnpc}{(1)} {(\mEQ{\sem \treal}})}\tag{A},\\
  \label{proof:probcost:example2}
  &\judge {} {M~0~(\cntick(2))}{M~1~(\cntick(2))} {\coden {T} {\asgnpc}{(0.20)} {(\mEQ{\sem \treal}})}\tag{B}
\end{align}
In judgment \eqref{proof:probcost:example1}, we pass the tick
operation $t = \lambda x.\ctick(x)$ itself to $M~0$ and $M~1$.  By the
fundamental property of $\coden {T} {\asgnpc}$, the difference of
costs between two runs of $M~0~t$ and $M~1~t$ is $1$, because each of
these programs reports cost $0$ and $1$ deterministically.  In
contrast, in judgment \eqref{proof:probcost:example2}, we pass to
$M~0$ and $M~1$ the probabilistic tick function $t' = \cntick(2)$ that
ticks a real number sampled from the Gaussian distribution with
variance $2^2 = 4$. Therefore the cost reported by the runs of
programs $M~0~t'$ and $M~1~t'$ follow the Gaussian distributions
$\dGauss(0,4)$ and $\dGauss(1,4)$, whose difference by $\divntv$ is
bounded by $0.20$.

We first show \eqref{proof:probcost:example1}.  By
\eqref{eq:costs:tick:top} and \ref{rule:consequence} of Proposition \ref{theorem:soundness}, we
have,
\begin{equation}
  \label{eq:example:1}
  \judge {\asucc_1} {\ctick (\u)} {\ctick (\d)} {\coden {T} {\asgnpc}{(1)}{(\top)}}. 
\end{equation}
By \eqref{eq:example:1}, and \ref{rule:monad:return}, \ref{rule:monad:bind} of Proposition \ref{theorem:soundness}, we obtain,
\begin{align}
  \asucc_1 \vdash& (\mlet {\_} {\ctick (\u)} {\mret {\u}},\notag\\
  \label{eq:example:3}
                 &\quad\mlet {\_} {\ctick (\d)} {\mret {\d-1}}) \colon \coden {T} {\asgnpc}{(1)}{(\mEQ\sem\treal)}.
\end{align}
By \eqref{eq:costs:laplacian_sliding}, \eqref{eq:example:3}, and \ref{rule:extensional:semantic_equivalence} and \ref{rule:monad:bind} of Proposition \ref{theorem:soundness} again, we conclude
\eqref{proof:probcost:example1}.

To show (\ref{proof:probcost:example2}), it suffices to replace
(\ref{eq:example:1}) by the following judgment proved by
(\ref{eq:costs:let_gaussian}), the inequality $\Pr_{r \sim
  \dGauss(0,4)}[|r| < 0.5] \leq 0.20$ and \ref{rule:consequence} of Proposition
\ref{theorem:soundness}:
\[
  \judge {\asucc_1} {\cntick~2~\u} {\cntick~2~\d} {\coden {T} {\asgnpc}{(0.20)}{(\top)}}. 
\]
The rest of proof is the same as (\ref{proof:probcost:example1}).

\section{Related Work}

This work is based on the frameworks for verifying the differential
privacy of probabilistic programs using relational logic, summarized
in Table \ref{tbl:summary}. Composable divergences employed in these
frameworks include the one for differential privacy, plus its recent
relaxations, such as, R\'enyi DP, zero-concentrated DP, and
truncated-concentrated DP \mcite{BDRSSTOC18,BSTCC16,MironovCSF17}.

\begin{table}[htb]
  \caption{Approximate Probabilistic Relational Logic} {\small
    \begin{tabular}{ccccc}
      Work & Monad & Relation & Lifting Method & Supported divergences \\ \hline
      \cite{DBLP:conf/csfw/BartheGAHKS14,DBLP:conf/lics/BartheGGHS16,DBLP:conf/popl/BartheKOB12}&  Dist & $\BRel \Set$ & coupling & DP\\ \hline
      \cite{DBLP:conf/icalp/BartheO13}&  Dist & $\BRel \Set$ & coupling &$f$-divergences\\ \hline
      \cite{DBLP:conf/lics/SatoBGHK19} & Giry & $\Span \Meas$ & coupling (spans) & composable ones \\ \hline
      \cite{DBLP:journals/entcs/Sato16} & Giry  & $\BRel \Meas$ & codensity & DP \\ \hline
      This work &  Generic & $\brelc$ & codensity & composable ones
    \end{tabular}
  }
  \label{tbl:summary}
\end{table}
The key semantic structure in these frameworks is
\tmem{graded relational liftings} of the probability distribution monad.  
Barthe et al. gave a graded relational
lifting of the distribution monad based on the existence of two witnessing probability
distributions (called {\em coupling}) \mcite{DBLP:conf/popl/BartheKOB12}. Since then, {\em
  coupling-based} liftings have been refined and used in several works
\mcite{DBLP:conf/csfw/BartheGAHKS14,DBLP:conf/lics/BartheGGHS16,DBLP:conf/icalp/BartheO13,DBLP:conf/lics/SatoBGHK19}.
They can be systematically constructed from {\em composable}
divergences on the probability distribution monad
\mcite{DBLP:conf/icalp/BartheO13}. One advantage of coupling-based
liftings is that, to relate two probability distributions, it suffices
to exhibit a coupling; this is exploited in the mechanized
verification of differential privacy of programs
\mcite{DBLP:conf/cav/AlbarghouthiH18,DBLP:journals/pacmpl/AlbarghouthiH18}.
These coupling-based liftings, however, are developed upon discrete
probability distributions, and measure-theoretic probability
distributions, such as Gaussian or Cauchy distributions, were not
supported until the work \mcite{DBLP:conf/lics/SatoBGHK19}.

The relational Hoare logic supporting sampling from continuous
probability measures is given in the study by
\cite{DBLP:journals/entcs/Sato16}. In his work, the graded relational
lifting for $(\epsilon,\delta)$-DP is given in the style of {\em
  codensity lifting} \mcite{DBLP:journals/lmcs/KatsumataSU18}, which
does not rely on the existence of coupling. Yet, it has been an open
question \cite[Section VIII]{DBLP:conf/lics/SatoBGHK19} how to extend
his graded relational lifting to support various relaxations of
differential privacy. This paper answers to this question as Theorem
\ref{th:fund}.  Later, coupling-based liftings has also been extended
to support samplings from continuous probability measures
\mcite{DBLP:conf/lics/SatoBGHK19}. This extension is achieved by
redefining the concept of binary relations as {\em spans} of measurable
functions.  Comparison of these approaches is in the next section.

The verification of differential privacy in functional programming
languages has also been pursued
\mcite{DBLP:conf/icfp/ReedP10,DBLP:conf/popl/GaboardiHHNP13,
  DBLP:conf/popl/BartheGAHRS15,DBLP:conf/lics/AmorimGHK19}.
\cite{DBLP:conf/icfp/ReedP10} introduced a linear functional
programming language with a graded monadic type that supports
reasoning about $\epsilon$-differential privacy. Later, Gaboardi et
al. strengthen Reed-Pierce type system with dependent types
\mcite{DBLP:conf/popl/GaboardiHHNP13}. A category-theoretic account of
Reed and Pierce type system is given in
\cite{DBLP:conf/lics/AmorimGHK19}, where general
$(\epsilon,\delta)$-differential privacy is also supported. These
works basically regard types as metric spaces, allowing us to reason
about {\em sensitivity} of programs with respect to inputs. The
coupling-based lifting techniques are also employed in the relational
models of higher-order probabilistic programming language
\mcite{DBLP:conf/popl/BartheGAHRS15}.

The study \cite{DBLP:conf/lics/AmorimGHK19} gives a categorical
definition of composable divergences in a general framework called
{\em weakly closed refinements of symmetric monoidal closed
  categories} \cite[Definition 1]{DBLP:conf/lics/AmorimGHK19}. A
comparison is given in Section \ref{sec:int}.

\cite{10.1145/2933575.2934518} introduced a quantitative refinement of
algebraic theory called {\em quantitative equational theory}, and
studied variety theorem for quantitative algebras.
\cite{bacci_et_al:LIPIcs.CALCO.2021.7} discussed tensor products of
quantitative equational theories. QETs and divergences on monads share
the common interest of measuring quantitative differences between
computational effects. Divergences on monads are derived as a
generalization of the composability condition of statistical
divergences studied by \cite{DBLP:conf/icalp/BartheO13}. To make a
precise connection between these two concepts, in Section
\ref{sec:adj}, we have given an adjunction between QETs of type
$\Omega$ over $X$ and $X$-generated divergences on the free monad
$T_\Omega$. The adjunction cuts down to the isomorphism between
{\em unconditional} QETs of type $\Omega$ over $X$ and $X$-generated
divergences on $T_\Omega$.

The use of metric-like spaces in the semantics is seen in several
recent work. \cite{10.1145/3209108.3209149} studies quantitative
refinements of Abramsky's applicative bisimilarity for Reed-Pierce
type system. He introduces a monadic operational semantics of the
language and formalized quantitative applicative bisimilarity using
monad liftings to the category of quantale-valued relations.
\cite{bonchi_et_al:LIPIcs:2018:9555} also used metric-like spaces to
study bisimulations and up-to techniques in the category of
quantale-valued relations.
In this work our interest is relational program verification of
effectful programs, and it is carried out in the relational category
$\brelc$, rather than $\Div\qQ\CC$. The quantitative difference of
computational effects measured by a divergence $\asgn$ is represented
by the binary relation $\tilde\asgn$ graded by upper bounds
of distance.

\section{Future Work}

The framework for relational cost analysis given in \mcite{Radicek:2017:MRR:3177123.3158124}(extension of $\mathsf{RelCost}$ \mcite{Cicek:2017:RCA:3093333.3009858}) consists of the relational logic verifying the difference of costs between two programs
and the unary logic verifying the lower and upper bound of costs (i.e. cost intervals) in one program.
We expect that the relational logic can be reformulated by an instantiation of acRL with the divergence $\mathsf{NCI}$ on $\pwr(\NN \times -)$ (or its variant).
However to reformulate the unary logic, we want a \emph{unary version of divergence} on 
$\pwr(\NN \times -)$ for cost intervals.
To establish the connection between the unary logic and relational logic, we want a conversion from the unary version of divergence (for cost intervals) to $\mathsf{NCI}$
(for cost difference).

There might be many other examples and applications of divergences on monads.
In this paper, we mainly discussed examples of divergences with basic endorelations $\mT$ and $\mEQ$, but various other basic endorelations can be considered.
\section{Measurable Spaces and Quasi-Borel Spaces}
\label{sec:meas}
\label{sec:qbs}

\paragraph*{Measurable Spaces.}
For the treatment of continuous probability distributions, we employ
the category $\Meas$ of measurable spaces and measurable functions.
For a measurable space $I$ we write $|I|$ and $\Sigma_I$ for the
underlying set and $\sigma$-algebra of $I$ respectively.  The category
$\Meas$ is a (well-pointed) CC, and it has all small limits and small colimits
that are strictly preserved by the forgetful functor
$|{-}| \colon \Meas \to \Set$. It is naturally
isomorphic to the global element functor $\Meas(1,-)$.

\paragraph*{Standard Borel Spaces.}
A standard Borel space is a special measurable space
$(|\Omega|, \Sigma_\Omega)$ whose $\sigma$-algebra $\Sigma_\Omega$ is
the coarsest one containing the topology $\sigma_\Omega$ of a Polish
space $(|\Omega|, \sigma_\Omega)$.  In particular, the real line $\RR$
forms a standard Borel space.  In fact, a measurable space $\Omega$ is
standard Borel if and only if there are $\gamma \colon \Omega \to \RR$
and $\gamma' \colon \RR \to \Omega$ in $\Meas$
forming a section-retraction pair, that is,
$\gamma' \circ \gamma = \id_\Omega$.
For example, $[0,1]$, $[0,\infty]$, $\NN$, $\RR^k$ ($k \in \NN$) are
standard Borel.

\paragraph*{The Giry Monad.}
We recall the Giry monad $\giry$~(\cite{Giry1982}).  For every
measurable space $I$, $\giry I$ is the set $|\giry I|$ of all
probability measures over $I$ with the coarsest $\sigma$-algebra
induced by functions $\ev_A \colon |\giry I | \to [0,1]$  ($A \in \Sigma_X$)
defined by $\ev_A(\mu) = \mu(A)$.  The unit $\eta_I \colon I \to \giry I$ assigns to
each $x \in I$ the Dirac distribution $\mathbf{d}_x$ centered at $x$.
For every $f \colon I \to \giry J$, the Kleisli extension
$f\kl \colon \giry I \to \giry J$ is given by
$(f\kl(\mu))(A) = \int_x f(x)(A)~d\mu(x)$ for each $\mu \in \giry I$.  We also denote by
$\sgiry$ the subprobabilistic variant of $\giry$ (called sub-Giry
monad), where the underlying set $|\sgiry I|$ of $\sgiry I$ is relaxed
to the set of subprobaility measures over $I$.

The Giry monad $\giry$ (resp. the sub-Giry monad $\sgiry$) carries a
(commutative) strength
$\theta_{I,J} \colon I \times \giry J \to \giry (I \times J)$
over the CC $(\Meas,1,(\times))$. It computes the product of
measures ($(x, \mu) \mapsto \mathbf{d}_x \otimes \mu$).  Therefore
$(\Meas,\giry)$ and $(\Meas,\sgiry)$ are (well-pointed) CC-SMs.

\paragraph*{Quasi-Borel Spaces.}
The category $\Meas$ is not suitable for the semantics of
\emph{higher-order} programming languages since it is not Cartesian
closed~(\cite{aumann1961}).  For the treatment of higher-order
probabilistic programs with continuous distributions, we employ the
Cartesian closed category $\QBS$ of quasi-Borel spaces and morphisms
between them, together with the probability monad $\probqbs$ on
$\QBS$~(\cite{HeunenKSY17}).  A quasi-Borel space is a pair
$I = (|I|,M_I)$ of a set $|I|$ and a subset $M_I$ of the function
space $\RR \Rightarrow |I|$ satisfying
\begin{enumerate}
\item for $\alpha \in M_I$ and a measurable function $ f \colon \RR \to \RR$,
$\alpha \circ f \in M_I$.
\item  for any $x \in I$, $(\lambda r \in \RR.x) \in M_I$.
\item  for all $P \colon \RR \to \NN$ and a family $\{ \alpha_i \}_{ i \in \NN}$ of functions $\alpha_i \in M_I$,  $(\lambda r \in \RR. \alpha_{P(r)}(r) ) \in M_I$.
\end{enumerate}
A morphism $f \colon (|I|,M_I) \to (|J|,M_J)$ is a function $f \colon |I| \to |J|$
such that $f \circ \alpha \in M_J$ holds for all $\alpha \in M_I$.
The category $\QBS$ is a (well-pointed) CCC, and has all countable
products and coproducts that are strictly preserved by the forgetful functor $|{-}| \colon \QBS \to \Set$.
It is naturally isomorphic to the global element functor $\QBS(1,-)$.

\paragraph*{Connection to Measurable Spaces: an Adjunction}
We can convert measurable spaces and quasi-Borel spaces using an
adjunction $\adjL \dashv \adjR \colon \Meas \to \QBS$. They
are given by
\begin{align*}
  \adjL I &\triangleq (|I|, \{ U \subseteq |I| ~|~ \forall \alpha \in M_X. \alpha^{-1}(I) \in \Sigma_{\RR} \}) & \adjL f &\triangleq f\\
  \adjR I &\triangleq (|I|, \Meas(\RR,I)) & \adjR f &\triangleq f
\end{align*}
For any standard Borel space $\Omega \in \Meas$, we have
$\adjL\adjR\Omega = \Omega$.  The right adjoint $\adjR$ is
full-faithful when restricted to the standard Borel
spaces~\cite[Proposition 15-(2)]{HeunenKSY17}.  The right adjoint $K$
preserves countable coproducts and function spaces (if exists) of
standard Borel spaces~\cite[Proposition 19]{HeunenKSY17}.

\paragraph*{Probability Measures and the Probability Monad.}
A probability measure on a quasi-Borel space $I$ is a pair
$(\alpha,\mu) \in M_I \times \giry \RR$.  We introduce an equivalence
relation $\sim_I$ over probability measures on $I$  by
\begin{displaymath}
  (\alpha,\mu) \sim_I (\beta,\nu) \iff \mu(\alpha^{-1}(-)) =
  \nu(\beta^{-1}(-)).
\end{displaymath}
Using this, we introduce a probability monad $\probqbs$ on $\QBS$ as
follows:
\begin{itemize}
\item On objects, we define $P:\Obj\QBS\to\Obj\QBS$ by
  \begin{align*}
    |\probqbs(I)| &\triangleq (M_I \times \giry \RR) / \sim_I,
    & M_{\probqbs(I)} \triangleq \{ \lambda r. [(\alpha,g(r))]_{\sim_I} ~|~ \alpha
      \in M_I, g \in \Meas(\RR,\giry \RR)\}.
  \end{align*}
\item The unit is defined by $\eta_I (x) \triangleq [\lambda r.x,\mu]_{\sim_I}$
  for an arbitrary $\mu \in \giry \RR$.
\item The Kleisli extension of $f \colon I \to \probqbs(J)$ is defined
  by $f^\sharp [\alpha,\mu]_{\sim_I} \triangleq [\beta, g^\sharp \mu]$ where
  there are $\beta \in M_J$ and $g \in \Meas(\RR,\giry\RR)$ satisfying
  $f \circ \alpha = \lambda r \in \RR. [\beta,g(r)]_{\sim_J}$ by
  definition of $M_{\probqbs(J)}$.
\end{itemize}
The monad $\probqbs$ is (commutative) strong with respect to the CCC
$(\QBS,1,(\times))$.

\section*{Acknowledgments}

Tetsuya Sato carried out this research under the support by JST ERATO
HASUO Metamathematics for Systems Design Project (No. JPMJER1603) and
JSPS KAKENHI Grant Number 20K19775, Japan. Shin-ya Katsumata carried
out this research under the support by JST ERATO HASUO Metamathematics
for Systems Design Project (No. JPMJER1603) and JSPS KAKENHI Grant
Number 18H03204, Japan. The authors are grateful to Ichiro Hasuo
providing the opportunity of collaborating in that project.  The
authors are grateful to Satoshi Kura, Justin Hsu, Marco Gaboardi,
Borja Balle and Gilles Barthe for fruitful discussions.

\renewcommand{\appendixprelim}{}
\bibliographystyle{plain}
\bibliography{mine}

\newpage
\appendix

\end{document}